\documentclass{jfm}
\usepackage{amsmath,amssymb}
\usepackage{graphicx,psfrag}
\usepackage{mathrsfs}
\usepackage{bbm}
\usepackage{xcolor}
\usepackage{color}
\usepackage{stmaryrd}
\usepackage{subfigure}
\usepackage{upgreek}
\usepackage{phonetic}
\usepackage{natbib}
\usepackage{comment}
\usepackage{tikz}
\usetikzlibrary{arrows}
\tikzstyle{block}=[draw opacity=0.7,line width=1.4cm]

\newtheorem{theorem}{Theorem}

\newtheorem{definition}{Definition}

\newtheorem{remark}{Remark}
\newtheorem{corollary}{Corollary}

\def\proof{\par{\vspace{0.5em}\it Proof.} \ignorespaces}

\def \R{{\mathbb{R}}}
\def \N{{\mathbb{N}}}

\newcommand{\myparallel}{{\mkern3mu\vphantom{\perp}\vrule depth 0pt\mkern2mu\vrule depth 0pt\mkern3mu}}

\DeclareFontFamily{OT1}{pzc}{}
\DeclareFontShape{OT1}{pzc}{m}{it}{<-> s * [1.2200] pzcmi7t}{}
\DeclareMathAlphabet{\mathpzc}{OT1}{pzc}{m}{it}

\title[Geometric formulation of the Cauchy invariants in flat and curved spaces]
      {Geometric formulation of the Cauchy invariants for incompressible Euler flow in flat and curved spaces}

\author[Nicolas Besse \& Uriel Frisch]
       { Nicolas Besse \footnotemark[1]  \footnotemark[2] , Uriel Frisch \footnotemark[1]} 

\begin{document}

\maketitle

\renewcommand{\thefootnote}{\fnsymbol{footnote}}

\footnotetext[1]{ 
 Universit\'e C\^ote d'Azur, Observatoire de la C\^ote d'Azur, Nice, France}
  
\footnotetext[2]{
{\tt Nicolas.Besse@oca.eu}}

\renewcommand{\thefootnote}{\arabic{footnote}}

\begin{abstract}

Cauchy invariants are now viewed as a powerful tool for investigating
the Lagrangian structure of three-dimensional (3D) ideal flow
\citep{FZ14,PZF16}. Looking
at such invariants with the modern tools of differential geometry and
of geodesic flow on the space SDiff of volume-preserving
transformations \citep{Arn66}, all manners of generalisations are
here derived. The Cauchy
invariants equation and the Cauchy formula, relating the vorticity and
the Jacobian of the Lagrangian map, are shown to be  two expressions of this
Lie-advection invariance, which are duals of each
other
(specifically, Hodge dual). Actually, this is
shown to be an instance of a general result which holds
for flow both in flat (Euclidean) space and in a curved Riemannian space:
any Lie-advection invariant $p$-form which is exact (i.e. is a
differential of a $(p-1)$-form) has an associated Cauchy invariants
equation and a Cauchy formula.
This constitutes a new fondamental result in linear transport theory,
providing a Lagrangian formulation of Lie advection for some classes
of differential forms. The result has a broad applicability: examples
include the magnetohydrodynamics (MHD) equations and various extensions thereof, discussed by
\citet{LMM16} and include also the equations of  \citet{Tao16}, Euler 
equations
with modified Biot--Savart law, displaying finite-time blow up. Our
main result is also used for new derivations --- and several new results
--- concerning local helicity-type invariants for fluids and MHD
flow in flat or curved spaces of arbitrary dimension.
\end{abstract}

\begin{flushleft}
\end{flushleft}


\section{Introduction}
\label{sec:intro}
About half a century before the discovery of the integral invariant of
velocity circulation, \citet{Cau27} found a \textit{local}
form of this conservation law, now called the \textit{Cauchy invariants},
which constitutes the central topic of the present paper. The somewhat
tortuous history of the Cauchy invariants has been documented by
\citet{FV14}.  Starting in the sixties, the Cauchy invariants were rediscovered
by application of the Noether theorem, which relates continuous invariance
groups and conservation laws; at first this was done without attribution to
Cauchy \citep{Eck60, Sal88, PM96}. But, eventually, near the end of the 20th
century, proper attribution was made \citep{AZY96, ZK97}.

In recent years, there has been growing interest in Cauchy invariants because
of the development of new applications, such as analyticity in time of
fluid-particle trajectories
\citep[][see also \citet{CEA15a, CEA15b}]{FZ14, ZF14, RVF15, BF17},
and the design of very accurate semi-Lagrangian numerical schemes
for fluid flow \citep{PZF16}.

Our geometric approach to the Cauchy invariants will allow us to achieve two
goals. On the one hand to unify various vorticity results: as we shall see,
the 3D Cauchy invariants equation, as originally formulated, the Cauchy formula relating
current and initial vorticity and Helmholtz's result on conservation of
vorticity flux may all be viewed as expressing the geometrical conservation 
law of vorticity. On the other hand it will
allow us to extend the
invariants into various directions: higher-order Cauchy invariants,
magnetohydrodynamics (MHD),
flow in Euclidean spaces of any dimension, and flow in curved spaces.  Of course,
flows of practical interest are not restricted to flat space \citep{MR99,KS97}. Curved spaces
appear not only in General Relativistic fluid dynamics \citep{Wei72, CB08}, but also
for flows in the atmosphere and oceans of planets \citep{SAM68},
for studies of the energy inverse cascade on 
negatively curved spaces
\citep[][see also \citet{KM12, AK98}]{FG14}, and also for
flows on curved biological membranes \citep{Sei91, RN11, LR15}.
Moreover,  recently, \citet{GV16} have used differential geometry tools such as pullback transport
to extend the generalised Lagrangian theory (GLM) of \citet{AMI78} to curved spaces.  
Hereafter, the notation 1D, 2D and 3D will refer to the usual 
one, two and three-dimensional flat (Euclidean) spaces.

For carrying out this program our key tools will be differential
geometry and, to a lesser extent, variational methods.

In differential geometry we shall make use of Lie's
generalisation of advection (transport). The Lie advection  of a scalar quantity is just
its invariance along fluid-particle trajectories. But, here, we consider more
general objects, such as vectors, $p$-forms and tensors. For example, for our purpose, it
is more convenient to consider the vorticity as a $2$-form (roughly an
antisymmetric second-order tensor), rather than as a vector field.  These
non-scalar objects live in vector spaces spanned by some basis, and Lie
advection requires taking into account the distortion of the underlying vector
space structure, which moves and deforms with the flow.  The generalisation of
the particular (material) derivative to tensors is thus the Lie derivative.

As to variational (least-action) methods, an important advantage is
that they are applicable with very little change to both flat and
curved spaces, provided one uses Arnold's formulation of ideal
incompressible fluid flow as geodesics on the space SDiff of volume
preserving smooth maps \citep{Arn66,AK98}. Since the  1950s, to derive or
rediscover the Cauchy invariants equation, a frequently used approach has
been via  Noether's theorem with the
appropriate continuous invariance group, namely the relabelling
invariance in Lagrangian coordinates. The latter can be viewed as a
continuous counterpart of the permutation of Lagrangian labels if the
fluids were constituted of a finite number of fluid elements; note
that continuous volume-preserving transformations may be approached by
such permutations \citep{Lax71, Shn85}.

The outline of the paper is as follows.
Sec.~\ref{sec:ARACLI} is about Lie derivatives, an extension to flow
on manifolds  of what
is called in fluid mechanics the Lagrangian or material derivative.
We then prove a very general result about
Lie-advection invariance for exact $p$-forms of order $p\geq 2$, namely that 
there are generalised Cauchy invariants equations (see Theorem~\ref{th:GCI}), a very
concrete Lagrangian expression  of Lie-advection invariance. 
This theorem is applicable both to linear transport theory, when the advecting
velocity is prescribed, and to nonlinear (or selfconsistent) transport, when
the Lie-advected quantity (e.g. the vorticity) is coupled back to the
velocity (e.g. through the Biot--Savart law). Contrary to most modern
derivations of Cauchy invariants, our proof  does not make use of Noether's
theorem. Actually, for the case of linear transport, there may not even be 
a suitable continuous symmetry group to ensure the existence of a Noether
theorem. Sec.~\ref{rk:DGCI} is about generalised Cauchy formulas, which are
actually the Hodge duals of the Cauchy invariants equations. Theorem~\ref{th:GCI} has
a broad applicability, as exemplified in the subsequent sections.
Sec.~\ref{ex:B} is about ideal incompressible MHD. Sec.~\ref{ss:ACF} is about
adiabatic and barotropic compressible fluids. Sec.\ref{ss:BCMHD} is about
barotropic ideal compressible MHD. Sec.~\ref{ss:ECMHD} is about extended ideal
compressible MHD. Finally, Sec.~\ref{ss:TaoEuler} is about Tao's recent
modification of the 3D Euler equation allowing finite-time blowup and its
geometric interpretation.

Then, in Secs.~\ref{sec:VIFIH}, we turn to various applications in ordinary
hydrodynamics. Problems of helicities for hydrodynamics and MHD and their
little-studied local variants are presented in Sec.~\ref{sec:AHMHD}.
Concluding remarks and a discussion of various open problems are found in
Sec.~\ref{sec:conc}.  There are two sets of
Appendices. Appendix~\ref{sec:GVDIE} gives proofs of certain technical
questions, not found in the existing literature.  Appendix~\ref{sec:CICF},
``Differential Geometry in a Nutshell,'' has a different purpose: it is meant
to provide an interface between the fluid mechanics reader and the sometimes
rather difficult literature on differential geometry. Specifically, whenever
we use a concept from differential geometry that the reader may not be
familiar with, e.g., a ``pullback,'' we give a soft definition in simple
language in the body of the text and we refer to a suitable subsection of
Appendix~\ref{sec:CICF}. There, the reader will find more precise definitions
and, whenever possible, short proofs of key results, together with precise
references (including sections or page numbers) to what, we believe, is
particularly readable specialised literature on the topic.

\section{A general result about Lie advection and Cauchy invariants}
\label{sec:ARACLI}

\subsection{A few words about differential geometry}
\label{sec:FWDG}
In the present paper we prefer not starting with a barrage of
mathematical definitions and we rather appeal to the reader's
intuition. For those hungry of precise definitions, more
elaborate -- but
still quite elementary -- material and guides to the literature are
found in Appendix~\ref{sec:CICF} and its various subsections. 
For reasons explained in the Introduction, we feel that 
it is essential not to restrict our discussion to flat spaces. Otherwise we
would have used a `half-way house' approach where all the differential
geometry is expressed in the standard language of vector operations, as
done, for example in the paper of \citet{Lars96}.

The concept of a differentiable manifold $M$ generalises to an
arbitrary dimension $d$ that of a curve or a surface embedded in the
3D Euclidean space $\R^3$. To achieve this in an intrinsic fashion without
directly using Cartesian coordinates, the most common procedure
makes use of collections of local charts, which are smooth bijections
(one-to-one correspondences) with pieces of $\R^d$. 

By taking
infinitesimal increments near a point $a\in M$, one obtains tangent
vectors, which are in the $d$-dimensional tangent
space $TM_a$, a generalisation of the tangent line to a curve and the
tangent plane to a surface. The union of all these tangent vectors
$\cup_{a\in M} TM_a$, denoted $TM$, is called the tangent bundle.  

As for ordinary vector spaces, one can define the dual of the tangent
bundle, noted $T^\ast M$, which can be constructed through linear
forms, called $1$-forms or cotangent vectors, acting on vectors of the
tangent bundle $TM$. The set of all these cotangent vectors is called
the cotangent bundle, noted $T^\ast M$. Similarly, $p$-forms, where
$p$ is an integer, are skew symmetric $p$-linear forms over the
tangent bundle $TM$. Note that in a flat (Euclidean) space $\R^d$ with
coordinates $x=\{x_i\}$, where $i = 1,\ldots d$, a 1-form is simply an
expression $\sum_i a_i(x)dx_i$, which depends linearly on the
infinitesimal increments $dx_i$. It is also interesting to note that
1-forms were in common use in fluid mechanics in the works of
D'Alembert, Euler and Lagrange more than a century before vectors were
commonly used, say, in the lectures of Gibbs.

An  important operator on $p$-forms is the \textit{exterior
  derivative}, $d$,
which linearly maps $p$-forms to $(p+1)$-forms (see
Appendix~\ref{ssec:CICF:EDIP}). An explicit definition of $d$  is not 
very helpful to build an intuitive feeling,
but it is worth pointing out that the square of $d$
is zero or, in words, an exact form (a form that is the exterior
derivative of another one) is closed (its exterior derivative
vanishes). Under certain conditions, to which we shall come back, the
converse is true.

\subsection{Lie advection: an extension of the Lagrangian (material) derivative}
\label{sec:CLT}
In this section we present some standard mathematical concepts needed
to introduce our theorem on \textit{generalised Cauchy invariants},
stated in the next section. For this, we need to generalise the fluid
mechanics concept of \textit{Lagrangian invariant}, which applies to a
scalar quantity  that does not change along fluid particle
trajectories.
The generalisation is called \textit{Lie-advection invariance}
(alternative terminologies found in the literature are
``Lie-transport'' and ``Lie-dragging''). 

First we introduce the \textit{pullback} and \textit{pushforward}
operations, which arise naturally when applying a change of variable,
here, between Lagrangian and Eulerian coordinates at a fixed time $t$
(later, we shall let this \textit{dynamical time}  vary). The Lagrangian
variable (initial position of the fluid particle), denoted by $a$, is
on a manifold $M$ (called here for concreteness \textit{Lagrangian}), while the Eulerian variable (current position of
the fluid particle), denoted by $x$, is on a manifold $N$ (called here
\textit{Eulerian}). The sets
$M$ and $N$ may or may not coincide.  The Lagrangian map
linking $a\in M$ to $x\in N$ is defined as follows
\begin{equation}
\begin{tabular}{rcll}
  $\varphi :$ & $M $ & $ \rightarrow $ &$N$\\
   & $a$ & $\mapsto$ & $x=\varphi(a)$.
\end{tabular}
\label{LaMa}
\end{equation}
The change of variable $a\rightarrow x=\varphi(a)$ induces two
operations that connect objects (such as functions, vectors, forms
and tensors), defined on $M$ to corresponding ones, defined on
$N$. They are the \textit{pushforward} operator, which sends objects
defined on $M$ to ones defined on $N$ and its inverse, the
\textit{pullback} operator.  To  define these transformations
precisely, it is convenient to consider successively the cases where 
these operators act on 
real-valued  functions (scalars), then on
vectors, then on $1$-forms, and finally on more involved
objects such as $p$-forms, obtainable from the former ones by linear
combinations of tensor products
(see Appendix~\ref{ssec:CICF:Te}). 

\begin{figure}
  \psfrag{g}{$\gamma_s$}
  \psfrag{TMa}{$TM_a$}
  \psfrag{M}{$M$}
  \psfrag{X}{$X$}
  \psfrag{a}{$a$}
  \psfrag{N}{$N$}
  \psfrag{TNx}{$TN_x$}
  \psfrag{xefa}{$x=\varphi(a)$}
  \psfrag{fogs}{$\varphi\circ\gamma_s$}
  \psfrag{fsX}{$\varphi_\ast X$}
  \psfrag{f}{$\varphi$}
  \psfrag{fi}{$\varphi^{-1}$}
  \psfrag{push}{$\varphi_\ast=\mathrm{pushforward}$}
  \psfrag{pull}{$\varphi^\ast=\mathrm{pullback}$}
  \psfrag{phrase1}{$\mbox{Set of objects}$}
  \psfrag{phrase2}{$\mbox{defined on }$}
  \psfrag{phrase3}{$\mbox{such as tensor fields, ...}$}
  \includegraphics[scale=0.3]{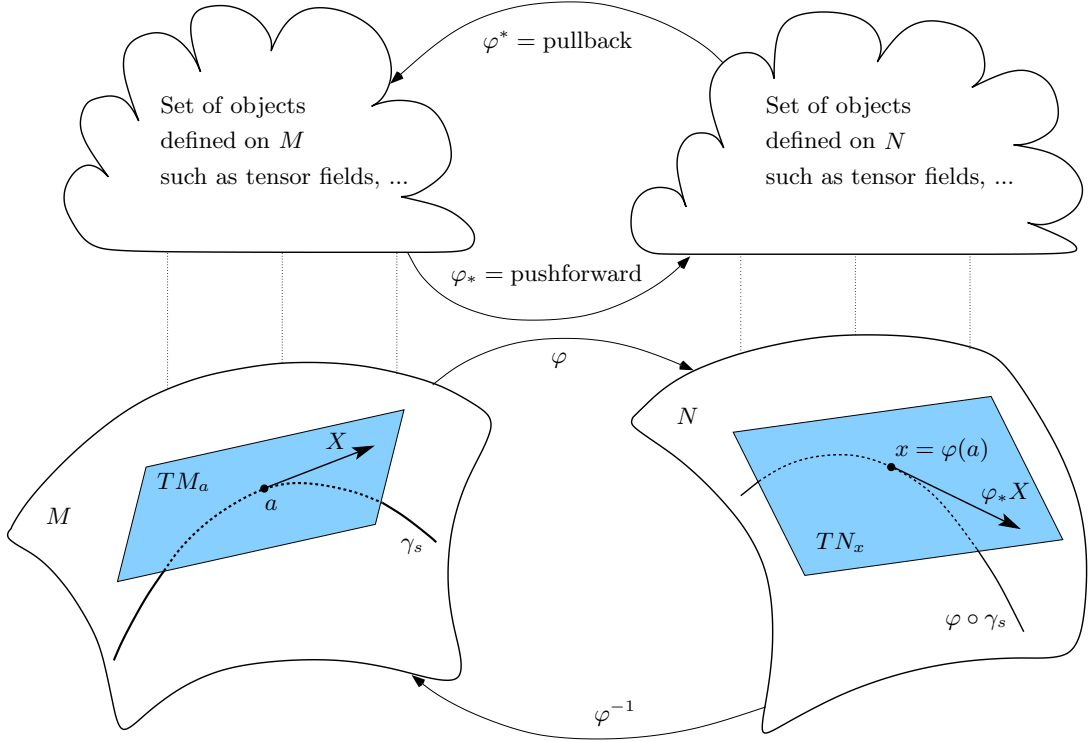}
  \caption{The curve $\gamma_s$ is the integral curve of a given
    vector field $X$, while the curve $\varphi\circ\gamma_s$
    is the integral curve of $\varphi_\ast X$. The pullback transformation $\varphi^\ast$ is a change
    of variables from Eulerian to Lagrangian coordinates, while the pushforward transformation $\varphi_\ast$
    is a change of variables from Lagrangian to Eulerian coordinates.}
  \label{difgeo}
\end{figure}

For the case of scalars, namely elements of 
$\mathcal{F}(M)$,  the set of real-valued smooth functions
defined on $M$, the pullback is simply a change of variable from
Eulerian to Lagrangian variables and the pushforward is the converse.
Specifically, the  pushforward of a function $f:M\rightarrow \R$ on
$M$ is $\varphi_\ast f=f\circ\varphi^{-1}$, where the symbol $\circ$
denotes the usual composition of maps. Conversely, the pullback
of a function $f:N\rightarrow \R$ on $N$ is $\varphi^\ast
f=f\circ\varphi$.

Now, we turn to vector fields, denoted $\mathcal{T}_0^1(M)$ on $M$, a subset
of the tangent bundle $TM$. (Why we use the notation $\mathcal{T}_0^1(M)$ will
become clear later.) At this point we cannot just make a change of
variable, because the Lagrangian and the Eulerian vectors take values in
different tangent spaces. But we can reinterpret tangent vectors to a
manifold in terms of differentials of scalar functions defined on that
manifold.  To implement this, it is useful to consider a vector field
$X(a)$ on $M$
as the generator of a suitable flow on $M$. For this, we need an auxiliary
time variable, denoted $s$, to parametrise a family of smooth maps
$\gamma_s:M\rightarrow M$. Observe that the time $s$ is not related to the
dynamical time $t$, which so far is held fixed. The maps satisfy the following
equations 
\[
\dot{\gamma_s}:=\frac{d\gamma_s}{ds} =X(\gamma_s), \quad \gamma_0(a)=a, \ \ a\in M.
\] 
The pushforward of
the vector field $Y\in \mathcal{T}_0^1(N)$ (also called the differential of
the map $\varphi$ or the tangent map) is now defined locally at the point
$a\in M$, as the linear map $ \varphi_\ast := T_a\varphi:TM_a \rightarrow
TN_{\varphi(a)}$, obtained by simply identifying the
resulting vector with the tangent vector to the mapped curve.
This is illustrated in figure~\ref{difgeo}.  Translated into
equations it means that 
\[
\varphi_\ast X= \varphi_\ast \left(\frac{d\gamma_s}{ds} \right) \!\bigg
|_{s=0} := \left(\frac{d}{ds}\varphi \circ\gamma_s\right)\!\bigg |_{s=0}
=T\varphi \circ X \circ \varphi^{-1},
\]
where $T$ denotes the tangent map and is given locally by the Jacobian matrix
$J_\varphi=J(\varphi)=\partial \varphi /\partial a$.
Recalling that $a^i$ denotes local coordinates on $M$ and $x^i$ local
coordinates on $N$, in terms of these local coordinates, this formula is
expressed equivalently as
\[
(\varphi_\ast X)^i(x)=\frac{\partial \varphi^i}{\partial
  a^j}(a)X^j(a)=\frac{\partial x^i}{\partial a^j}(a)X^j(a).
\]
To define the inverse operation, the pullback denoted $\varphi^\ast$, we just interchange $\varphi$
and $\varphi^{-1}$. Thus we have 
$\varphi_\ast=(\varphi^{-1})^\ast=(\varphi^\ast)^{-1}$ and 
$\varphi^\ast=(\varphi^{-1})_\ast=(\varphi_\ast)^{-1}$).  It thus  follows that the pullback
of a vector field $Y\in \mathcal{T}_0^1(N)$ on $N$ is
\[
\varphi^\ast Y=(T\varphi)^{-1} \circ Y \circ \varphi, \quad \mbox{or componentwise }\ \
(\varphi^\ast Y)^i(a)=\frac{\partial (\varphi^{-1})^i}{\partial x^j}(x)Y^j(x)=\frac{\partial a^i}{\partial x^j}(x)Y^j(x).
\]
Therefore we find that $\varphi^\ast Y=(\varphi^\ast Y)^i(a)(\partial/\partial
a^i)=Y^i(x)(\partial/\partial x^i)$ and $\varphi_\ast X =(\varphi_\ast
X)^i(x)(\partial/\partial x^i)=X^i(a)(\partial/\partial a^i)$.  Notice that
$\varphi$ must be a diffeomorphism (one-to-one smooth map) in order for the
pullback and pushforward operations to make sense; the only exception to this
is the pullback of functions (and covariant tensors, see
Appendix~\ref{ssec:CICF:PP}), since the inverse map is then not needed. Thus
vector fields can only be pulled back and pushed forward by diffeomorphisms.

We can extend pullback and pushforward operations to linear forms on vector
fields, that is 1-forms or covectors. The set of such $1$-forms fields on $M$ is denoted by $\mathcal{T}_1^0(M) \subset T^\ast M$
(see Appendices~\ref{ssec:CICF:Te}~and~\ref{ssec:CICF:EADF}).  
In order to define the pullback of a $1$-form $\alpha\in\mathcal{T}_1^0(N)$, we  introduce the linear map
$\varphi^\ast:T^\ast N_{\varphi(a)}\rightarrow T^\ast M_a$, defined by
\begin{equation}
 \label{duality_formvector}
\langle \varphi^\ast\alpha ,X\rangle : =  \langle\alpha ,\varphi_\ast X\rangle, \quad X\in\mathcal{T}_0^1(M), \ \  \alpha\in\mathcal{T}_1^0(N),
\end{equation}
where the duality bracket $\langle \cdot ,\cdot\rangle$ is the natural pairing
between the spaces $TM_a$ and $T^\ast M_a$ or between the spaces
$TN_{\varphi(a)}$ and $T^\ast N_{\varphi(a)}$. The pushforward of a $1$-form
$\beta\in\mathcal{T}_1^0(M)$, is defined by changing $\varphi$ to
$\varphi^{-1}$, i.e.  $\varphi_{\ast}:=(\varphi^{-1})^\ast$. In terms of local
coordinates we have
\[
(\varphi^\ast \alpha)_i(a)=\frac{\partial x^j}{\partial a^i}(a)\alpha_j(x),\quad \quad
(\varphi_\ast \beta)_i(x)=\frac{\partial a^j}{\partial x^i}(x)\beta_j(a).
\]
Therefore we find that $\varphi^\ast \alpha:=(\varphi^\ast \alpha)_i(a)da^i=\alpha_i(x)dx^i$ and
$\varphi_\ast \beta:=(\varphi_\ast \beta)_i(x)dx^i=\beta_i(a)da^i$.

Pullback and pushforward operations are easily generalised to
tensor fields $\Uptheta \in\mathcal{T}_p^q(M)$, where $\mathcal{T}_p^q(M)$
denotes the set of $p$-covariant and $q$-contravariant tensor fields on
$M$. Such generalisations follow naturally since a $p$-covariant and
$q$-contravariant tensor can be written  as linear combinations of
tensor products of $p$ $1$-forms
and $q$ vectors (see Appendix~\ref{ssec:CICF:Te}).

In order to define the \textit{Lie derivative}, we bring in the dynamical time
$t$. For this, we specialize to the case where the Lagrangian and the
Eulerian points are on the same manifold (with $N=M$) and we  consider
a time-dependent vector field $v_t$, the \textit{velocity field,} taken in
$\mathcal{T}_0^1(M)$ for all $t\geq 0$. This velocity field is
prescribed a priori and we do not have to specify which dynamical
equation it satisfies. We define a time-dependent Lagrangian
map $\varphi_t$ in the usual fluid mechanical sense as mapping the
initial position of a fluid particle, following the flow, to its
position at time $t$, namely as the 
solution of the ordinary differential equation (ODE)
\begin{equation}
\label{ODE_LaFlow}
\dot{\varphi}_t=\frac{d}{dt}\varphi_t=v(t,\varphi_t), \quad \varphi_0={\rm Identity}.
\end{equation}
From this equation, 
we also define a $2$-time Lagrangian map $\varphi_{t,s}$ with $t\ge
0$ and $s\ge 0$ as the map from the position of a fluid particle at time $s$ to
its position at time $t$. Allowing the time to run backwards, we do not
impose $t>s$. Obviously, we have
$\varphi_{t,0}=\varphi_t$. Furthermore, we obviously have the group composition rule
\begin{equation}
\label{group_prop_transitif}
{\varphi}_{t,s}=\varphi_{t,\tau}\circ \varphi_{\tau,s}\qquad \forall\,
t\ge 0, \forall\,\tau\ge 0, \forall\, s\ge 0.
\end{equation}

In this dynamical setting, the pullback and pushforward operations
consist roughly in following a given tensor field, while taking into account
the geometrical deformation of the tensor basis, along the Lagrangian
flow.  This will naturally lead to considering a derivative with respect
to the Lagrangian flow, called the Lie derivative.  The Lie
derivative of a structure (for instance a vector, a $1$-form or a
tensor field) with respect to the time-dependent vector field $v_t$
measures the instantaneous rate of geometrical variation of the
structure (tensor basis) as it is transported and deformed by the
Lagrangian flow
$\varphi_t$ generated by $v_t$. 

Specifically, we first define the Lie derivative acting on a time-independent tensor field $\Uptheta \in
\mathcal{T}_p^q(M)$.  To the Lagrangian map $\varphi_{\tau,t}$ we
associate its pullback  $\varphi_{\tau,t}^\ast$, constructed just as
earlier in this section, when the dynamical time was held fixed. The
Lie derivative with respect to $v_t$  is defined by 
\begin{equation}
\label{def:LieDer}
\mathsterling_{v_t} \Uptheta := \left(\frac{d }{d\tau}\varphi_{\tau,t}^\ast \Uptheta\right) \!\bigg |_{\tau=t}.
\end{equation}

Now,  we turn to a time-dependent tensor field
$\Uptheta_t \in \mathcal{T}_p^q(M)$ and we derive the \textit{Lie derivative theorem}
$\forall t\geq 0$. For this, we calculate the time-derivative of
$\varphi_{t,s}^\ast \Uptheta_t$, 
using the product rule for derivatives and  obtain
\[
\frac{d }{d t}\varphi_{t,s}^\ast \Uptheta_t =\frac{d }{d\tau}\varphi_{\tau,s}^\ast \Uptheta_\tau \bigg |_{\tau=t}
= \frac{d }{d\tau}\varphi_{\tau,s}^\ast \Uptheta_t \bigg |_{\tau=t} +
\varphi_{t,s}^\ast \frac{d }{d\tau}\Uptheta_\tau \bigg |_{\tau=t}.
\]
Then, using the group composition rule \eqref{group_prop_transitif}, this equation becomes
\[
\frac{d }{d t}\varphi_{t,s}^\ast \Uptheta_t
= \frac{d }{d\tau}(\varphi_{\tau,t}\circ \varphi_{t, s})^\ast \Uptheta_t \bigg |_{\tau=t} +
\varphi_{t,s}^\ast \partial_t \Uptheta_t. 
\]
Using a property for the pullback of map composition (see Appendix~\ref{ssec:CICF:PP}), 
namely $(\varphi \circ \psi)^\ast= \psi^\ast \varphi^\ast$, we obtain
\[
\frac{d }{d t}\varphi_{t,s}^\ast \Uptheta_t
= \frac{d }{d\tau} \varphi_{t, s}^\ast \varphi_{\tau,t}^\ast\Uptheta_t \bigg |_{\tau=t} +
\varphi_{t,s}^\ast \partial_t \Uptheta_t=
\varphi_{t, s}^\ast\frac{d }{d\tau}\varphi_{\tau,t}^\ast\Uptheta_t \bigg |_{\tau=t} +
\varphi_{t,s}^\ast \partial_t \Uptheta_t.
\]
Finally, using the definition of the Lie derivative \eqref{def:LieDer}, this equation leads 
to the following formula, known as the Lie derivative theorem:
\begin{equation}
\label{LDT}
\frac{d }{dt}\varphi_{t,s}^\ast  \Uptheta_t= \varphi_{t,s}^\ast(\partial_t  \Uptheta_t+ \mathsterling_{v_t}  \Uptheta_t), \quad \forall t\geq 0, \ \  \forall s\geq 0.
\end{equation}

In this paper, a central role will be played by tensor fields that
are \textit{Lie-advection invariant} (in short
\textit{Lie invariant}). 
A Lie-advection invariant 
tensor field  $\Uptheta_t$ is such that its Lagrangian pullback, i.e. its 
pullback  to time $t=0$,  is equal to the initial tensor field, that is
\begin{equation}
\label{Lie-pullback}
\varphi_{t}^\ast  \Uptheta_t= \Uptheta_0.
\end{equation}
In fluid mechanics terms, one then states that the tensor field $\Uptheta_t$ is frozen into the flow $\varphi_t$.
From the Lie derivative theorem \eqref{LDT}, we immediately find that
this is equivalent to having the tensor field $\Uptheta_t$ satisfying
the equation 
\begin{equation}
\label{frozen}
\partial_t  \Uptheta_t+ \mathsterling_{v_t}  \Uptheta_t=0,
\end{equation}
which is called the Lie advection equation. A tensor field $\Uptheta_t$ 
satisfying the Lie-advection  equation \eqref{frozen} is said to be  Lie-advected by the flow of $v_t$.

It is easily checked
that when $\Uptheta_t$ is a scalar field (denoted $\theta_t$) and when the
manifold reduces to an Euclidean space,  \eqref{frozen} becomes just
\begin{equation}
\label{frozenscalar}
\partial_t  \theta_t+ v_t^i \partial_i\theta_t=0,
\end{equation}
where $\partial_i$ is the Eulerian derivative.
Hence, in the scalar case,
Lie-advection invariance of $\theta_t$ is the same as stating that $\theta_t$
is a Lagrangian (material) invariant in the usual fluid mechanical
sense. The advantage of the Lie-advection invariance formulation for
higher-order objects is that, e.g., in 3D the vorticity -- when considered
as a $2$-form -- is then also Lie-advection invariant, as noticed for
the first time  (in 19th century language) by \citet{Hel58}.

\subsection{Generalised Cauchy invariants}
\label{sec:GCI}

In this section we state a general theorem about Lie-advection invariance using differential geometry tools.
The result is a natural generalisation
of Cauchy invariants that arises when we consider,
in an Euclidean space $\R^d$ or on a
$d$-dimensional Riemannian manifold $(M,g)$, a Lie-advected  $p$-form
with a crucial additional constraint of exactness (or some genereralization).
We recall that a $d$-dimensional Riemannian manifold $(M,g)$ is a differentiable manifold $M$ of dimension $d$, together with a $2$-covariant
tensor field, the metric tensor $g$, which associates to any point $a
\in M$ a $2$-covariant tensor $\mathcal{T}_2^0(M)$(see Appendices~\ref{ssec:CICF:Te}~and~\ref{ssec:CICF:RM}). The
metric tensor $g$ allows one both to define a metric on $M$ for measuring distances between two points on $M$,
and  to define a suitable scalar product for vectors lying in a
tangent space (see Appendix~\ref{ssec:CICF:RM}). 

The main new result of the present section will be to show that, to each
exact Lie-advected $p$-form, corresponds a generalised Cauchy
invariant.
This is of course a result with applications beyond hydrodynamics,
but it is not just a rewriting of Lie-advection invariance in
Lagrangian coordinates: the Cauchy-invariants formulation requires an
additional condition other than Lie-advection invariance.
The method of proving this is quite general but, of course,
also applies to Euler flow in the ordinary flat 3D space. In that
case, we already have the original proof of \citet{Cau27}, which
juggles with Eulerian and Lagrangian coordinates and thus has a flavour
of pullback-pushforward  argument. In addition, we have all the
relatively recent derivations using Noether's theorem in conjunction
with a variational formulation of the Euler equations and the
relabelling invariance \citep[see, e.g.,][]{Sal88}. What we now
present constitutes in a sense a third approach, rooted in differential geometry
and allowing generalisation to a variety of  hydrodynamical and MHD
problems,
discussed in Sections~\ref{ss:broad}, \ref{sec:VIFIH} and \ref{sec:AHMHD}.

Let $\Omega \subset M$ be a bounded region of the $d$-dimensional Riemannian manifold $(M,g)$.
We remind the reader that a  $p$-form $\gamma \in \Lambda^p(\Omega)$ is exact if it is the
exterior derivative of a $(p-1)$-form 
$\alpha\in  \Lambda^{p-1}(\Omega)$, that is
\begin{equation}
\label{def:exact_form}
\gamma=d\alpha, \quad \gamma \in \Lambda^p(\Omega), \ \ \alpha\in  \Lambda^{p-1}(\Omega),
\end{equation}
where $d$ denotes the exterior derivative (see Appendix~\ref{ssec:CICF:EDIP}). 
We recall that a family of $p$-forms $\gamma_t  \in \Lambda^p(\Omega)$, $t>0$, are Lie-advected by the flow of $v_t$
if they satisfy the Lie advection equation
\begin{equation}
\label{eqn:LieAdvection}
\partial_t \gamma_t + \mathsterling_{v_t} \gamma_t =0, \quad \mbox{on} \ \Omega\subset M, \quad \mbox{with} \ \gamma_0 \ \mbox{given}.
\end{equation}
Here the vector field $v_t$ is the generator of the Lagrangian flow
$\varphi_t$ defined by \eqref{ODE_LaFlow}.
 
\begin{theorem}(Generalised Cauchy invariants equation).
\label{th:GCI}
For $t>0$, let $\gamma_t \in \Lambda^p(\Omega)$ be a time-dependent family of exact $p$-forms 
(i.e. satisfying  \eqref{def:exact_form})
that are Lie-advected (i.e. satisfy  \eqref{eqn:LieAdvection}); then we have the generalised
Cauchy invariants equation
\begin{equation}
\label{eqn:GCI}
\frac{1}{(p-1)!}\,\delta_{j_1\ldots j_{p-1}}^{i_1\ldots i_{p-1}} d\alpha_{i_1\ldots i_{p-1}}\,
\wedge dx^{j_1}\wedge \ldots \wedge dx^{j_{p-1}}=\gamma_0.
\end{equation}
\end{theorem}
Here,  $x=\varphi_t$ denotes  Eulerian coordinates and $\delta_{j_1\ldots j_{p}}^{i_1\ldots i_{p}}$  the generalised
Kronecker symbol (see Appendix~\ref{ssec:CICF:PKD}). Note that, henceforth, in
connection with Cauchy invariants, we use the singular for ``equation'', since in
modern writing a vector or a tensor is considered a single object.
\begin{proof} 
Since $\gamma$ is Lie-advected,  by the Lie derivative theorem \eqref{LDT}, we have
$\varphi_t^\ast\gamma = \gamma_0$. Then, we write
$\varphi_t^\ast\gamma$ in terms of its component in the
$a$-coordinates (see Appendix~\ref{ssec:CICF:PP}), to obtain 
\begin{eqnarray}
\gamma_0 &=&\varphi_t^\ast\gamma\nonumber \\
&=& \sum_{i_1<\ldots < i_p} (\varphi_t^\ast \gamma)_{i_1 \ldots i_p}  \,da^{i_1}\wedge \ldots \wedge da^{i_p} \nonumber \\
&=& \frac{1}{p!} \frac{\partial x^{j_1}}{\partial a^{i_1}} \ldots \frac{\partial x^{j_p}}{\partial a^{i_p}}\,
\gamma_{j_1\ldots j_p}(x)\,  da^{i_1}\wedge \ldots \wedge da^{i_p}. \label{gci_1}
\end{eqnarray}  
Next, using the generalised Kronecker symbol $\delta_{j_1\ldots
  j_{p}}^{i_1\ldots i_{p}}$, we obtain
\begin{eqnarray*}
  \gamma &=& d\alpha \nonumber \\
         &=& d \sum_{i_1<\ldots < i_{p-1}} \alpha_{i_1\ldots  i_{p-1}} (a)\,\wedge da^{i_1}\wedge \ldots \wedge da^{i_{p-1}}\\
         &=&  \sum_{i_1<\ldots < i_{p-1}} d\alpha_{i_1\ldots  i_{p-1}} \,\wedge da^{i_1}\wedge \ldots \wedge da^{i_{p-1}}\\
  &=& \sum_{i_1<\ldots < i_{p-1}} \frac{\partial}{\partial a^k} \alpha_{i_1 \ldots i_{p-1}} da^k\wedge da^{i_1}\wedge \ldots \wedge da^{i_{p-1}}\\
  &=&   \sum_{i_1<\ldots < i_{p-1}} \delta_{j_1\ldots j_{p}}^{k i_1\ldots i_{p-1}}
  \frac{\partial}{\partial a^k} \alpha_{i_1\ldots  i_{p-1}} da^{j_1}\wedge \ldots \wedge da^{j_{p}},
\end{eqnarray*}  
from which we deduce
\begin{equation}
\label{gci_2}
\gamma_{l_1\ldots  l_{p}}(x) =  \delta_{l_1\ldots l_{p}}^{k i_1\ldots i_{p-1}}\frac{\partial}{\partial x^k} \alpha_{i_1\ldots  i_{p-1}}(x). 
\end{equation}  
Substituting \eqref{gci_2} into   \eqref{gci_1} we obtain
\begin{eqnarray}
  \gamma_0&=&  \frac{1}{p!} \delta_{l_1\ldots l_{p}}^{j_p j_1\ldots j_{p-1}}
  \frac{\partial x^{l_1}}{\partial a^{i_1}} \ldots \frac{\partial x^{l_p}}{\partial a^{i_p}}\,
  \frac{\partial}{\partial x^{j_p}}  \alpha_{j_1\ldots  j_{p-1}}
  \,  da^{i_1}\wedge \ldots \wedge da^{i_p} \nonumber\\
  &=&  \frac{1}{p!} \delta_{l_1\ldots l_{p}}^{j_p j_1\ldots j_{p-1}}
  \frac{\partial x^{l_1}}{\partial a^{i_1}} \ldots
  \frac{\partial x^{l_p}}{\partial a^{i_p}} \frac{\partial a^k}{\partial x^{j_p}}\,
  \frac{\partial}{\partial a^{k}}  \alpha_{j_1\ldots  j_{p-1}}
  \,  da^{i_1}\wedge \ldots \wedge da^{i_p} \nonumber\\
  &=&\frac{1}{p!} (-1)^{p-1}\delta_{l_1\ldots l_{p}}^{j_1\ldots j_{p}}
  \frac{\partial x^{l_1}}{\partial a^{i_1}} \ldots
  \frac{\partial x^{l_p}}{\partial a^{i_p}} \frac{\partial a^k}{\partial x^{j_p}}\,
  \frac{\partial}{\partial a^{k}}  \alpha_{j_1\ldots  j_{p-1}}
  \,  da^{i_1}\wedge \ldots \wedge da^{i_p}.   \label{gci_3}
\end{eqnarray}  
Using now the Laplace expansion of determinants, we may define recursively
\begin{eqnarray}
\delta_{i_1\ldots i_{p}}^{j_1\ldots j_{p}}&=&
\left |
\begin{tabular}{lll}
  $\delta_{i_1}^{j_1}$ & $\ldots $ & $ \delta_{i_p}^{j_1} $\\
  $\vdots$ & $\ddots$ & $\vdots$\\
   $\delta_{i_1}^{j_p}$ &$\ldots $ &  $ \delta_{i_p}^{j_p} $
\end{tabular}
\right| \nonumber \\  
  &=&\sum_{k=1}^p(-1)^{p+k}
  \delta_{i_k}^{j_p}\delta_{i_1\ldots \widehat{i}_k \ldots i_{p}}^{j_1\ldots j_{k}\ldots\widehat{j}_p},
  \label{recursiv_kron}
\end{eqnarray}
where the hat indicates an omitted  index in the sequence. Using \eqref{recursiv_kron},
equation \eqref{gci_3} becomes
\begin{eqnarray*}
  \gamma_0&=& \frac{1}{p!}\sum_{n=1}^p(-1)^{n-1}
  \delta_{l_n}^{j_p}\delta_{l_1\ldots \widehat{l}_n \ldots l_{p}}^{j_1\ldots j_{p-1}}
  \frac{\partial x^{l_1}}{\partial a^{i_1}} \ldots
  \frac{\partial x^{l_p}}{\partial a^{i_p}} \frac{\partial a^k}{\partial x^{j_p}}\,
  \frac{\partial}{\partial a^{k}}  \alpha_{j_1\ldots  j_{p-1}}
  \,  da^{i_1}\wedge \ldots \wedge da^{i_p} \nonumber \\
  &=&\frac{1}{p!}\sum_{n=1}^p(-1)^{n-1}
  \delta_{l_1\ldots \widehat{l}_n \ldots l_{p}}^{j_1\ldots j_{p-1}}
  \frac{\partial x^{l_1}}{\partial a^{i_1}} \ldots
  \frac{\partial x^{l_p}}{\partial a^{i_p}} \frac{\partial a^k}{\partial x^{l_n}}\,
  \frac{\partial}{\partial a^{k}}  \alpha_{j_1\ldots  j_{p-1}}
  \,  da^{i_1}\wedge \ldots \wedge da^{i_p} \nonumber \\
  &=&\frac{1}{p!}\sum_{n=1}^p(-1)^{n-1}
  \delta_{l_1\ldots \widehat{l}_n \ldots l_{p}}^{j_1\ldots j_{p-1}}
  \frac{\partial x^{l_1}}{\partial a^{i_1}} \ldots
  \widehat{\frac{\partial x^{l_n}}{\partial a^{i_n}}} \ldots
  \frac{\partial x^{l_p}}{\partial a^{i_p}}
  \delta_{i_n}^k\,
  \frac{\partial}{\partial a^{k}}  \alpha_{j_1\ldots  j_{p-1}}
  \,  da^{i_1}\wedge \ldots \wedge da^{i_p} \nonumber \\
    &=&\frac{1}{p!}\sum_{n=1}^p(-1)^{n-1}
  \delta_{l_1\ldots \widehat{l}_n \ldots l_{p}}^{j_1\ldots j_{p-1}}
  \frac{\partial x^{l_1}}{\partial a^{i_1}} \ldots
  \widehat{\frac{\partial x^{l_n}}{\partial a^{i_n}}} \ldots
  \frac{\partial x^{l_p}}{\partial a^{i_p}}
  \frac{\partial}{\partial a^{i_n}}  \alpha_{j_1\ldots  j_{p-1}}
  \,  da^{i_1}\wedge \ldots \wedge da^{i_p} \nonumber \\
  &=&\frac{1}{p!}\sum_{n=1}^p(-1)^{n-1}
  \delta_{l_1\ldots \widehat{l}_n \ldots l_{p}}^{j_1\ldots j_{p-1}}
  \,  \left(\frac{\partial x^{l_1}}{\partial a^{i_1}}da^{i_1}\right)\wedge \ldots\wedge
  \left(\frac{\partial}{\partial a^{i_n}}  \alpha_{j_1\ldots  j_{p-1}} da^{i_n}\right)\wedge \ldots \wedge 
  \left( \frac{\partial x^{l_p}}{\partial a^{i_p}}da^{i_p}\right) \nonumber \\
   &=&\frac{1}{p!}\sum_{n=1}^p
  \delta_{l_1\ldots \widehat{l}_n \ldots l_{p}}^{j_1\ldots j_{p-1}}\,
  d\alpha_{j_1\ldots  j_{p-1}} \wedge dx^{l_1}\wedge \ldots \wedge \widehat{dx^{l_n}}
  \wedge \ldots \wedge  dx^{l_p} \nonumber \\
   &=&\frac{1}{(p-1)!}
  \delta_{l_1 \ldots l_{p-1}}^{j_1\ldots j_{p-1}}\,
  d\alpha_{j_1\ldots  j_{p-1}} \wedge dx^{l_1}\wedge \ldots
 \wedge  dx^{l_{p-1}}, \nonumber \\
\end{eqnarray*}
which ends the proof.
\end{proof}

\begin{remark}{\rm
\label{rem:closedexact}
(\textit{Sufficient conditions for exactness of differential forms}).       
In Theorem \ref{th:GCI} on the construction of generalised Cauchy invariants,
we demand that the  $p$-form $\gamma$ be exact. There are several
ways to obtain such an exact $p$-form.
\begin{itemize}
\item[1.] In some problems  a $p$-form $\gamma$ appears naturally as the exterior differential of
a $(p-1)$-form $\beta\in \Lambda^{p-1}(\Omega)$, i.e
$\gamma=d\beta$. As we will see in later sections, this is the case for the vorticity $2$-form and the magnetic field $2$-form.
\item[2.] When $\gamma_t$ is Lie advected and the initial condition $\gamma_0$
  is  exact, it follows from the commutation of the exterior derivative and the
  pushforward operator $\varphi_{t\ast}$\, (see  Appendix~\ref{ssec:CICF:EDIP}), that  $\gamma_t$ is exact. Indeed $\varphi_t^\ast\gamma=\gamma_0$ implies that
\[
\gamma=\varphi_{t\ast}\gamma_0=\varphi_{t\ast}d\alpha_0= d(\varphi_{t\ast}\alpha_0).
\]
\item[3.] Let us introduce $Z^p(M; \R)$, the subspace of $\Lambda^p(M)$ constituted of all closed $p$-forms and
$B^p(M; \R)$, the subspace of $Z^p(M;\R)$ constituted of all exact
  $p$-forms. Obviously,  we have $B^p(M; \R)\subset Z^p(M; \R) \subset \Lambda^p(M)$.
  Altough  $B^p$ and $Z^p$ are infinite-dimensional, in many cases their
  quotient space, called the $p$-th cohomology vector space and noted
\[
H^p(M;\R):= \frac{Z^p(M;\R)}{B^p(M;\R)},
\]
is finite-dimensional. For example, this is the case when $M$ is a compact finite-dimensional manifold.
The dimension of the vector space $H^p$ is called the $p$-th Betti number, written $b_p = b_p(M)$ and defined by
$b_p(M) := {\rm dim }\  H^p(M;\R)$. Thus the Betti number $b_p(M)$ is
the
maximum number of closed $p$-forms on $M$, such that all linear
combinations with non-vanishing coefficients are not exact. The knowledge of the Betti numbers of a given manifold
$M$ for $p\geq 1$ yields an exact quantitative answer to the question about exactness of a closed $p$-form:
\[
\mbox{a closed $p$-form is exact if and only if } b_p(M)=0. 
\]
Two closed forms are equivalent or cohomologous if they differ by an exact form, and
a closed $p$-form is exact if and only if it is cohomologous to zero. The values of the Betti
numbers are related to the topological properties of the manifold $M$ (e.g., homology, connectedness, curvature, ...).
For more details on cohomology and homology we refer the reader to Appendix~\ref{ssec:CICF:HCHD} and references therein.

\item[4.] By the Poincar\'e theorem \citep[see, e.g.,][Theorem 6.4.14]{AMR88}, if the $p$-form $\gamma$ is closed on $\Omega\subset M$, i.e.
$d\gamma=0$ on $\Omega$, then $\gamma$ is locally exact; that is, there exists a neighborhood $U\subset \Omega$ 
about each point of $\Omega$, on which $\gamma_{|_{U}} = d\alpha $ for
some $(p-1)$-form
$\alpha\in\Lambda^{p-1}(U)$.
The same result holds globally on a {\it contractible} domain \citep[][see Lemma 6.4.18]{AMR88}.
A contractible domain is roughly one in which, for any given point, the whole domain can be continuously shrunk into it
(see  Appendix~\ref{ssec:CICF:Ma}). By the Poincar\'e lemma, if $M$ is a compact $d$-dimensional contractible manifold,
all the Betti numbers (for $p\geq 1$) vanish, i.e.  $b_1(M)=\ldots =
b_d(M) =0$,  and $b_0(M) = 1$.
Contractibility is, however, an excessivily strong requirement to
ensure that closeness implies exacteness. For differential forms of a
given order $p$, the vanishing of the  single Betti number, $b_p(M) =
0$ is actually sufficient to ensure this. 
\end{itemize}
}\end{remark}

\subsection{Alternative formulations and extensions of Theorem~\ref{th:GCI}}
\label{sec:AFE}
Hereafter we discuss alternative representations of Theorem~\ref{th:GCI}, which are local,
such as  the generalised Cauchy formula, or global, such as  the integral
formulation of the Cauchy invariants equations.
We also give extensions of  Theorem~\ref{th:GCI} for some non-exact differential forms.

\subsubsection{Generalised Cauchy  formula}
\label{rk:DGCI}
An important operation in differential geometry  is the Hodge duality, which associates to any $p$-form a Hodge-dual $(d-p)$-form such that their
exterior product is the fundamental metric volume $d$-form $\mu=\sqrt{{\rm g}}da^1\wedge \ldots da^d$, 
with $\sqrt{\mathrm{g}}=\sqrt{{\rm det}(g_{ij})}$ (see
Appendix~\ref{ssec:CICF:HSECD}). For example, in 3D the vorticity
2-form and the vorticity vector field (as known since the work of \citet{Hel58})
are Hodge duals of each other. It is therefore of
interest to rewrite the Cauchy invariants equation and its
generalisations in Hodge-dual form. For example, as we shall see in
the next section, this will give us the Cauchy vorticity formula.

The generalised Cauchy invariants equation \eqref{eqn:GCI} has a corresponding generalised Cauchy formula
obtained by applying the Hodge dual operator, denoted $\star$, to \eqref{eqn:GCI}, that is
\begin{equation}
\label{eqn:GCIstar}
\frac{1}{(p-1)!}\,\delta_{j_1\ldots j_{p-1}}^{i_1\ldots i_{p-1}} \star(d\alpha_{i_1\ldots i_{p-1}}\,
\wedge dx^{j_1}\wedge \ldots \wedge dx^{j_{p-1}})=\star \gamma_0.
\end{equation}
This generalised Cauchy formula can be written in the covariant,
contravariant or mixed form, by using what is known in differential
geometry as the raising-lowering duality. We have already seen that
the space $ \mathcal{T}_0^1(M)$ is the vector space of
$1$-contravariant vector fields, while $ \mathcal{T}_1^0(M)$, its
dual, is the vector space of linear forms on $\mathcal{T}_0^1(M)$,
i.e. the space of $1$-covariant vector fields (also called covector or
$1$-form fields). We then introduce the index raising operator
$(\cdot)^\sharp: \mathcal{T}_1^0(M)\rightarrow \mathcal{T}_0^1(M)$,
which in flat space transforms the differential of a function into its
gradient vector. In curved spaces $\alpha^\sharp$ denotes the
$1$-contravariant vector field obtained from the $1$-form field
$\alpha$, by using the index raising operation $\alpha^\sharp =
(\alpha_ida^i)^\sharp=(\alpha^\sharp)^i\partial_i=g^{ij}\alpha_j\partial_i$;
that is componentwise $(\alpha^\sharp)^i=g^{ij}\alpha_j$.
Conversely $v^\flat$ is the $1$-form field obtained from the vector
field by applying the index lowering operator $(\cdot)^\flat:
\mathcal{T}_0^1(M)\rightarrow \mathcal{T}_1^0(M)$ according to the
formula $v^\flat =
(v^i\partial_i)^\flat=(v^\flat)_idx^i=g_{ij}v^jdx^i$; 
componentwise, this is  $(v^\flat)_i=g_{ij}v^j$ (see
Appendix~\ref{ssec:CICF:RM}).  Therefore, to obtain
\eqref{eqn:GCIstar} in the desired formulation (covariant,
contravariant or mixed form), it is required to successively apply as
many times as necessary the lowering and raising operators.

\begin{remark}{\rm
    We observe that the generalised Cauchy invariants equation (i.e.  Theorem~\ref{th:GCI})
    requires only a structure of differentiable manifold, without the
    Riemannian structure. In contrast, 
    the generalised Cauchy formula \eqref{eqn:GCIstar} requires such a
    Riemannian structure (see Appendix~\ref{ssec:CICF:RM}),
    because of the use of  Hodge duality (see Appendix~\ref{ssec:CICF:HSECD}).
  }
\end{remark}

\subsubsection{Space-integrated form of generalised Cauchy invariants equations}
\label{ss:SIGCI}
  Since the generalised Cauchy invariant is an exact $p$-form, we can
  apply to it what are known as the Hodge decomposition and/or the Stokes theorem.  First we write the
  generalised Cauchy invariant as an explicit exterior
  differential. We have indeed 
  \[
   \frac{1}{(p-1)!}\,\delta_{j_1\ldots j_{p-1}}^{i_1\ldots i_{p-1}} d\alpha_{i_1\ldots i_{p-1}}\,
   \wedge dx^{j_1}\wedge \ldots \wedge dx^{j_{p-1}}=
     \frac{1}{(p-1)!}\,\delta_{j_1\ldots j_{p-1}}^{i_1\ldots i_{p-1}} d(\alpha_{i_1\ldots i_{p-1}}\,
   \wedge dx^{j_1}\wedge \ldots \wedge dx^{j_{p-1}}).
   \]
   Since  $\gamma_0=d\alpha_0$, using the Hodge decomposition for closed forms (see Appendix~\ref{ssec:CICF:HCHD}), we obtain
   \begin{equation}
    \label{eq:integratedGCIE}
     \frac{1}{(p-1)!}\,\delta_{j_1\ldots j_{p-1}}^{i_1\ldots i_{p-1}} \alpha_{i_1\ldots i_{p-1}}\,
   \wedge dx^{j_1}\wedge \ldots \wedge dx^{j_{p-1}} = \alpha_0 + d\beta + {h}.
   \end{equation}
   Here, if $M$ is a compact manifold without (resp. with) boundary, 
   $\beta$ is an arbitrary $(p-2)$-form 
   (resp. normal $(p-2)$-form with vanishing tangential components;
   see next-to-last paragraph of Appendix~\ref{ssec:CICF:HCHD} and references therein).
   In \eqref{eq:integratedGCIE}, the $(p-1)$-form ${h}$ is harmonic, that is $d{h}=0$ and $d^\star{h}=0$.
   Here, the operator $d^\star:\Lambda^{p}(\Omega)\rightarrow\Lambda^{p-1}(\Omega)$ with $p\geq 0$
   is the exterior coderivative, obtained from the
   exterior derivative, but acting on the Hodge-dual space (for details
   see Appendix~\ref{ssec:CICF:HSECD}). More precisely, if $\gamma \in \Lambda^p(\Omega)$
   then we have the $(p-1)$-form $d^\star\gamma = (-1)^{d(p-1)+1} \star d\, \star \gamma$. 
   Note that the latter looks actually more like an integration than a differentiation. 

   Now, we want to integrate this
   form over suitable domains, called $1$-chains, $2$-chains, \hbox{... .}
   In a flat space, a $1$-chain is just a finite set of 1D contours.
   For a general definition of $p$-dimensional $p$-chains on
   manifolds, see Appendix~\ref{ssec:CICF:IDF}.
   Let  $c$ be a $(p-1)$-chain on the
   manifold $M$. Choosing the $(p-2)$-form $\beta$ with suitable
   values on the boundary $\partial c$ of $c$ to avoid having a
   boundary contribution  (if a boundary is present), we obtain, using  the
   Stokes theorem (see Appendix~\ref{ssec:CICF:IDF}), 
   \[
    \frac{1}{(p-1)!}\,\delta_{j_1\ldots j_{p-1}}^{i_1\ldots i_{p-1}} \int_c  \alpha_{i_1\ldots i_{p-1}}\,
   \wedge dx^{j_1}\wedge \ldots \wedge dx^{j_{p-1}}=  \int_c \alpha_0 +  \int_c {h}.
   \]
   Moreover, if the Betti number $b_{p-1}(M)=0$, then the second term
   on the right-hand side of the previous formula vanishes.
   Considering now a $p$-chain $c$, using the Stokes theorem, we obtain
   \[
    \frac{1}{(p-1)!}\,\delta_{j_1\ldots j_{p-1}}^{i_1\ldots i_{p-1}} \int_{\partial c}  \alpha_{i_1\ldots i_{p-1}}\,
   \wedge dx^{j_1}\wedge \ldots \wedge dx^{j_{p-1}}=  \int_{\partial c }\alpha_0.
   \]

\subsubsection{Generalisation to some non-exact differential forms}
\label{rk:GCI}
From Theorem~\ref{th:GCI}, the following question arises naturally: can we
extend the result of Theorem~\ref{th:GCI} when the $p$-form $\gamma$
is not exact? The answer is yes under some conditions.

We suppose that 
the $p$-form $\gamma$ of Theorem~\ref{th:GCI} can be written as $\gamma={\rm Op\,}\pi$,
where $\pi$ is a $q$-form and the operator ${\rm Op\,}:\Lambda^q(\Omega)\rightarrow \Lambda^p(\Omega)$
is a linear operator which satisfies the following conditions:
\begin{itemize}
\item [$\, $(i)] The commutation relation $[{\rm Op\,},\mathsterling_{v}]=0$ holds.
\item [(ii)] The kernel of the operator ${\rm Op\,}$ is such that
  ${\rm Ker}\,{\rm Op\,}=\{ \mbox{closed }\mbox{$q$-form, i.e. } \kappa\in \Lambda^q(\Omega) \ | \ d\kappa=0 \}$.
\end{itemize}
From assumption $(i)$ the Lie-advection equation \eqref{eqn:LieAdvection} is equivalent to 
${\rm Op\,}(\partial_t \pi + \mathsterling_{v} \pi)=0$. From assumption $(ii)$, this equation
is also equivalent to $\partial_t \pi + \mathsterling_{v} \pi=\kappa$, with $\kappa$ a closed $q$-form.
Taking the exterior derivative to this equation, we obtain the equation $\partial_t d\pi + \mathsterling_{v} d\pi=0$,
to which we can apply  Theorem~\ref{th:GCI} with $p=q+1$, $\gamma=d\pi$ and $\alpha=\pi$.

We give now three examples.
Choosing ${\rm Op\,}\equiv d$, the first one is obvious. The second example is 
${\rm Op\,}\equiv \star d:\Lambda^{d-p-1}(\Omega)\rightarrow \Lambda^p(\Omega)$. 
where the star denotes the Hodge dual operator. Then we have
${\rm Ker}\,\star d= \{ \mbox{exact }\mbox{$q$-form } + \mbox{ harmonic }\mbox{$q$-form}\} 
\subset \{ \mbox{closed }\mbox{$q$-form}\}$, where a harmonic $q$-form $h$ satisfies
$dh=d^\star h=0$, with  $d^\star \equiv  (-1)^{dp+1} \star d\, \star$.
In addition, the operator $\star d$ satisfies the commutation relation $[\star d,\mathsterling_{v}]=0$
if and only if $[\star,\mathsterling_{v}]=0$ since $[d,\mathsterling_{v}]=0$. Generally the
Lie derivative and the Hodge star operator do not commute. When these operators do commute, i.e. when the commutation relation
$[ \mathsterling_{v},\star]=0$ holds we can extend Theorem~\ref{th:GCI} to forms which are the Hodge duals of exact forms.
An example of such commutation relation is when the vector field $v$ generates an isometry (see Appendix~\ref{ssec:CICF:HSECD}).
The third example is when the $p$-form $\gamma$ is co-exact, i.e. $\gamma=d^\star \beta$, with $\beta$ a $(p+1)$-form. Setting
${\rm Op\,}\equiv (-1)^{dp+1} \star d$, we fall in the case of the second example with $\pi=\star \beta \in \Lambda^{d-p-1}(\Omega)$. 
Of course, other interesting examples can be constructed.

\subsubsection{A Lagrangian Biot--Savart problem}
\label{ss:ALBSP}
So far, the Lie-advected $p$-form $\gamma$ was just assumed to be 
expressible as the exterior derivative $d\alpha$ of a $(p-1)$-form $\alpha$. As we
shall now see, the generalised Cauchy invariants equation \eqref{eqn:GCI},
allows an inversion, which can be viewed as solving a Biot--Savart problem in Lagrangian variables:
the corollary hereafter  gives an explicit expression for the  
$(p-1)$-form $\alpha$, in which we use the  notation
\[
\Delta_a = \sum_{i=1}^d {\partial}_{a^i}^{\,2},
\]
for the Laplacian in Lagrangian variables and $\Delta_a^{-1}$ for its formal inverse.

\begin{corollary} (A Lagrangian Biot--Savart problem).
\label{Cor:inversionGCI} 
Under assumptions of Theorem~\ref{th:GCI}, the generalised Cauchy invariants
equation \eqref{eqn:GCI} leads to

\begin{equation}
\label{eqn:inversionGCI}
\alpha_{i_1 \ldots i_{p-1}} = \delta^{k\ell} \Delta_a^{-1} \frac{\partial }{\partial a^k}
\left (
\gamma_{0\ell j_1\ldots j_{p-1}} \frac{\partial a^{j_1}}{\partial x^{i_1}} \ldots \frac{\partial a^{j_{p-1}}}{\partial x^{i_{p-1}}} 
\right), \quad  1\leq i_1 < \ldots < i_{p-1} \leq d.
\end{equation}  
\end{corollary}

\begin{proof}
The  generalised Cauchy invariants equation \eqref{eqn:GCI} gives componentwise
\[
\frac{1}{(p-1)!} \delta_{j_1\ldots j_{p-1}}^{i_1 \ldots i_{p-1}} \frac{\partial \alpha_{i_1 \ldots i_{p-1}}}{\partial a^\ell}
\frac{\partial x^{j_1}}{\partial a^{l_1}} \ldots \frac{\partial x^{j_{p-1}}}{\partial a^{l_{p-1}}}
= \gamma_{0\ell l_1\ldots l_{p-1}}. 
\]
Multiplying by $p-1$ suitably chosen inverse Jacobian matrices, we obtain
\[
\frac{1}{(p-1)!} \delta_{j_1\ldots j_{p-1}}^{i_1 \ldots i_{p-1}} \frac{\partial \alpha_{i_1 \ldots i_{p-1}}}{\partial a^\ell}
\frac{\partial x^{j_1}}{\partial a^{l_1}}\frac{\partial a^{l_1}}{\partial x^{k_1}} 
\ldots \frac{\partial x^{j_{p-1}}}{\partial a^{l_{p-1}}}\frac{\partial a^{l_{p-1}}}{\partial x^{k_{p-1}}} 
= \gamma_{0\ell l_1\ldots l_{p-1}}
\frac{\partial a^{l_1}}{\partial x^{k_1}} \ldots \frac{\partial a^{l_{p-1}}}{\partial x^{k_{p-1}}},
\]
that is
\begin{eqnarray}
\gamma_{0\ell l_1\ldots l_{p-1}}
\frac{\partial a^{l_1}}{\partial x^{k_1}} \ldots \frac{\partial a^{l_{p-1}}}{\partial x^{k_{p-1}}}
&=& \frac{1}{(p-1)!} \delta_{j_1\ldots j_{p-1}}^{i_1 \ldots i_{p-1}}  \delta_{k_1}^{j_1} \ldots \delta_{k_{p-1}}^{j_{p-1}}
\frac{\partial \alpha_{i_1 \ldots i_{p-1}}}{\partial a^\ell} \nonumber \\
&=& \frac{1}{(p-1)!} \delta_{k_1\ldots k_{p-1}}^{i_1 \ldots i_{p-1}} 
\frac{\partial \alpha_{i_1 \ldots i_{p-1}}}{\partial a^\ell}.
\label{eqn:invGCI_1}
\end{eqnarray}  
Since $\alpha_{i_1 \ldots i_{p-1}}$ is skew-symmetric, we have
\[
\frac{1}{(p-1)!} \delta_{k_1\ldots k_{p-1}}^{i_1 \ldots i_{p-1}} 
 \alpha_{i_1 \ldots i_{p-1}} =\alpha_{k_1 \ldots k_{p-1}},  
\]
and \eqref{eqn:invGCI_1} becomes
\begin{equation}
\label{eqn:invGCI_2}
\frac{\partial \alpha_{i_1 \ldots i_{p-1}}}{\partial a^\ell}=
\gamma_{0\ell j_1\ldots j_{p-1}}
\frac{\partial a^{j_1}}{\partial x^{i_1}} \ldots \frac{\partial a^{j_{p-1}}}{\partial x^{i_{p-1}}}.
\end{equation}  
By application of the differential operator $\delta^{k\ell} (\partial/\partial a^k)$ to \eqref{eqn:invGCI_2}
and summation over index $\ell$,  \eqref{eqn:invGCI_2} becomes
\[
\Delta_a \alpha_{i_1 \ldots i_{p-1}}=
\delta^{k\ell}  \frac{\partial }{\partial a^k}
\left (
\gamma_{0\ell j_1\ldots j_{p-1}} \frac{\partial a^{j_1}}{\partial x^{i_1}} \ldots \frac{\partial a^{j_{p-1}}}{\partial x^{i_{p-1}}} 
\right).
\]
This equation gives \eqref{eqn:inversionGCI} after formal inversion of the Laplacian operator $\Delta_a$,
expressed in Lagrangian variables.
We observe that this inversion is reminiscent of that of the
Biot--Savart law, with the left-hand side  of \eqref{eqn:invGCI_2}
playing roughly the role of the curl of the $(p-1)$-form $\alpha$.
\end{proof}  
    
\subsection{Broad applicability of Theorem~\ref{th:GCI}}
\label{ss:broad}

Our key result, namely Theorem~\ref{th:GCI}, may be viewed as a new fondamental
result in \textit{linear transport theory},  giving an alternative Lagrangian formulation
of Lie advection for a large class of differential forms.${}^{\mbox{\scriptsize{\footnotemark[1]}}}$
\footnotetext[1]{This was pointed out to us by Peter Constantin who made 
us realize that this new result in linear transport theory may be of
independent interest.}
Indeed there is no need to have a selfconsistent coupling between the transporter (vector fields $v$) 
and the transported (differential forms $\gamma$) to obtain generalised Cauchy invariants equations.
For the first time, it is  here shown that Cauchy invariants equations exist for non-selfconsistent 
linear transport.
It must be pointed out that, when the Cauchy
invariants were rediscovered in the 20th century, most of the time it
was by making use of Noether's theorem 
in the case of selfconsistent nonlinear equations
\citep{FV14}. Although
Noether's theorem is usually  not available for linear transport 
equations, our key result shows that such generalised Cauchy invariants still do exist
in linear transport theory. Consequently, our result is applicable to 
a large class of  fluid dynamical equations that rely on Lie advection. Hereafter,
we give some important examples. 
Some more material, dealing specifically
with helicity problems in fluids and MHD, will be
presented in Section~\ref{sec:AHMHD}.

\subsubsection{Induction equation in ideal incompressible MHD}
\label{ex:B}
In incompressible ideal MHD, the magnetic flux conservation law
(induction or Faraday's  equation) can be rewritten as a Lie advection
equation,
provided the magnetic  field is considered as a $2$-form   \citep[see,
  e.g.,][]{Fla63}.  Denoting the magnetic field $2$-form by $B$ and the 
magnetic (vector) potential $1$-form by $A$, we have
\begin{equation}
  \label{B-2form}
  B=dA,
\end{equation}
and the induction equation reads 
\begin{equation}
  \label{B-2form-eq}
  \partial_t B + \mathsterling_v B =0. 
\end{equation}
Indeed \eqref{B-2form-eq} results from the Maxwell-Faraday equation $\partial_t B + dE=0$, the Maxwell-Gauss equation $dB=0$, and
the (ideal) induction equation $E-{\rm i}_vB=0$, where $E$ is the dual $1$-form associated to the electric (vector) field.
Therefore from Theorem~\ref{th:GCI}, we obtain the  following Cauchy invariants equation
\begin{equation}
\label{cauchymhd}
dA_k\wedge dx^k = B_0=dA_0.
\end{equation}
Let us note that this equation and \eqref{B-2form-eq} can be extended to Riemannian manifolds of any dimension
by keeping the same covariant form, i.e. as they stand.

We observe that \eqref{B-2form} and \eqref{B-2form-eq} are known, at
least  for the  3D flat case \citep{Fla63}. 
As to \eqref{cauchymhd}, in the
flat case, it is the well-known law of conservation of magnetic flux, which
is here shown to be a Cauchy-type equation.

\subsubsection{Adiabatic and barotropic ideal compressible fluid}
\label{ss:ACF}

Here and in Section~\ref{ss:BCMHD} we use geometrical tools for writing fluid
equations that will be discussed in more details in Section~\ref{sec:VIFIH}.

An adiabatic ideal compressible fluid, with equation of state $p=p(\rho,\eta)$,
where the scalar $\rho$ and $\eta$ are respectively the density and the entropy, is governed
by the equations
\begin{eqnarray}
\partial_t v^\flat +\mathsterling_v v^\flat  &=& -\frac{dp}{\rho} + \frac{1}{2}d(v,v)_g \label{ACFv}\\
\partial_t \mathfrak{m} + \mathsterling_v \mathfrak{m} &=& 0 \label{ACFm}\\
\partial_t \eta + \mathsterling_v \eta &=& 0. \label{ACFe}
\end{eqnarray}  
Here, $\mathfrak{m}$ denotes the mass $d$-form defined by $\mathfrak{m}:=\rho\mu$. Since by definition
we have ${\rm div}_\mu v:= \mathsterling_v \mu$, \eqref{ACFm} is equivalent to $\partial_t \rho + {\rm div}_\mu (\rho v)=0$.
The Lagrangian formulation of \eqref{ACFm}-\eqref{ACFe} is
\[
\rho_t \circ \varphi_t = {\rho_0}/{J_\mu(\varphi_t)}, \quad \mbox{and} \quad \eta_t \circ \varphi_t = {\eta_0}, 
\]
where $J_\mu(\varphi_t):=\varphi_t^\ast\mu/\mu=\sqrt{g\circ \varphi_t}{\rm det}(\partial \varphi_t/\partial a)$ is the Jacobian of the
Lagrangian flow $\varphi_t$ generated from the vector field $v$, and $\rho_0= \rho_0(a)$ and $\eta_0=
\eta_0(a)$ are the   initial density and
entropy. We now introduce the $1$-form $\gamma$, with zero initial value (i.e. $\gamma_0=0$), which satisfies
the equation
\begin{equation}
  \label{gauge_acf}
  \partial_t \gamma +\mathsterling_v \gamma  =-\frac{dp}{\rho}.
\end{equation}
Using the Lie derivative theorem \eqref{LDT}, integration of \eqref{gauge_acf} yields the $1$-form $\gamma$ such that
\[
\gamma = -\varphi_{\ast t}\int_0^td\tau \varphi_\tau^\ast\Big(\frac{dp}{\rho}\Big).
\]
Defining the modified $1$-form velocity $\tilde{v}^\flat := {v}^\flat -\gamma$, and
the modified $2$-form vorticity $\tilde{\omega}:=d\tilde{v}^\flat$, from  \eqref{ACFv} and \eqref{gauge_acf},
we obtain
\[
\partial_t \tilde{\omega} + \mathsterling_v \tilde{\omega}=0.
\]
We can now apply Theorem~\ref{th:GCI} to this equation. We then obtain for \eqref{ACFv} the following
Lagrangian formulation
\[
d \tilde{v}_k^\flat\wedge dx^k = \omega_0:=dv_0^\flat.
\]
Let us note that the Ertel potential vorticity $3$-form $dv^\flat\wedge d\eta$ is a  Lagrangian invariant since
$(\partial_t + \mathsterling_v) dv^\flat\wedge d\eta=0$, which results from \eqref{ACFv}, \eqref{ACFe}
and the identity $dp\wedge d\rho\wedge d\eta=0$ by virtue of the dependence $p=p(\rho,\eta)$. In three dimension, $d=3$,
we can easily show that the scalar local Ertel potential vorticity $\star (dv^\flat\wedge d\eta)$ satisfies also a
Lie-advection equation; thus it is also a local conserved quantity. Let us
also note that in the barotropic case \citep{KC89},
since $p=p(\rho)$, we obtain $d(dp/\rho)=0$; thus we have $\gamma=0$, $\tilde{v}^\flat={v}^\flat$ and $\tilde{\omega}={\omega}:=dv^\flat$.

\subsubsection{Barotropic ideal compressible MHD}
\label{ss:BCMHD} Let $b$ be the magnetic vector field, $b^\flat$ its dual $1$-form and $B$ its dual $2$-form.
For an example of a detailed derivation of MHD models we refere to \citet{GP04}.
The barotropic ideal compressible MHD, in a coordinate-free form, reads
\begin{eqnarray}
\partial_t \mathfrak{m} + \mathsterling_v \mathfrak{m} &=& 0 \label{BCMHDm}\\
\partial_t v^\flat +\mathsterling_v v^\flat  &=& \frac{\mathsterling_b b^\flat - d(b,b)_g}{\rho} - d\Big(h- \frac{1}{2}d(v,v)_g\Big)
\label{BCMHDv}\\
\partial_t B+\mathsterling_v B&=&0.\label{BCMHDB}
\end{eqnarray}  
Here, the barotropic equation of state $p=p(\rho)$ is used, and the enthalpy $h$ is related to the pressure $p$ via the
relation $dh=dp/\rho$. In \eqref{BCMHDv}, the term $(\mathsterling_b b^\flat - d(b,b)_g)/\rho$ is the dual $1$-form 
of the Lorentz force field. It is obtained from the Amp\`ere law, $d\star B=\mu_0\, \mathfrak{j}$, where
$\mathfrak{j}$ is the current form while $j=(\star\, \mathfrak{j})^\sharp=(\star\, d\star B )^\sharp/\mu_0$ is the current-density vector
field. Indeed, in the three-dimensional case $d=3$, this $1$-form can be expressed as $-{{\rm i}_{j} B}/{\rho}$
which is the dual $1$-form of the vector field $j\times b/\rho$ (Lorentz force) where the current density vector
$j$ is related to the magnetic (vector) field $b$ by the Amp\`ere law
$\mu_0j=\nabla \times b$ (the displacement current being neglected).
In the three-dimensional case $d=3$, let us note that using the relations $B={\rm i}_{b/\rho}\rho \mu$ and
$[\mathsterling_{v},{\rm i}_{b/\rho}] ={\rm i}_{[v,b/\rho]}$ (see Appendix~\ref{ssec:CICF:EDIP}), equation \eqref{BCMHDB} is equivalent to
$\partial_t (b/\rho)+\mathsterling_v (b/\rho)=0$.  We now introduce the $1$-form $\gamma$,
with zero initial value (i.e. $\gamma_0=0$), which satisfies
the equation
\begin{equation}
  \label{gauge_bcmhd}
  \partial_t \gamma +\mathsterling_v \gamma  = \frac{\mathsterling_b b^\flat - d(b,b)_g}{\rho}. 
\end{equation}
Using the Lie derivative theorem \eqref{LDT}, integration of \eqref{gauge_bcmhd} yields the $1$-form $\gamma$ such that
\[
\gamma = \varphi_{\ast t}\int_0^td\tau \varphi_\tau^\ast\Big(\frac{\mathsterling_b b^\flat - d(b,b)_g}{\rho}\Big),
\]
where $\varphi_t$ is the Lagrangian flow generated from the vector field $v$.
Defining the modified $1$-form velocity $\tilde{v}^\flat := {v}^\flat -\gamma$, and
the modified $2$-form vorticity $\tilde{\omega}:=d\tilde{v}^\flat$, from  \eqref{BCMHDv} and \eqref{gauge_bcmhd},
we obtain
\begin{equation}
  \label{vor-2-form_bcmhd}
\partial_t \tilde{\omega} + \mathsterling_v \tilde{\omega}=0.
\end{equation}
Therefore, we can again apply Theorem~\ref{th:GCI}  to \eqref{BCMHDB} and \eqref{vor-2-form_bcmhd}. We then obtain
for the system \eqref{BCMHDv}-\eqref{BCMHDB} the following
Lagrangian formulation
\[
d \tilde{v}_k^\flat\wedge dx^k = \omega_0:=dv_0^\flat, \quad \mbox{and} \quad d {A}_k\wedge dx^k = B_0:=dA_0.
\]
Of course, the Lagrangian formulation of 
the equation of mass conservation \eqref{BCMHDm} is the same as in Sec.~\ref{ss:ACF}.
Let us note that we can extend this formulation to adiabatic ideal compressible MHD with the equation of state
$p=p(\rho,\eta)$ by adding to equations \eqref{BCMHDm}-\eqref{BCMHDB} the entropy equation \eqref{ACFe}.
Let us also note that in fact there are several ways in which the full nonlinear ideal MHD
equations can be recast as Lie-advection problems: for example one can use the dual $1$-forms
of the \citet{Els56} variables \citep{MM87}.

\subsubsection{Extended ideal compressible MHD}
\label{ss:ECMHD}
The extended MHD equations \citep{GP04, DML16, LMM16}, in covariant form, reads
\begin{eqnarray}
\partial_t \mathfrak{m} + \mathsterling_v \mathfrak{m} &=& 0 \label{ECMHDm}\\
\partial_t B_\pm+\mathsterling_{v_\pm} B_\pm&=&0,\label{ECMHDB}
\end{eqnarray}  
where, $B_\pm=dA_\pm$, $A_\pm=A+(d_e^2/\rho) \star d \star B+ \kappa_\pm v^\flat$,
and $v_\pm=v-\kappa_{\mp}\,(\star\, d \star B)^\sharp /\rho$. Here, the constants $\kappa_\pm$ are the solutions of the quadratic equation
$\kappa^2-d_i\kappa -d_e^2=0$, where $d_i$ and $d_e$ serve as the normalized ion and electron skin depths, respectively.
In addition the variables $\mathfrak{m}$ and $v$ denote the total-mass form and the center-of-mass velocity vector, respectively.
As in Section~\ref{ex:B} the magnetic potential $1$-form by $A$ is linked to the magnetic field $2$-form $B$
by $B=dA$. Let us note that here the assumption of a barotropic equation of state was used.
We can directly apply Theorem~\ref{th:GCI}  to \eqref{ECMHDB} for
obtaining the following Cauchy invariants equations
\[
d A_{\pm\, k}\wedge dx_{\pm}^k = B_{\pm\,0}:=dA_{\pm\, 0},
\]
where $x_{\pm\, t}$ are the Lagrangian maps generated by the vector fields $v_\pm$.
Once again, the Lagrangian formulation of
the equation of mass conservation \eqref{ECMHDm} is the same as in Sec.~\ref{ss:ACF}.
When $d_e\rightarrow 0$, we have $\kappa_\pm=d_i$ and we obtain what is called  Hall MHD.
When $d_i\rightarrow 0$, we have $\kappa_\pm=\pm d_e$ and we obtain what is called inertial MHD. Let us
note that when  $d_i\rightarrow 0$ and  $d_e\rightarrow 0$ simultaneously, we
obtain $\kappa_\pm=0$ and thus we do not recover
the full ideal compressible MHD, since both equations \eqref{ECMHDB}
degenerate into only one equation, namely 
\eqref{BCMHDB}.

\subsubsection{Tao's modification of the incompressible Euler equations in Euclidean space}
\label{ss:TaoEuler}
The dynamics of vorticity for the case of the ordinary incompressible Euler
equation will be discussed in detail in  Section~\ref{sec:VIFIH}, but we wish
to mention that recently \citet{Tao16} has proposed an interesting
modification of the incompressible Euler equations in Euclidean spaces that
preserves much of its differential geometric content, but sometimes allows
(proven) blowup, that is loss of regularity in a finite time.
This modfication consists in keeping the Lie-advection equation for the vorticity $2$-form
$\omega$, namely $(\partial_t +\mathsterling_{v})\, \omega=0$, but replacing the Biot--Savart law  $v^\flat=d^\star \Delta_{\rm H}^{-1}\omega$ by
the following selfconsistent coupling  $v^\flat=d^\star \!A\omega$. Here, $A$ is a linear
pseudodifferential operator which is self-adjoint (like $\Delta_{\rm H}^{-1}$) and has
the same degree of regularity as $\Delta_{\rm H}^{-1}$. \citet{Tao16} has shown that there exist some operators $A$
for which the corresponding classical solutions blow up in finite time.
Since the Lie-advection equation for the vorticity $2$-form is preserved in
these models,
by Theorem~\ref{th:GCI}, there is a
corresponding generalised Cauchy invariants equation. Indeed, since $\omega=dv^\flat$
and using the modified velocity $1$-form $u^\flat:= d^\star \!A d v^\flat$, we can
now define
two Lagrangian maps $x_t$ and $y_t$ by
\[
\dot{x}_t:=\frac{d x_t}{dt} = u(t,x_t) \quad \mbox{and} \quad \dot{y}_t:=\frac{d y_t}{dt} = v(t,y_t),
\]
where the vector fields $u$ and $v$ are linked by the relation $u=(d^\star \!A d v^\flat)^\sharp$.
Recalling that in Euclidean spaces covariant and contravariant components are identical,
the corresponding Cauchy invariants equation then reads
\begin{equation}
  \label{CIE_Tao}
d \dot{y}_k \wedge d x^k = \omega_0. 
\end{equation}

\begin{remark}{\rm
(\textit{Well-posedness: linear and nonlinear issues}).
\label{rem:non-vs-self-consistent}
As mentioned at the beginning of Section~\ref{ss:broad}, the coupling
between the $p$-form  $\gamma$  and the vector field $v$,
in the Lie-advection equation \eqref{frozen}, could be either
non-selfconsistent or selfconsistent. In the former case, also called
passive,  $v$ is prescribed at all times and  there is no feedback  of $\gamma$ on $v$. In
the latter case, $v$ is not prescribed (except perhaps at the initial time) and
the  feedback of $\gamma$ on $v$ is  given by at
least one additional equation linking  $v$ to $\gamma$;  an instance is the
full  Euler equation, where the vorticity $2$-form is the exterior
derivative of the velocity $1$-form (cf. Section~\ref{sec:VIFIH}).
    
In the non-selfconsistent case, when the vector field $v$ is Lipschitz
continuous (not necessarily divergence-free or incompressible),
the associated Lagrangian flow exists globally in time \citep{Tay96}.
Therefore,   \eqref{frozen} is well posed and
has global-in-time regular solutions; thus Lie and Cauchy invariants exist
globally in time too \citep{Tay96}.

In the selfconsistent case, well-posedness of the coupled system, i.e.
existence of solutions to the system
constituted of \eqref{frozen} plus the additional equation
linking  $v$ to $\gamma$, depends of course on the specific
selfconsistent coupling considered.

For example, in the case where the vector field $v$ is the velocity field
given by the 3D-Euclidean incompressible Euler equations  and
the $p$-form $\omega$ is the $2$-vorticity form, the selfconsistent
coupling is given by $\omega=dv^\flat$ (in the simplest case this
means that the vorticity vector is the curl of the velocity vector).
Using the Biot--Savart law, this selfconsistent coupling can be rewritten 
as $v^\flat=d^\star \Delta_{\rm H}^{-1}\omega$, where $d^\star$ is the exterior
co-derivative and $\Delta_{\rm H}:=dd^\star+d^\star d$ is the Laplace-de Rham
operator (see Appendix~\ref{ssec:CICF:HCHD}).
The corresponding Cauchy problem is known to be well posed in time when
the initial velocity is in H\"older or Sobolev spaces
with suitable indexes of regularity. This was established in the seminal
work of \citet{Lic25, Lic27}
and \citet{Gun26, Gun34} for the case of the whole Euclidean space and, of
\citet{EM70} for the case of bounded domains.
Therefore \eqref{frozen} has local-in-time regular
solutions, so that Lie and Cauchy invariants exist at least
for short times. 

Although the modified Euler equations of \citet{Tao16} satisfy helicity and energy (or
Hamiltonian) conservation laws and possess a Kelvin circulation theorem,
Tao has shown that there exist some operators $A$ for which the
corresponding classical solutions blow up in finite time.  It does not mean
that we can conjecture a finite-time blow-up for classical solutions of the
original incompressible Euler equations (for $d\geq3$), but rather that a
possible absence of blow-up cannot be proved with the only properties of the
Euler equations that are shared by these modified models. Although
Lie-advection equation for the vorticity $2$-form is preserved in these
models, the Cauchy invariants equation \eqref{CIE_Tao} shows that a
modification of the Biot--Savart law induces a change in the geometry of the
original incompressible Euler equations. Indeed \eqref{CIE_Tao} involves two
families of characteristic curves, whereas the original incompressible Euler
equations deal with only one such family. In other words,
on the set of incompressible vector fields we have
$d^\star \Delta_{\rm  H}^{-1}d={\rm Id}$, whereas $d^\star A d\neq{\rm Id}$.
}\end{remark}

\section{Vorticity and incompressible flow in hydrodynamics}
\label{sec:VIFIH}
In this section we apply our main result, Theorem~\ref{th:GCI}, to the incompressible
Euler equations on a $d$-dimensional Riemannian manifold.
This will extend to Riemannian manifolds of any dimension the notion
of Cauchy invariants,  
first introduced by \citet{Cau27} for the three-dimensional incompressible 
Euler equations in ``flat'' Euclidean spaces.
First, we need to write the Euler equations in a covariant form, i.e. in terms of a $1$-form
$v^\flat$ for the velocity vector field $v$; this is the aim of Sec.~\ref{ssec:CFVE}.
The velocity  $1$-form $v^\flat$ is here called the infinitesimal velocity
circulation, 
because if we were  in a flat space, we would have $v^\flat=\vec{v}\cdot \overrightarrow{dx}$.
Henceforth, ordinary vectors will be surmounted by an arrow when they might otherwise
be mistaken for $p$-forms. Then, the exterior derivative of the
covariant form of the Euler equations gives a Lie-advection
equation of the form \eqref{frozen} for the vorticity $2$-form $\omega$,
here called the covariant vorticity equation.
Henceforth, $\omega$ always denotes the vorticity $2$-form and not the vorticity vector;
the latter being $ \vec{\omega}$. In Sec.~\ref{sec:VarRepre}, applying
Theorem~\ref{th:GCI} to the covariant vorticity equation,
we show that the Cauchy invariants equation can have different representations.
In particular we show that the Cauchy invariants equation is an alternative formulation of
the well-known Lie advection of the vorticity $2$-form.  
From this point of view,  the Cauchy invariant and the 
Cauchy vorticity formula are representations of the same conservation law, 
related by Hodge duality. Finally, we note that
the covariant vorticity equation and the Cauchy invariants equation
on a manifold
have alternative derivations using variational methods in conjunction with the relabelling
symmetry and Noether's theorem.

\subsection{Covariant formulation of the vorticity equation}
\label{ssec:CFVE}
In this section the vorticity will be  considered as a $2$-form
$\omega$. We start with the incompressible Euler equations on a
$d$-dimensional Riemannian manifold $(M,g)$. Written in terms of the velocity
vector field $v$ and of the scalar pressure field $p$, they read
\begin{equation}
\label{euler_eq_contravariantFOR}
\partial_t v^i+ v^k\nabla_kv^i= -g^{ik}\partial_k p\ \ (\mbox{Euler}),  \quad
\nabla_iv^i=0\ \ (\mbox{incompressibility condition})\quad x \in\Omega, \ t\in]0,T]. 
\end{equation}
Here the symbol  $\nabla_k$ denotes the covariant derivative, which
can be seen as a generalisation to curved spaces of the classical
partial derivative $\partial_k$ of Euclidean spaces (for a
more detailed definition, see Appendix~\ref{ssec:CICF:RCCD}).
The geometric interpretation of the incompressible Euler equations is recalled in Appendix~\ref{sec:geodiffeuler},
while their simplest derivation is obtained from a variational
formulation (least action principle), as explained in Appendix~\ref{sec:eulerFLAP}.

The Euler equations and incompressibility condition, written in the
contravariant formulation \eqref{euler_eq_contravariantFOR}, can be
rewritten in the covariant formulation, i.e. in term of $1$-form
fields instead of vector fields.  Let $v^\flat$ be the $1$-form field
obtained from the vector field by the index lowering operator
$(\cdot)^\flat: \mathcal{T}_0^1(\Omega)\rightarrow
\mathcal{T}_1^0(\Omega)$; that is, we set $v^\flat =
(v^i\partial_i)^\flat=(v^\flat)_idx^i=g_{ij}v^jdx^i$. Using  the
preservation of the metric of the Riemann--Levi-Civita connection, namely $\nabla_k g_{ij}=0$,
we easily find
\begin{equation}
\label{euler_eq_covar}
\partial_t v_i+ v^k\nabla_kv_i= -\partial_i p,  \quad x \in\Omega\subset M, \ t\in]0,T], 
\end{equation}
and
\begin{equation}
\label{incomp_covar}
g^{ij}\nabla_iv_j= 0, \quad x \in\Omega, \ t\in]0,T]. 
\end{equation}
In compact form,  \eqref{euler_eq_covar} can be written as
\begin{equation}
\label{euler_eq_covar_compact}
\partial_t v^\flat+ (\nabla_{v}v)^\flat= -dp,  \quad x \in\Omega, \ t\in]0,T]. 
\end{equation}
Now we rewrite \eqref{euler_eq_covar_compact} as
\begin{equation}
\label{euler_eq_oneform}
\partial_t v^\flat + \mathsterling_{v} v^\flat + d\Big(p-\frac12 (v,v)_g \Big)= 0,  
\quad x \in\Omega, \ t\in]0,T], 
\end{equation}   
an equation which
differs from  a Lie advection condition for $v^\flat$ by just an
additional exact differential (which will disappear upon application
of yet another exterior differential). To obtain
\eqref{euler_eq_oneform} 
we use the Cartan formula $\mathsterling_{v}
v^\flat  = d\hskip 0.6pt{\rm i}_{v} v^\flat + {\rm i}_{v} d v^\flat$
for the Lie derivative (Appendix~\ref{ssec:CICF:EDIP}) and a rewrite
of the right-hand side of the Cartan formula, precisely
$d{\rm i}_vv^\flat +{\rm i}_vdv^\flat = (\nabla_v v)^\flat + \frac12 d (v,v)_g$. This is
established in  Appendix~\ref{appendixCF}.
In these equations  ${\rm i}_{v}:\Lambda^p(\Omega)\rightarrow \Lambda^{p-1}(\Omega)$ is the interior (or inner) product
with the vector $v$, which acts as an integration; as to
$(\cdot,\cdot)_g$,  it denotes the Riemannian scalar product for
vector fields in
$\mathcal{T}_0^1(\Omega)$,  defined by 
$(v,w)_g = g_{ij}v^iw^j$ (see Appendix~\ref{ssec:CICF:RM}).

From a fluid-mechanical point of view, specialising to the Euclidean case, it is of interest to rewrite the Euler
equations \eqref{euler_eq_oneform} in standard vector notation as
\begin{equation}
\label{notlamb}
\partial_t \vec{v} + \overrightarrow{\nabla}
|\vec{v}|^2 + \vec{\omega}\times
\vec{v} +\overrightarrow{\nabla} \left(p-\frac{|\vec{v}|^2}{2}\right) =0,
\end{equation}
where $\times$ denotes the vector product and $\overrightarrow{\nabla}$ is the standard
gradient operator in Euclidean coordinate.
This has some similarity to what
is known as Lamb's form of the incompressible Euler equations, in which 
$\vec{\omega}\times \vec{v}$ also appears. It would not
be advisable to simplify \eqref{notlamb} to Lamb's form by combining the two terms involving a gradient of the 
local kinetic energy, because the second and third term on the left-hand side
of \eqref{notlamb} are both needed to obtain a Lie derivative and all the
nice consequences.

Indeed, we can now define the vorticity $2$-form as the exterior derivative of the 
infinitesimal velocity circulation
$1$-form $v^\flat$, that is
\begin{equation}
\label{def_vorticity}
\omega = d  v^\flat.
\end{equation}
Taking the exterior derivative of the covariant formulation \eqref{euler_eq_oneform}
of the Euler equations, and using the commutation relation $[d,\mathsterling_v]=0$, we obtain 
\begin{equation} 
\label{vort-2-form-eq}
\partial_t\omega + \mathsterling_{v} \omega= 0.
\end{equation}
This establishes  that the vorticity $2$-form is Lie advected, a
result essentially known since \citet{Hel58}.
In terms of the $1$-form $v^\flat$, the incompressibility condition $\nabla_i v^i=0$ reads $d^\star v^\flat=0$ (see Appendix~\ref{ssec:CICF:HSECD}).
Using the Hodge theorem (see Appendix~\ref{ssec:CICF:HCHD}), we obtain the
Biot--Savart law $v=(d^\star\Delta_{\rm H}^{-1}\omega)^\sharp$, which selfconsistently expresses
the velocity vector field $v$ in terms of the vorticity $2$-form $\omega$. Indeed, using
the incompressibility condition  $d^\star v^\flat=0$, we have
$v=(d^\star\Delta_{\rm H}^{-1}\omega)^\sharp=(\Delta_{\rm H}^{-1}d^\star dv^\flat)^\sharp=(\Delta_{\rm H}^{-1}(d^\star d + d d^\star)v^\flat)^\sharp=(v^\flat)^\sharp=v$.

Finally from \eqref{vort-2-form-eq}, using the lesser known commutation relation $[\sharp^{d-p}\,\star, \mathsterling_v]=0$ (with $p=2$),
whose proof is given in Appendix~\ref{appendixCR},
we obtain that the vorticity vector is also Lie-advected. Here, by vorticity vector, we understand the $(d-2)$-vector $\vec{\omega}:=(\star \omega)^{\sharp^{d-2}}$
(in other words a $(d-2)$-contravariant tensor). Namely, we have
\[
\partial_t\vec{\omega} + \mathsterling_{v} \vec{\omega}= 0.
\]

\begin{remark}{ \rm
    \begin{itemize}
    \item[1.] An alternative derivation of the covariant vorticity equation
      \eqref{vort-2-form-eq}
      from the Euler equations \eqref{euler_eq_contravariantFOR} is to use
      the  relabelling symmetry
      and Noether's theorem (see  Appendix~\ref{sec:CV2FFNT}). This
      derivation leads to
      \begin{equation}
        \label{ConservationOfTwoVortictyForm}
        \varphi_t^\ast \omega =  \omega_0,
      \end{equation}
      from where Lie-advection of the vorticity $\omega$ follows readily
      (see \eqref{Lie-pullback} and \eqref{frozen}). 
    \item[2.] Let us note that in Appendix~B of \citet{GV16}, the authors give a variational derivation
      of the covariant Euler equations \eqref{euler_eq_oneform}.
    \item[3.] In Appendix~\ref{appendixCR}, the proof of the commutation relation $[\sharp^{d-p}\,\star, \mathsterling_v]=0$
      is done by following an algebraic approach. A dynamical approach based on infinitesimal pullback transport and the
      Lie-derivative theorem  could be used for an alternative proof, along
      the
      lines used in Appendix~\ref{ssec:CICF:HSECD}.
    \end{itemize}
}\end{remark}

\subsection{Cauchy invariants equation and Cauchy formula}
\label{sec:VarRepre}
We are now ready to extend to Riemannian manifolds of any dimension
the  Cauchy invariants equation
and the Cauchy formula.
We begin by observing that  all assumptions
of Theorem~\ref{th:GCI} are now satisfied: the vorticity $2$-form $\omega=dv^\flat$ is
exact and is Lie-advected. Therefrom follows 
Corollary~\ref{th:InvVorticity} for which we also give a direct simplified proof.
\begin{corollary}(Cauchy invariants equations on a Riemannian manifold). 
\label{th:InvVorticity}
Let $\varphi_t$ be the Euler flow. We set $x=\varphi_t$ and $v=\dot{\varphi}_t$, with
$v_0=\dot{\varphi}_0$. Then we have
\begin{equation}
\label{DefEqInvCau}
dv_k\wedge dx^k = \varphi_t^\ast\omega =\omega_0:= dv_0^\flat.
\end{equation}  
\end{corollary}  

\begin{proof}
We begin by showing that $\varphi_t^\ast\omega=dv_k\wedge dx^k$. 
Indeed,  we have
\begin{eqnarray*}
\varphi_t^\ast\omega&=&(\varphi_t^\ast\omega)_{ij}da^i\wedge da^j= 
\sum_{i<j}\frac{\partial x^l}{\partial a^i}\frac{\partial x^k}{\partial a^j}\omega_{lk}(x_t(a))da^i\wedge da^j\\
&=&\sum_{i<j}\frac{\partial x^l}{\partial a^i}\frac{\partial x^k}{\partial a^j}\left(
\frac{\partial v_k}{\partial x^l}-\frac{\partial v_l}{\partial x^k}\right)da^i\wedge da^j=
\sum_{i<j}\frac{\partial v_k}{\partial x^l}\left(\frac{\partial x^l}{\partial a^i}\frac{\partial x^k}{\partial a^j}
-\frac{\partial x^l}{\partial a^j}\frac{\partial x^k}{\partial a^i}\right)da^i\wedge da^j\\
&=&  \frac{\partial v_k}{\partial x^l}\frac{\partial x^l}{\partial a^i}\frac{\partial x^k}{\partial a^j}da^i\wedge da^j  
  = \frac{\partial v_k}{\partial a^i}\frac{\partial x^k}{\partial a^j}da^i\wedge da^j\\
&=&    dv_k\wedge dx^k. 
\end{eqnarray*}  
Corollary~\ref{th:InvVorticity} follows from the covariant
vorticity equation \eqref{vort-2-form-eq} or the conservation of
the vorticity $2$-form, i.e. $\varphi_t^\ast\omega=\omega_0:= dv_0^\flat$.
\end{proof}

\begin{remark}{\rm
\begin{itemize}
\item[1.]  (\textit{Contravariant formulation}).
  In terms of components, the Cauchy invariants equation \eqref{DefEqInvCau} reads
  \[
  \partial_{k}(\dot{x}^ig_{ij})\partial_l x^j -  \partial_{l}(\dot{x}^ig_{ij})\partial_k x^j= \omega_{0kl}, \quad  1\leq k < l \leq d.
  \]
  The contravariant form of this equation reads
  \[
  \varepsilon^{kli_1\ldots i_{d-2}} \partial_{k}(\dot{x}^ig_{ij})\partial_l x^j=\sqrt{g} \omega_0^{i_1\ldots i_{d-2}}, \quad  1\leq i_1 < \ldots < i_{d-2} \leq d,
  \] 
  where the $(d-2)$-vector $\vec{\omega}_0:=(\star \omega_0)^{\sharp^{d-2}}$ is defined componentwise by
  \[
  \omega_0^{i_1\ldots i_{d-2}} = \frac{1}{2\sqrt{g}}\varepsilon^{kli_1\ldots i_{d-2}}\omega_{0kl}.
  \] 
\item[2.] (\textit{Integrated (circulation) form of the Cauchy invariants equation}).
  Since the Cauchy invariant may be rewritten as an exact $2$-form, i.e.
  \[
  dv_k\wedge dx^k= d(v_k dx^k),
  \]
  using  Hodge's decomposition, we obtain
  \begin{equation}
    \label{eq:integratedCIE}
    v_kdx^k= v_0^\flat + d\psi + h,
  \end{equation}
  where $\psi$ is an arbitrary $0$-form (scalar function) and $h$ is a harmonic $1$-form.
  Let $c$ be a $1$-chain on the  manifold $M$. 
  Choosing the function $\psi$ with suitable value on the boundary
  $\partial c$ (if it exists),
  from the Stokes theorem we obtain 
  \[
  \int_c v_kdx^k=  \int_c v_0^\flat + \int_c h.
  \]
  Moreover if the Betti number $b_{1}(M)=0$,
  then the second term
  on the right-hand side of the previous formula vanishes. Some examples for which $b_1(M)=0$ are given in
  Appendix~\ref{ssec:CICF:HCHD}. Considering now a $2$-chain $c$, using the Stokes theorem, we obtain
  \[
  \int_{\partial c} v_kdx^k=  \int_{\partial c} v_0^\flat.
  \]
  This is the famous theorem of conservation of circulation, frequently
  ascribed to \citet{Kel69} but actually discovered by 
  \citet[][see also \citet{FV14}]{Han61}, using essentially the argument given above.
\item[3.](\textit{Variational derivation of the Cauchy invariants equation}). 
  The Cauchy invariants equation \eqref{DefEqInvCau} on a Riemannian
  manifold has a
  variational derivation, using the  relabelling symmetry and
  Noether's theorem without appealing to Theorem~\ref{th:GCI}
  (see Appendix~\ref{sec:Rsym2form}).
\end{itemize}
}\end{remark}

We turn now to a  corollary that clarifies the relationship between the Cauchy
invariants equation and the Cauchy vorticity
formula, which are actually Hodge dual of
each other.
We refer the reader to Appendix~\ref{ssec:CICF:HSECD} for detailed definition of the 
the Hodge duality operator $\star : \Lambda^p(\Omega)\rightarrow  \Lambda^{d-p}(\Omega)$,
which implements the already mentioned Hodge duality. Indeed, applying the
Hodge dual operator to \eqref{DefEqInvCau},
we obtain the following

\begin{corollary}(Cauchy vorticity formula on a Riemannian manifold).
\label{cor:CICF}
Under the same assumptions as in Corollary \ref{th:InvVorticity}, we
have the Cauchy vorticity formula, written in general as 
\begin{equation}
\label{cauchyformula}
\star(dv_k\wedge dx^k)=\star \varphi_t^\ast\omega = \star\omega_0,
\end{equation}
and, in the case of a  three-dimensional curved space, as 
\begin{equation}
\label{cauchyformula_d3}
{\omega^i} =  \frac{\partial x^i}{\partial a^j}{\omega_0^j}, \quad i=1,2,3, 
\end{equation}
where the vorticity vector is defined componentwise by
\begin{equation}
\label{def:vorticity_RM}
\omega^i=\frac{1}{2\sqrt{g}} \varepsilon^{ijk}\omega_{jk}, \quad i=1,2,3.
\end{equation}
\end{corollary}
\begin{proof}Eq.~\eqref{cauchyformula} is of course an immediate consequence of
\eqref{DefEqInvCau}. To derive \eqref{cauchyformula_d3} in the case
$d=3$, we make  use again of the index raising operator $(\cdot)^\sharp: \mathcal{T}_1^0(M)\rightarrow \mathcal{T}_0^1(M)$.
In the three-dimensional curved case, \eqref{cauchyformula} is an equality between $1$-forms. Applying
the raising operator to \eqref{cauchyformula},  we obtain an equality between ($1$-contravariant) vectors, given by
\begin{equation}
\label{cauchyformula_d3_2}
[\star(dv_k\wedge dx^k)]^\sharp=(\star \varphi_t^\ast\omega)^\sharp = (\star\omega_0)^\sharp.
\end{equation}
Now, we expand \eqref{cauchyformula_d3_2} and show that it is  equivalent to the Cauchy formula \eqref{cauchyformula_d3}.
We set the notation $\mathrm{g}_0=\mathrm{g}(a)$ and  $\mathrm{g}=\mathrm{g}(x)$.
First, in terms of components of a $1$-form, and using the inversion formula
\[
\mathrm{g}_0^{-1}\varepsilon^{ijk} g_{0kn}=\varepsilon_{lmn} g_0^{il} g_0^{jm},
\]
we have
\begin{equation*}
\widetilde{\omega}_{0i_1}:= (\star \omega_0)_{i_1} = \frac12\sqrt{\mathrm{g}_0} \varepsilon_{i_1j_1j_2}\omega_0^{j_1j_2}=
\frac12\sqrt{\mathrm{g}_0} \varepsilon_{i_1j_1j_2}g_0^{lj_1}g_0^{mj_2}\omega_{0lm} 
=\frac12 \mathrm{g}_0^{-1/2}\varepsilon^{plm}g_{0i_1p}\omega_{0lm}.
\end{equation*}
In terms of components of a vector, we then obtain
\begin{equation*}
(\widetilde{\omega}_0^\sharp)^s=g_0^{sq} \widetilde{\omega}_{0q}=\frac12
\mathrm{g}_0^{-1/2}g_0^{sq}\varepsilon^{plm}g_{0qp}\omega_{0lm}
=\frac12 \mathrm{g}_0^{-1/2}\delta_p^s\varepsilon^{plm}\omega_{0lm}=\frac12 \mathrm{g}_0^{-1/2}\varepsilon^{slm}\omega_{0lm}.
\end{equation*}
From the definition of the vorticity vector \eqref{def:vorticity_RM}, we then have
\begin{equation*}
\label{eqnOS}
{\omega}_0^s=([\star \omega_0]^\sharp)^s=\frac{1}{2\sqrt{\mathrm{g}_0}}\varepsilon^{slm}\omega_{0lm}.
\end{equation*}
Second, in terms of components  of a $1$-form, we have
\begin{equation*}
\widetilde{\omega}_{i_1}:=(\star\varphi_t^\ast\omega)_{i_1}=
\frac12\sqrt{\mathrm{g}_0}
\varepsilon_{i_1j_1j_2}(\varphi_t^\ast\omega)^{j_1j_2}=\frac12
\sqrt{\mathrm{g}_0}
\varepsilon_{i_1j_1j_2}g_0^{lj_1}g_0^{mj_2}(\varphi_t^\ast\omega)_{lm}
=\frac12\mathrm{g}_0^{-1/2}\varepsilon^{plm}g_{0i_1p}(\varphi_t^\ast\omega)_{lm}.
\end{equation*}
In terms of vector components, and using ${\rm det}(\partial x/\partial a)=\sqrt{\mathrm{g}_0/\mathrm{g}}$, we then obtain
\begin{eqnarray*}
\widetilde{\omega}^s&=&(\widetilde{\omega}^\sharp)^s=g_0^{sq} \widetilde{\omega}_{q}=
\frac12\mathrm{g}_0^{-1/2}\varepsilon^{plm}g_0^{sq}g_{0qp}(\varphi_t^\ast\omega)_{lm}
=\frac12\mathrm{g}_0^{-1/2}\delta_p^s\varepsilon^{plm} (\varphi_t^\ast\omega)_{lm}
=\frac12\mathrm{g}_0^{-1/2}\varepsilon^{slm}(\varphi_t^\ast\omega)_{lm}\\
&=&\frac12\mathrm{g}_0^{-1/2}\varepsilon^{slm}\frac{\partial x^i}{\partial a^l}\frac{\partial x^j}{\partial a^m}\omega_{ij}
= \frac12\mathrm{g}_0^{-1/2}{\rm det}\left(\frac{\partial x}{\partial a} \right)\varepsilon^{kij}
\frac{\partial a^s}{\partial x^k}\omega_{ij}
=\frac{\partial a^s}{\partial x^k}\frac{1}{2\sqrt{\mathrm{g}}}\varepsilon^{kij}\omega_{ij}
\\
&=&
\frac{\partial a^s}{\partial x^k}\omega^k,
\end{eqnarray*}
where we have used the definition of the vorticity vector \eqref{def:vorticity_RM}.
Therefore, we have 
\[
\frac{\partial a^s}{\partial x^k}\omega^k=\omega_0^s,
\]
which gives \eqref{cauchyformula_d3} after inversion.
The latter is the vector form of the  Cauchy formula for a
three-dimensional Riemannian manifold $(M,g)$.
\end{proof}

In dimensions $d>3$ the Cauchy vorticity formula is no more an equality in terms of
$1$-forms (or vectors by the lowering-raising duality) but an equality
in terms of $(d-2)$-forms (or $(d-2)$-contravariant tensors  by the lowering-raising duality).
Thus for $d>3$, there still exists a Cauchy-type formula for the
vorticity, but given in general
by \eqref{cauchyformula}.

Specializing further, we then consider the flat 3D case and obtain the relations actually written by
\citet{Cau27} in modern vector notation (Cauchy, of course, wrote them
component by component):

\begin{corollary}(The flat Euclidean case: Cauchy (1815)). Let
  $M=\R^3$. Then the Cauchy invariants equation reads
\begin{equation*}
\label{CI3DEFS}
\sum_k\overrightarrow{\nabla} \dot{x}^k\times \overrightarrow{\nabla} x^k=\vec{\omega}_0,
\end{equation*}
while the Cauchy vorticity formula reads
\begin{equation*}
\label{CF3DEFS}
\vec{\omega}= Dx \,\vec{\omega}_0.
\end{equation*}
\end{corollary}

\begin{proof}
For the three-dimensional Euclidean flat space ($M=\R^3$), we have $g_{ij}=\delta_{ij}$, so that
first we obtain
\[
[\star(dv_k\wedge dx^k)]^\sharp = \sum_k\overrightarrow{\nabla} v^k\times
\overrightarrow{\nabla} x^k= \sum_k\overrightarrow{\nabla} \dot{x}^k\times \overrightarrow{\nabla} x^k,
\]
and second, we obtain
$
[\star \omega_0]^\sharp=\vec{\omega}_0.
$
Therefore we obtain the classical vector form of the Cauchy invariants
found by \citet{Cau27}:
\[
\sum_k\overrightarrow{\nabla} \dot{x}^k\times \overrightarrow{\nabla} x^k=\vec{\omega}_0.
\]
Multiplying the latter by the Jacobian matrix $Dx$, and using the relation
\[
\sum_k \left(Dx\left[\overrightarrow{\nabla} \dot{x}^k\times \overrightarrow{\nabla} x^k\right]\right)^j
=  \sum_k  \overrightarrow{\nabla} {x}^j\cdot\left(\overrightarrow{\nabla}
\dot{x}^k\times \overrightarrow{\nabla} x^k\right)
= \omega^j,
\]
we obtain
\[
\vec{\omega}= Dx\, \vec{\omega}_0,
\]
which is the classical vector form of the Cauchy vorticity formula.
\end{proof}

\section{Local helicities in hydrodynamics and MHD}
\label{sec:AHMHD}
In this section we show that there are interesting instances of
applications of Theorem~\ref{th:GCI} to $p$-forms having $p > 2$. In
particular there are various local helicities. We shall not, here, discuss
global (space-integrated) helicity \citep{Mor61,Mof69}. By ``local'', we mean
without spatial integration. One well-known instance is the magnetic
helicity in ideal MHD flow, for which it was shown by \citet{Els56} that it is a
material invariant.  Actually, all 3D known global helicities (kinetic
helicity in hydrodynamics, magnetic and cross helicities in MHD) have local
counterparts, which are Lie-advection-invariant $3$-forms along fluid particle
trajectories (in fact, Hodge duals of material-invariant pseudo-scalars).

In what follows, we shall make repeated use of the standard
result that the exterior product of a $p$-form $\omega$ and of a
$q$-form $\gamma$,
both of which are Lie advected, is also Lie
advected. Indeed, we have
 
\begin{equation*}
\label{Lie_go}
\partial_t \gamma + \mathsterling_v \gamma =0, \qquad \partial_t \omega + \mathsterling_v \omega =0.
\end{equation*}
Then, using  the following identity (see Appendix~\ref{ssec:CICF:EADF})
\[
\mathsterling_v(\gamma\wedge\omega)=\mathsterling_v\gamma\wedge\omega+ \gamma\wedge\mathsterling_v\omega,
\]
we obtain
\[
\partial_t(\gamma\wedge\omega)=\partial_t\gamma\wedge\omega +\gamma\wedge\partial_t\omega
=-\mathsterling_v\gamma\wedge\omega - \gamma\wedge \mathsterling_v\omega =-\mathsterling_v(\gamma\wedge\omega),
\]
which establishes the Lie advection of $\gamma\wedge\omega$.

\subsection{Local helicity in ideal hydrodynamics}
\label{ex:v.omega}
Here we assume that $\Omega$ is of dimension three ($d=3$).  
Let us recall the covariant Euler equations \eqref{euler_eq_oneform},  written in terms of
the velocity circulation $1$-form $v^\flat$:
\begin{equation}
  \label{euler_eq_oneform_2}
  \partial_t v^\flat + \mathsterling_{v} v^\flat = d\kappa.
\end{equation}
Here the $0$-form $\kappa$, is given by
\[
\kappa:=\frac12 (v,v)_g - p, \quad \mbox{or}  \quad \kappa:=\frac12 (v,v)_g - h, \ \ \mbox{with} \ \ dh=dp/\rho,
\]
in the incompressible case and the barotropic compressible case, respectively.
Let us introduce the $0$-form $\ell$ defined by the following equation
\begin{equation}
\label{eqn_LieAdv_Lagrangian}
\partial_t \ell + \mathsterling_{v} \ell =\kappa.  
\end{equation}  
Equation \eqref{eqn_LieAdv_Lagrangian} can be integrated along the flow $\varphi_t$
generated by the velocity vector field $v$, since \eqref{eqn_LieAdv_Lagrangian} is equivalent to
\begin{equation}
  \label{eqn_LieAdv_Lagrangian_2}
  \frac{d}{dt} \ell \circ \varphi_t =\kappa \circ \varphi_t. 
\end{equation}
Integrating \eqref{eqn_LieAdv_Lagrangian} in time, we obtain
\[
\ell(t,x) = \ell(0,a) + \int_0^t \kappa \circ \varphi_\tau\, d\tau,
\]
with the initial condition $\ell(0,a)=\ell_0(a)$. The function $\ell$ appears 
for the first time in the work of \citet{Web68} and might be called the Weber 
function.
Let us introduce, $u$, the modified velocity circulation  $1$-form defined by
\begin{equation}
  \label{def:u_1-form}
  u = v^\flat - d\ell. 
\end{equation}  
From the definition \eqref{def:u_1-form},
using \eqref{euler_eq_oneform_2}-\eqref{eqn_LieAdv_Lagrangian}, the $1$-form
$u$ satisfies
\begin{equation}
  \label{u_eq_oneform}
  \partial_t u + \mathsterling_{v} u = 0,
\end{equation}
and is thus Lie-advected. The $1$-form $u$ appears for the first time in
\citet{Cle59}, where it  takes the form $u = m d\psi$. Here, $m$ and $\psi$
are two material invariants (Lie-advected $0$-forms), now called the Clebsch 
variables; $u$ might thus be called the Clebsch $1$-form and the associated
vector the Clebsch velocity.
Of course,  the vorticity $2$-form
$
\omega=du=dv^\flat,
$
still satisfies the Lie advection equation
\begin{equation}
  \label{omega_eq_2form}
  \partial_t \omega + \mathsterling_{v} \omega = 0.
\end{equation}
From \eqref{u_eq_oneform}-\eqref{omega_eq_2form}, we deduce that the
local helicity
$3$-form $\sigma\in\Lambda^3(\Omega)$, which is defined by
\begin{equation*}
  \label{def:spirality}
  \sigma =u\wedge \omega=(v^\flat-d\ell)\wedge dv^\flat=v^\flat \wedge d v^\flat -d\ell\wedge dv^\flat, 
\end{equation*}
satisfies
\begin{equation}
  \label{sigma_eq_3form}
  \partial_t \sigma + \mathsterling_{v} \sigma = 0.
\end{equation}
This is a result of  
\citet[][where helicity is called spirality]{Ose88}.
Taking the Hodge dual of \eqref{sigma_eq_3form} and using the properties of the
Lie derivative (see Appendix~\ref{ssec:CICF:LD}) and of the Hodge dual operator (see Appendix~\ref{ssec:CICF:HSECD}),
we observe that the scalar local helicity $\star \sigma$ also satisfies a Lie-advection equation; thus it is also
a local conserved quantity,  as shown by \citet{Kuz83} in the 3D flat space.
Given that $\sigma$ is a $3$-form in a three-dimensional space, we obviously have $d\sigma =0$, and thus
$\sigma$ is closed on $\Omega$.
The situation is different for $d > 3$, because the $4$-form
$dv^\flat\wedge dv^\flat$ no longer vanishes.
Indeed, the wedge product is  not commutative in general (see
Appendix~\ref{ssec:CICF:EADF}); hence, the wedge product
$\alpha \wedge \alpha$ is identically zero only if the degree of the differential form $\alpha$ is
  odd (as is the case for the cross-product of two identical vectors).
Hence, $\sigma$ is not
closed; nevertheless, the helicity $3$-form $\sigma$ is still a local
invariant since \eqref{sigma_eq_3form} remains valid on Riemannian
manifolds of any dimension.

Thus local helicity, as a Lie-advection invariant $3$-form, actually exists in any
dimension $d\ge 3$, although it cannot in general  be associated (by Hodge
duality) to a material-invariant scalar.

Returning to the three-dimensional case,
we now suppose that the Betti number $b_3=0$ (see
Remark~\ref{rem:closedexact} and Appendix~\ref{ssec:CICF:HCHD}). This
guarantees that the closed $3$-form $\sigma$ is exact - that is, there exists a $2$-form
$\pi\in\Lambda^2(\Omega)$ such that
\begin{equation}
  \label{sigma=dpi}
  \sigma=d\pi, \quad \pi\in\Lambda^2(\Omega).
\end{equation}
From \eqref{sigma_eq_3form}-\eqref{sigma=dpi}, and using Theorem~\ref{th:GCI},
we obtain the following Cauchy invariants equation
\begin{equation}
  \label{eqn:CIpi}
  \frac12 \delta_{ij}^{kl}\,d\pi_{kl}\wedge dx^{i}\wedge dx^{j}=\sigma_{0}.
\end{equation}
In principle $\sigma_{0}=u_0\wedge \omega_0=v_0^\flat\wedge dv_0^\flat-dl_0\wedge dv_0^\flat$, but if we
choose the initial condition $l_0=0$,  we obtain
$\sigma_0=v_0^\flat\wedge dv_0^\flat$.
As stated in Corollary~\ref{Cor:inversionGCI},
\eqref{eqn:CIpi} can actually be inverted to obtain the $2$-form
$\pi$. In the present case, this is particularly simple: from
\eqref{eqn:CIpi}, using the inverse Lagrangian map, one obtains
componentwise
\begin{equation*}
  \label{eqn:CIpi2}
  \frac{\partial \pi_{ij}}{\partial a^l}=\sigma_{0lmn}\frac{\partial a^m}{\partial x^i}
  \frac{\partial a^n}{\partial x^j}.
\end{equation*}
Taking the divergence of this equation and inverting a Laplacian, one 
formally obtains
\[
\pi_{ij} = \delta^{kl}\Delta_a^{-1}\frac{\partial}{\partial{a^k}} \left( 
\sigma_{0lmn}\frac{\partial a^m}{\partial x^i}
\frac{\partial a^n}{\partial x^j}\right),
\]
where $\Delta_a^{-1}$ denotes the formal inverse of the Laplacian operator $\Delta_a=\sum_{i=1}^3\partial_{a^i}^2$
in cartesian coordinates, and $\delta^{kl}=1$ if $k=l$ and zero otherwise.

\subsection{Local helicities in ideal MHD}

\subsubsection{Local magnetic helicity}
\label{ex:A.B}
Here we assume that $\Omega$ is of dimension three ($d=3$).
From definition \eqref{B-2form}, and given that the Lie derivative and the exterior
derivative commute, integration of the induction equation \eqref{B-2form-eq}
leads to the following equation for the magnetic potential $1$-form: 
\begin{equation}
  \label{A-1form-eq_00}
  d(\partial_t A + \mathsterling_v A) = 0. 
\end{equation}
Using Hodge's decomposition for closed forms (Appendix~\ref{ssec:CICF:HCHD}) and \eqref{A-1form-eq_00},
there exists a harmonic $1$-form $\mathfrak{h}$ such that
\begin{equation}
  \label{A-1form-eq_01}
  \partial_t A + \mathsterling_v A = dK + \mathfrak{h}, 
\end{equation}
with $K$ an arbitrary $0$-form (scalar function) depending on the
choice of gauge condition for  the magnetic potential $1$-form $A$. 
We now assume that the Betti number $b_1=0$, as is the case, e.g.,
when the manifold is simply connected, contractible
or has a positive Ricci curvature (see Appendix~\ref{ssec:CICF:HCHD} and references therein). This ensures the
vanishing of the harmonic  $1$-form $\mathfrak{h}$, so that
\eqref{A-1form-eq_01} reduces to
\begin{equation}
  \label{A-1form-eq}
  \partial_t A + \mathsterling_v A = dK.
\end{equation}
We now introduce the $0$-form  $L$, which
is defined by the following equation
\begin{equation}
  \label{eqn_LieAdv_L}
  \partial_t L + \mathsterling_{v} L = K.  
\end{equation}  
Equation \eqref{eqn_LieAdv_L} can be integrated along the flow $\varphi_t$
generated by the velocity vector field $v$, since \eqref{eqn_LieAdv_L} is equivalent to
\begin{equation}
  \label{eqn_LieAdv_L_2}
  \frac{d}{dt} L \circ \varphi_t = K \circ \varphi_t. 
\end{equation}
Integrating \eqref{eqn_LieAdv_L_2} in time, we obtain
\[
L(t,x) = L(0,a) + \int_0^t K \circ \varphi_\tau\, d\tau,
\]
with the initial condition $L(0,a)=L_0(a)$.
We also introduce, $\mathcal{A}$, the modified magnetic potential $1$-form defined by
\begin{equation}
  \label{def:A_1-form}
  \mathcal{A} = A - dL. 
\end{equation}  
From the definition \eqref{def:A_1-form}, and
using \eqref{A-1form-eq}-\eqref{eqn_LieAdv_L}, the $1$-form
$\mathcal{A}$ satisfies
\begin{equation}
  \label{A_eq_oneform}
  \partial_t \mathcal{A} + \mathsterling_{v} \mathcal{A} = 0.
\end{equation}
From \eqref{B-2form-eq} and \eqref{A_eq_oneform},
we infer immediately that the magnetic helicity $3$-form 
$h\in\Lambda^3(\Omega)$,
which is defined by
\begin{equation*}
  \label{hm-3form}
  h=\mathcal{A}\wedge B = \mathcal{A}\wedge d A = A\wedge d A - dL\wedge dA,
\end{equation*}
satisfies
\begin{equation}
  \label{hm-3form-eq}
  \partial_t h + \mathsterling_v h =0. 
\end{equation}
Taking the Hodge dual of \eqref{hm-3form-eq} and using the properties of the
Lie derivative (see Appendix~\ref{ssec:CICF:LD}) and of the Hodge dual operator (see Appendix~\ref{ssec:CICF:HSECD}),
we observe that the scalar magnetic helicity $\star h$ also satisfies a Lie-advection equation; thus it is also
a local conserved quantity, as shown  first by \citet[][see also \citet{Wol58}]{Els56} in the 3D flat space.
Given that $h$ is a $3$-form in a three-dimensional space, we obviously have $dh =0$, and thus
$h$ is closed on $\Omega$.
The situation is different for $d > 3$, because the $4$-form $dA\wedge
dA$ no longer vanishes; hence $h$ is not closed, but the magnetic
helicity $3$-form $h$ is still a local invariant, since
\eqref{hm-3form-eq} remains valid on Riemannian manifolds of any
dimension provided that the Betti number $b_1=0$. Returning to the three-dimensional case, we now suppose
that the Betti number $b_3=0$ (see Remark~\ref{rem:closedexact} and
Appendix~\ref{ssec:CICF:HCHD}). This guarantees that the closed form
is exact, that is 
there exists a $2$-form $\alpha\in\Lambda^2(\Omega)$ such that
\begin{equation}
  \label{h=dalpha}
  h=d\alpha, \quad \alpha\in\Lambda^2(\Omega).
\end{equation}
From \eqref{hm-3form-eq}-\eqref{h=dalpha}, and using Theorem~\ref{th:GCI},
we obtain yet another Cauchy invariants equation, namely
\begin{equation}
  \label{eqn:CIalpha}
  \frac12 \delta_{ij}^{kl}\,d\alpha_{kl}\wedge dx^{i}\wedge dx^{j}=h_0.
\end{equation}
In principle $h_{0}=\mathcal{A}_0\wedge B_0=A_0\wedge dA_0-dL_0\wedge dA_0$, but if we
choose the initial condition $L_0=0$, we obtain $h_0=A_0\wedge dA_0$.
Equation \eqref{eqn:CIalpha}
can be solved, similarly to what was done in Sec~\ref{ex:v.omega}, to obtain 
the $2$-form $\alpha$ as 
\[
\alpha_{ij} = \delta^{kl}\Delta_a^{-1}\frac{\partial}{\partial{a^k}} \left( 
h_{0lmn}\frac{\partial a^m}{\partial x^i}
\frac{\partial a^n}{\partial x^j}\right).
\]

\subsubsection{Local cross-helicity}
\label{ex:v.B}
Here we assume that $\Omega$ is of dimension three ($d=3$).
We define the cross-helicity $3$-form $\xi\in\Lambda^3(\Omega)$ by
\[
\xi= u\wedge B.
\]
First from  \eqref{B-2form-eq}  and \eqref{u_eq_oneform}, we find that the $3$-form $\xi$
satisfies
\begin{equation}
  \label{xi_eq_3form}
  \partial_t \xi + \mathsterling_{v} \xi = 0.
\end{equation}
Taking the Hodge dual of \eqref{xi_eq_3form} and using the properties of the
Lie derivative (see Appendix~\ref{ssec:CICF:LD}) and of the Hodge dual operator (see Appendix~\ref{ssec:CICF:HSECD}),
we observe that the scalar cross-helicity $\star \xi$ also satisfies a
Lie advection equation; thus it is also
a local conserved quantity,  as shown by \citet{Kuz83} for the 3D flat space.
Given that $\xi$ is a $3$-form in a three-dimensional space,
we obviously have $d\xi =0$, and thus
$\xi$ is closed on $\Omega$. 
We now assume that
the Betti number $b_3=0$ (see Remark~\ref{rem:closedexact} and
Appendix~\ref{ssec:CICF:HCHD}). This guarantees that the closed $2$-form
$\xi$ is exact, that is, 
there exists a $2$-form $\chi\in\Lambda^2(\Omega)$
such that
\begin{equation}
  \label{xi=dchi}
  \xi=d\chi, \quad \chi\in\Lambda^2(\Omega).
\end{equation}
From \eqref{xi_eq_3form}-\eqref{xi=dchi}, and using Theorem~\ref{th:GCI},
we obtain still another Cauchy invariants equation
\begin{equation}
  \label{eqn:CIxi}
  \frac12 \delta_{ij}^{kl}\,d\chi_{kl}\wedge dx^{i}\wedge dx^{j}=\xi_{0}.
\end{equation}
In principle $\xi_{0}= u_0\wedge B_0=v_0^\flat\wedge dA_0-d\ell_0\wedge dA_0$, but if we
choose the initial condition $\ell_0=0$, then we obtain $\xi_0=v_0^\flat\wedge dA_0$.
Equation \eqref{eqn:CIxi},
can be solved to find the $2$-form $\chi$ by proceeding along the
same line as in Sec~\ref{ex:v.omega}. We thus
formally obtain
\[
\chi_{ij} = \delta^{kl}\Delta_a^{-1}\frac{\partial}{\partial{a^k}} \left( 
\xi_{0lmn}\frac{\partial a^m}{\partial x^i}
\frac{\partial a^n}{\partial x^j}\right).
\]

\subsubsection{Local extended helicities}
\label{ex:emhd_h}
Here, we consider local helicities associated to the extended ideal compressible
MHD equations \eqref{ECMHDm}-\eqref{ECMHDB} of Section~\ref{ss:ECMHD}. As 
shown by \citet{LMM16}, equations \eqref{ECMHDB} can be rewritten in
such a way that the unknowns become the magnetic potential $1$-forms $A_{\pm}$, instead of the
magnetic field $2$-forms $B_\pm$, with $B_\pm=dA_\pm$ and $A_\pm:=A+(d_e^2/\rho) \star d \star B+ \kappa_\pm v^\flat$.
More precisely, the magnetic potential $1$-forms $A_{\pm}$ satisfy
\begin{equation}
  \label{emhd_A}
\partial_t A_\pm +\mathsterling_{v_\pm} A_\pm=d\psi_\pm,  
\end{equation}
with the earlier defined vector fields $v_\pm:=v-\kappa_{\mp}\,(\star\, d \star B)^\sharp /\rho$.
Explicit expressions of the $0$-forms $\psi_\pm$  are not needed here \cite[see][]{LMM16}.
Let us now introduce the $0$-forms $L_\pm$, which
are defined by the following equations
\begin{equation}
  \label{eqn_LieAdv_Lpm}
  \partial_t L_\pm + \mathsterling_{v_\pm} L_\pm = \psi_\pm,  
\end{equation}
with initial condition $L_\pm(0,a)=L_{0\,\pm}(a)$.
Equations \eqref{eqn_LieAdv_Lpm} can be integrated along the Lagrangian flows $\varphi_{\pm\, t}$
generated by the vector fields $v_\pm$, similarly to what was done in Section~\ref{ex:A.B}. 
Let us introduce, $\mathcal{A}_\pm$, the modified magnetic potential $1$-forms defined by
\begin{equation}
  \label{def:Apm_1-form}
  \mathcal{A}_\pm = A_\pm - dL_\pm. 
\end{equation}  
From the definition \eqref{def:Apm_1-form}, and
using \eqref{emhd_A}-\eqref{eqn_LieAdv_Lpm}, the $1$-forms
$\mathcal{A}_\pm$ satisfy
\begin{equation}
  \label{Apm_eq_oneform}
  \partial_t \mathcal{A}_\pm + \mathsterling_{v_\pm} \mathcal{A}_\pm = 0.
\end{equation}
From \eqref{ECMHDB} and \eqref{Apm_eq_oneform},
we infer immediately that the \textit{extended magnetic helicity} $3$-forms 
$h_\pm$, here defined by
\begin{equation*}
  \label{hmpm-3form}
  h_\pm=\mathcal{A}_\pm\wedge B_\pm = \mathcal{A}_\pm\wedge d A_\pm = A_\pm\wedge d A_\pm - dL_\pm\wedge dA_\pm,
\end{equation*}
satisfy
\begin{equation}
  \label{hmpm-3form-eq}
  \partial_t h_\pm + \mathsterling_{v_\pm} h_\pm =0. 
\end{equation}
From \eqref{hmpm-3form-eq} we  obtain that the extended magnetic helicity $3$-forms 
$h_\pm$ are local invariants. By spatial integration, these local
conservation laws imply also the
known global conservation laws for
the integrals of the  $3$-forms $\mathcal{K}_\pm:= A_\pm\wedge d A_\pm$, established by \citet{LMM16}.
Indeed, noting that $h_\pm=\mathcal{K}_\pm - d(L_\pm dA_\pm)$, using the Stokes theorem, 
the Lie-derivative theorem \eqref{LDT} and equation \eqref{hmpm-3form-eq},
we obtain, for any domain $\Omega$,
\[
0=\int_{\varphi_{\pm\, t}(\Omega)}
\partial_t h_\pm + \mathsterling_{v_\pm} h_\pm
=\frac{d}{dt}\int_{\varphi_{\pm\, t}(\Omega)}h_\pm
=\frac{d}{dt}\int_{\varphi_{\pm\, t}(\Omega)}\mathcal{K}_\pm +
\frac{d}{dt}\int_{\partial\varphi_{\pm\, t}(\Omega)}L_\pm dA_\pm=
\frac{d}{dt}\int_{\varphi_{\pm\, t}(\Omega)}\mathcal{K}_\pm,
\]
where we have supposed that the generalised vorticities $B_\pm$ vanish on the boundaries of
$\varphi_{\pm\, t}(\Omega)$. In the three-dimensional case $d=3$, taking the Hodge dual of \eqref{hmpm-3form-eq}
and using the properties of the Lie derivative (see Appendix~\ref{ssec:CICF:LD}) and of the Hodge dual operator
(see Appendix~\ref{ssec:CICF:HSECD}), we observe that the scalar extended magnetic helicities $\star h_\pm$
also satisfy Lie-advection equations; thus they are also local
conserved quantities. In a three-dimensional space $\Omega$,
given that $h_\pm\in \Lambda^3(\Omega)$ are  $3$-forms, we obviously have $dh_\pm =0$, and thus
$h_\pm$ is closed on $\Omega$. We now suppose
that the Betti number $b_3=0$ (see Remark~\ref{rem:closedexact} and
Appendix~\ref{ssec:CICF:HCHD}), which guarantees that closed forms
are exact. Then  there exist $2$-forms $\alpha_\pm\in\Lambda^2(\Omega)$ such that
\begin{equation}
  \label{hpm=dalphapm}
  h_\pm=d\alpha_\pm, \quad \alpha_\pm\in\Lambda^2(\Omega).
\end{equation}
From \eqref{hmpm-3form-eq}-\eqref{hpm=dalphapm}, and using Theorem~\ref{th:GCI},
we obtain two more Cauchy invariants equations
\begin{equation}
  \label{eqn:CIalphapm}
  \frac12 \delta_{ij}^{kl}\,d\alpha_{\pm\,kl}\wedge dx_{\pm}^{i}\wedge dx_{\pm}^{j}=h_{\pm\,0},
\end{equation}
where $x_\pm$ are the Lagrangian maps generated by the vector fields $v_\pm$.
In principle $h_{\pm\,0}=\mathcal{A}_{\pm\,0}\wedge B_{\pm\,0}=A_{\pm\,0}\wedge dA_{\pm\,0}-dL_{\pm\,0}\wedge dA_{\pm\,0}$, but if we
choose the initial conditions $L_{\pm\,0}=0$, we obtain $h_{\pm\,0}=A_{\pm\,0}\wedge dA_{\pm\,0}=\mathcal{K}_{\pm\,0}$.
Equations \eqref{eqn:CIalphapm}
can be solved, similarly to what was done in Sec~\ref{ex:v.omega}, to obtain 
the $2$-forms $\alpha_\pm$ as 
\[
\alpha_{\pm\,ij} = \delta^{kl}\Delta_a^{-1}\frac{\partial}{\partial{a^k}} \left( 
h_{\pm\,0lmn}\frac{\partial a^m}{\partial x_\pm^i}
\frac{\partial a^n}{\partial x_\pm^j}\right).
\]

\subsection{Other high-order local invariants in hydrodynamics}
Here we consider a $d$-dimensional Riemannian manifolds $(M,g)$, with $d$
an odd natural integer and $\Omega$ a bounded region of $M$. Again,
we consider the velocity circulation $1$-form $u$, which is
defined by \eqref{def:u_1-form}. Using the $1$-form $u$, we define the
$d$-form $J\in\Lambda^d(\Omega)$ \citep{Ser84, GF93} by
\begin{equation}
  \label{dformJ}
  J= u\wedge (\wedge \, du)^{(d-1)/2},
\end{equation}
where $ (\wedge \,du)^{(d-1)/2}$ stands for  $(d-1)/2$ times the exterior product of the $2$-form
$du$. It was proven by \citet{GF93} that $J$ is Lie advected by the velocity field
$v$. Indeed, first the $1$-form $u$ satisfies the Lie-advection equation \eqref{u_eq_oneform}.
Second, taking the exterior derivative of equation \eqref{u_eq_oneform}
the $2$-form $du$ satisfies the same Lie-advection equation \eqref{u_eq_oneform}, because
Lie derivative and exterior derivative commute. Therefore
we obtain
\begin{equation}
  \label{LAFJ}
  \partial_t J + \mathsterling_{v} J = 0.
\end{equation}
Since $J\in\Lambda^d(\Omega)$, we obviously have $dJ =0$, and thus $J$
is closed on $\Omega$.
We now assume again that
the Betti number $b_{p}=0$ (see Remark~\ref{rem:closedexact} and
Appendix~\ref{ssec:CICF:HCHD}). This guarantees that the closed form
$J$ is exact - that is, there exists a $(d-1)$-form $I$
such that $J= dI$.
From exactness of the $d$-form $J$ and  \eqref{dformJ}-\eqref{LAFJ},  using Theorem~\ref{th:GCI},
we then obtain our last Cauchy invariants equation
\[
\frac{1}{(d-1)!}
\delta_{l_1 \ldots l_{d-1}}^{j_1\ldots j_{d-1}}\,
dI_{j_1\ldots  j_{d-1}} \wedge dx^{l_1}\wedge \ldots
\wedge  dx^{l_{d-1}} = J_0.
\]
In principle $J_{0}=u_0\wedge  (\wedge  \,dv_0^\flat)^{(d-1)/2}
=v_0^\flat\wedge  (\wedge  \,dv_0^\flat)^{(d-1)/2} -d\ell_0\wedge   (\wedge \, dv_0^\flat)^{(d-1)/2}$, but if we
choose a gauge such that $\ell_0=0$,  we obtain $J_0=v_0^\flat\wedge  (\wedge \, dv_0^\flat)^{(d-1)/2}$.
By  Corollary~\ref{Cor:inversionGCI}, the $(d-1)$-form $I$ can be
written as
\[
I_{i_1\ldots  i_{d-1}}=\delta^{k\ell}\Delta_a^{-1}\frac{\partial}{\partial a^k}
\left(
J_{0\ell j_1\ldots  j_{d-1}} \frac{\partial a^{j_1}}{\partial x^{i_1}}\ldots\frac{\partial a^{j_{d-1}}}{\partial x^{i_{d-1}}}
\right).
\]

\section{Conclusion and open problems}
 \label{sec:conc}

A key result of this paper, with all manners of applications to fluid
mechanics, is Theorem~\ref{th:GCI} of Sec.~\ref{sec:ARACLI} on generalised
Cauchy invariants equations.  A straightforward instance, is the \citet{Han61}
proof that the \cite{Cau27} invariants are equivalent to the \citet{Hel58}
theorem on the Lagrangian invariance of the vorticity flux through an
infinitesimal surface element. Our result is much more general, stating that
any Lie-invariant and exact $p$-form has an associated generalised Cauchy invariants 
equation, together with a Hodge dual formulation that generalises Cauchy's
vorticity formula.   The result, when
applied to suitable $3$-forms, also implies various generalisations of local
helicity conservation laws for Euler and MHD flow. There are several ways in
which the full nonlinear ideal MHD equations (compressible or incompressible)
can be recast as Lie-advection problems, leading to Cauchy invariants equations. It
is however not clear at the moment if such formulations lead to interesting
results on time-analyticity and numerical integration by Cauchy-Lagrange-type
methods \citep{ZF14,PZF16}. Similar questions arise for the extended MHD
models discussed in Section~\ref{ss:ECMHD}.

Cauchy-type formulations exist already for the compressible Euler--Poisson
equations in both an Einstein--de Sitter universe \citep[][see also
  \citet{EB97}]{ZF14} and a $\Lambda$CDM universe \citep{RVF15}. It is now clear
that the results are applicable to compressible
models, such as the barotropic fluid equations, and to the Euler--Poisson
equations or compressible MHD for fluid plasmas.

We remind the reader that problems with a Cauchy invariants formulation have
potentially a number of applications.  For example, we believe that Cauchy's
invariants should play an important part in understanding the regularity of
classical solutions to the 3D incompressible Euler equations through the
depletion phenomenon. Indeed, the Cauchy invariants involve finite sums of
vector products of gradients. Individual gradients are typically growing in 
the course of time but the constancy of the invariants put some
geometrical constraints on, for example, their alignments. This may, in due time, lead to the discovery of new estimates helping
to establish 3D regularity results, possibly for all times.

We also note that the Cauchy invariants formulation for the 3D incompressible
Euler equation allows constructive proofs of the regularity of Lagrangian map
through recursion relations among time-Taylor coefficients. These can then in
principle be implemented numerically, without being limited by the
Courant--Friedrichs--Lewy condition on time steps \citep{PZF16}. Given that
Cauchy invariants formulation apply both to flow in Euclidean (flat) space and
to flow on Riemannian curved spaces of any dimension, it is natural to ask if
the constructive and numerical tools just mentioned can be extended to flow in
curved spaces. This would allow us, for example, to numerically study the energy
inverse cascade on negatively curved spaces, recently investigated by
\citet{FG14} from an analytical point of view. It would also probably help
with flow in relativistic cosmology \citep{BO12, AEA15}.  

When leaving flat space, vector quantities involving tangent spaces at two or
more spatially distinct locations cannot be simply added or averaged. This
problem was encountered by \citet{GV16} in trying to handle the Generalised
Lagrangian Mean (GLM) theory on curved spaces; they solved it by using
pullback transport and optimal transport techniques. Another difficulty occurs
with time-Taylor series. Time derivatives of different orders, even when they
are evaluated at the same location, do live in tangent spaces of different
orders and cannot be readily combined.  Classical tools of differential
geometry, such as the exponential map, parallel transport, Lie series or
Lie transformations \citep{Nay73, DF76, Car81, Ste86} could be useful to
overcome this difficulty.

Finally, even in flat space, a generalised-coordinate formulation of the Cauchy
invariants equation can be useful in designing Cauchy-Lagrange
numerical schemes in non-cartesian coordinates. This could help the
investigation of  swirling axisymmetric flow in a cylinder, for which
finite-time
blowup is predicted by some numerical studies \citep{LH14a, LH14b}. In 
\citet{BF17} it was
shown that a constructive proof of finite-time regularity, based
on recursion relations adapted to wall-bounded Euler flow is available. 
The main difficulty is the high-precision implementation, needed to allow
reliable extrapolation without getting too close to the putative
blowup time. \\[-5.ex]

\section*{Acknowledgements}
We are grateful to Peter Constantin, Boris Khesin,
Manasvi Lingam, Philip J.  Morrison, Rahul Pandit and
anonymous referees for useful remarks and references.
This work was supported by the VLASIX and EUROFUSION projects respectively
under the grants No ANR-13-MONU-0003-01 and EURATOM-WP15-ENR-01/IPP-01. 

\appendix \section{Geometric and variational developments of the incompressible Euler equations}
\label{sec:GVDIE}

\subsection{Geometric interpretation of the incompressible Euler equations}
\label{sec:geodiffeuler}
We start by introducing briefly the notions of Lie groups and Lie algebra, 
which are important in the geometric view of
the incompressible Euler equations. A Lie group is a differentiable manifold
$G$ endowed with an associative multiplication, that is, a map
\[
\begin{tabular}{rll}
$ G\times G$ & $\rightarrow $ & $ G $   \\ 
$ (\eta,\sigma) $ & $\mapsto $ & $ \eta \sigma $
\end{tabular}
\]
making $G$ into a group and such that $(\tau \eta)\sigma=\tau(\eta \sigma)$ (associativity). 
Moreover there is an element $e\in G$ called the identity such that $e\eta=\eta e=\eta$. 
Such multiplication mapping, as well as, the inversion mapping 
\[
\begin{tabular}{rll}
$G$& $\rightarrow$ & $G$\\    
$\eta$ & $\mapsto$ & $\eta^{-1}, \quad \mbox{with} \quad \eta\eta^{-1}=e$,
\end{tabular}
\]
must be differentiable. 
To the Lie group $G$, we can naturally associate the Lie algebra $\mathfrak{g}$ defined by
\[
\mathfrak{g}=TG_e, 
\]
i.e. the tangent vector space of $G$ at the identity $e\in G$. In fluid
dynamics, the space $TG$ represents the Lagrangian (material) description
while the space $\mathfrak{g}$ represents the Eulerian (spatial)
description. For more details about Lie groups, Lie algebra, and their
applications in physics, we refer the reader, for example, to \citet{Arn66,
  AMR88, AK98, BA02, BCA10, DK00, Fec06, Fra12, HSS09, Ibr92, Ibr94, Ibr13,
  II07, Olv93}.

Here, the flow takes place on an oriented $d$-dimensional Riemannian manifold $(M,g)$
with metric volume form $\mu=\sqrt{\mathrm{g}} da^1\wedge \ldots \wedge
da^{d}\equiv\sqrt{\mathrm{g}}da$, where $\sqrt{\mathrm{g}}=\sqrt{{\rm
    det}(g_{ij})}$ (see Appendix~\ref{ssec:CICF:RM}).  Let $\Omega$ be a bounded
region of $M$.  In the \citet{Arn66} geometric interpretation of the
incompressible Euler equations, the solutions 
can be viewed as geodesics of the right-invariant Riemannian metric
given by the kinetic energy on the infinite-dimensional group of
volume-preserving diffeomorphisms. Indeed, let us define ${\rm
  SDiff}(\Omega,\mu)$ as the group of diffeomorphisms $\varphi:\Omega
\rightarrow \Omega$ preserving the metric volume form $\mu$,
i.e. $\varphi^\ast\mu =\mu$. Here the group multiplication is the composition
mapping denoted by ``$\circ$'' and $\varphi^\ast \mu$ is the pullback of the
$d$-form $\mu$ through the diffeomorphism $\varphi$. A precise definition of
the action on a tensor $ \Uptheta$ of the pullback operator $\varphi^\ast$ is
given in Appendix~\ref{ssec:CICF:PP}, but roughly speaking it consists in
evaluating the tensor $ \Uptheta$ at the point $\varphi(a)$, $a\in\Omega$
(that is the right composition of $ \Uptheta$ with $\varphi$), while taking
into account the deformation of the structure induced by the map $\varphi$
(reminiscent of a Jacobian matrix).  For the volume form $\mu$,  the
$d$-covariant antisymmetric tensor
\begin{equation*}
\label{def_form_mu}
\mu(a)=
\sqrt{\mathrm{g}(a)}da^{1}\wedge\ldots\wedge da^{d}=
\frac{1}{d!}\delta_{i_1\ldots  i_d}^{1\ldots.. d}\sqrt{\mathrm{g}(a)}da^{i_1}\wedge\ldots\wedge da^{i_d},
\end{equation*}
where $\delta_{i_1\ldots i_d}^{1\ldots .. d}$ is the generalised Kronecker symbol 
(see Appendix~\ref{sec:CICF}), we obtain by pullback
\begin{eqnarray*}
\label{def_pullback_of_form_mu}
\varphi^\ast\mu&=&\frac{1}{d!}\delta_{j_1\ldots j_d}^{1\ldots.. d}
\frac{\partial \varphi^{j_1}}{\partial a^{i_1}}\ldots
\frac{\partial \varphi^{j_d}}{\partial a^{i_d}}
\sqrt{\mathrm{g}(\varphi(a))}da^{i_1}\wedge\ldots\wedge da^{i_d}
=\frac{1}{d!}\delta_{i_1\ldots i_d}^{1\ldots.. d}
{\rm det}\left(\frac{\partial \varphi}{\partial a}\right) 
\sqrt{\mathrm{g}(\varphi(a))}da^{i_1}\wedge\ldots\wedge da^{i_d}\\
&=& {\rm det}\left(\frac{\partial \varphi}{\partial a}\right) \sqrt{\mathrm{g}(\varphi(a))}
da^{1}\wedge\ldots\wedge da^{d}.
\end{eqnarray*}
$G:={\rm SDiff}(\Omega,\mu)$ is a Lie group when $\Omega$ is a
compact differentiable manifold. Even if it not so, we
can associate to $G$ the Lie algebra $\mathfrak{g}:=TG_e$ consisting of all 
divergence-free vector fields $v$ tangent to the boundary (if it is not empty), i.e.  such that 
\begin{equation*} 
\nabla_i  v^i= 0, \ \ \mbox{on} \ \ \Omega, \quad \mbox{and} \quad
(v,\nu) =0, \ \  \mbox{on} \ \ \partial\Omega,
\end{equation*}
where $\nabla_k$ is the covariant derivative and $\nu$ denotes the unit outer normal vector at the boundary 
$\partial \Omega$. The covariant derivative $\nabla_k$ is a generalisation to curved spaces of the classical
partial derivative $\partial_k$ to Euclidean spaces (for a
more detailed definition, see Appendix~\ref{ssec:CICF:RCCD}).

In the algebra $\mathfrak{g}$, we define the scalar product
of two vector fields $v_1, v_2 \in \mathfrak{g}$, as 
\begin{equation}
\label{sp-g}
\langle v_1, v_2 \rangle_{\mathfrak{g}} = \int_{\Omega} (v_1,v_2)_g\, \mu,
\end{equation}
where the scalar product $(\cdot, \cdot)_g$, induced by the Riemannian metric $ds^2=g=g_{ij}da^i \otimes da^j$,
is given by $(v,w)_g  = g_{ij}v^iw^j, \ \,  v,\, w\in TM_a, \ a\in M$. 
Finally let us introduce the right translation acting on the group $G$.
Every element $\varphi$ of the group $G$ defines diffeomorphisms of the group onto 
itself:
\begin{equation}
\label{right-T}
R_{\varphi}: G \rightarrow G, \quad \quad R_{\varphi}\psi = \psi {\varphi}, \ \ \forall \psi\in G.
\end{equation}
The induced map on the tangent bundle $TG$ will be denoted by
\begin{equation}
\label{right-T-dual}
R_{\varphi\ast}: TG_\psi \rightarrow TG_{\psi\varphi}, \quad  \ \ \forall \psi\in G.
\end{equation}
Then a Riemannian metric on the group $G$ is called right-invariant if it is preserved 
under all right translations $R_\varphi$, i.e., if the derivative of the right translation carries
every vector to a vector of the same length. Thus it is sufficient to give a right-invariant
metric at one point of the group (for instance the identity), since the metric can be
carried over to the remaining points of the group by right translations.

We now consider the flow of a uniform ideal (incompressible and non-viscous) fluid in
the region $\Omega$. Here, and henceforth, by ``flow'' we understand a Lagrangian map
$M\ni a\rightarrow \varphi_t(a)\in M$, which, at this point, need not be a solution of the
Euler equations. Such a flow is given by a curve $t\rightarrow \varphi_t$ in
the group  ${\rm SDiff}(\Omega,\mu)$. This means that the diffeomorphism $\varphi_t$ maps
every particle of the fluid from the position $a$ it had at time $0$ to the position 
$x$ at time $t$.

If $\varphi_t$ is to be a solution of the Euler equations then,
according to the variational formulation \citep[see, e.g.,][]{Arn66}, the curve $\varphi_t$ is a  geodesic of the
group ${\rm SDiff}(\Omega,\mu)$. Such a curve extremizes the (Maupertuis) action 
defined as the time-integral of the kinetic energy:
\begin{equation}
\label{action_KE}
\mathcal{A}_K:=\frac{1}{2}\int_0^Tdt  \langle v(t),  v(t) \rangle_{\mathfrak{g}},
\end{equation}
where $v(t)$ is the Eulerian velocity vector field belonging to $\mathfrak{g}$.
This formulation is explicitly given in \citet{Arn66} but was probably already
known to \citet{Lag88} who never wrote it explicitly
because he switched quickly from variational formulations to so-called 
virtual velocity formulations.

It easily shown that the kinetic energy of the moving fluid is a
right-invariant Riemannian metric on the group ${\rm
  SDiff}(\Omega,\mu)$. Indeed, suppose that after time $t$ the flow of the
fluid gives a diffeomorphism $\varphi_t$, and the velocity at this moment of
time is given by the Eulerian vector field $v$. Then the diffeomorphism
realized by the flow after time $t+ dt$ (with $dt \ll 1$) will be
\begin{equation}
\label{flow_def}
\varphi_{t+dt} = \exp(vdt)\varphi_t + o(dt), 
\end{equation}
where $\tau\rightarrow \exp(v\tau)$ is in one-parameter group with vector $v$,
i.e.  the Lagrangian flow of the differential equation defined by the vector
field $v$. From \eqref{flow_def} and using the definitions
\eqref{right-T}-\eqref{right-T-dual} we have
\begin{equation*}
\label{RTg}
R_{{\varphi_t}^{-1}}\left(\frac{\varphi_{t+dt}-\varphi_{t}}{dt}\right) =
R_{{\varphi_t}^{-1}}\left(\frac{\exp(vdt)-e}{dt}\right)\varphi_t + o(1),
\end{equation*}
which, after taking the limit $dt\rightarrow 0$, leads to 
\begin{equation}
\label{def_v_eulerien}
v=R_{\varphi_t^{-1}\ast}\dot{\varphi}_t=\dot{\varphi}_t \circ {\varphi}_t^{-1}\quad 
\mbox{ or } \quad \dot{\varphi}_t = R_{\varphi_t\ast} v=v\circ {\varphi}_t.
\end{equation}
In mathematical language the velocity field $v$ is in the algebra
$\mathfrak{g}$ and is obtained from the vector $\dot{\varphi}_t$, tangent to
the group at the point ${\varphi}_t$, by right translation. In fluid-dynamics
terms the vector field $v=v_t(x)$ is the Eulerian velocity field.  We pass
from the Lagrangian to the Eulerian description by right translations. We note that if
we replace $\varphi$ by the composition $\varphi \circ \eta$, for a fixed
(time-independent) map $\eta \in {\rm SDiff}(\Omega,\mu)$, then
$\dot{\varphi}_t \circ {\varphi}_t^{-1}$ is independent of $\eta$. This
reflects the right invariance of the Eulerian description ($v$ is invariant
under composition of $\varphi$ by $\eta$ on the right). Therefore
$t\rightarrow {\varphi}_t$ is the geodesic, on the group ${\rm
  SDiff}(\Omega,\mu)$, of the right-invariant Riemannian metric given by the
quadratic form \eqref{sp-g}. From the Hamiltonian least action principle we
obtain the following Euler equations \eqref{euler_eq_contra} in contravariant
form. For the sake of completeness, details of the derivation are given in
Appendix~\ref{sec:eulerFLAP}.  Let $v\in\mathfrak{g}$ be the velocity field
defined by the right translation \eqref{def_v_eulerien}. Then there exists a
scalar function $p: ]0,T] \times \Omega \ni (t,x)\rightarrow \R$, the
    so-called pressure function, such that $(v,p)$ satisfy the following Euler
    equations
\begin{equation}
\label{euler_eq_contra}
\partial_t v^i+ v^k\nabla_kv^i= -g^{ik}\partial_k p,  \quad x \in\Omega, \ t\in]0,T]. 
\end{equation}

\subsection{Derivation of the Euler equations from a least action principle}
\label{sec:eulerFLAP}

From the discussion of Appendix~\ref{sec:geodiffeuler}, the geodesic motions
$t\rightarrow \varphi_t$ on ${\rm SDiff}(\Omega,\mu)$, which correspond to the
right-invariant Riemann metric defined by \eqref{sp-g}, are given by the
extrema of the action \eqref{action_KE} where $\dot{\varphi}_t =
R_{\varphi_t\ast} v=v(t,\varphi_t)$, under condition
$\varphi_t^\ast\mu=\mu$. To perform the extremization of the action
\eqref{action_KE} over ${\rm SDiff}(\Omega,\mu)$, it is convenient to impose
the volume-preservation constraint $\varphi_t^\ast\mu=\mu$ through a Lagrange
multiplier $\lambda(t,a)$ by adding to the action \eqref{action_KE} the term
\begin{equation}
\label{action_IC}
\mathcal{A}_I:=\int_0^T \lambda (\varphi_t^\ast\mu-\mu)\, da dt,
\end{equation}
We now compute the first 
variation of the action 
\begin{equation}
\label{action_T}
\mathcal{A}(\varphi,\lambda,\Omega)=\mathcal{A}_K(\varphi,\Omega) + \mathcal{A}_I(\varphi,\lambda,\Omega).
\end{equation}
We start with $\delta \mathcal{A}_K$. For its evaluation, we mainly use an integration by parts in time,
the symmetry of the metric tensor $g_{ij}$, the
definition of the covariant derivative (see Appendix~\ref{ssec:CICF:RCCD}), the change of variable 
$x=\varphi_t(a)=\varphi(t,a)$, 
the equations $\dot{\varphi}_t = v(t,\varphi_t)$ and $\varphi_t^\ast\mu=\mu$.
For the first variation of 
$\delta \mathcal{A}_K$ with volume preservation $\varphi= \varphi_t=\varphi(t,a)$, we then obtain
\begin{eqnarray}
\delta \mathcal{A}_K(\varphi,\Omega)[\delta\varphi] &=& \frac{1}2 \delta \int_0^Tdt\int_{\Omega}\mu(x)\,
g_{ij}(x)v^i(t,x)v^j(t,x)\nonumber\\
&=&  \frac{1}2 \delta \int_0^Tdt\int_{\Omega}\mu(a)\,
g_{ij}(\varphi(t,a))\partial_t\varphi^i(t,a) \partial_t\varphi^j(t,a) \nonumber\\
&=&  \frac{1}2 \int_0^Tdt\int_{\Omega}\mu(a)\,
\partial_k g_{ij}(\varphi(t,a))\delta \varphi^k(t,a) 
\partial_t\varphi^i(t,a) \partial_t\varphi^j(t,a)\nonumber\\
&& + \int_0^Tdt\int_{\Omega}\mu(a)\,
g_{ij}(\varphi(t,a))
\partial_t\varphi^i(t,a) \partial_t\delta\varphi^j(t,a)\nonumber\\
&=& \int_0^Tdt\int_{\Omega}\mu(x) \,\delta\varphi^j(t,\varphi_t^{-1}(x))\left\{
 - g_{ij}(x)\left[\partial_tv^i(t,x) + v^k(t,x)\partial_kv^i(t,x) \right]\right.\nonumber\\
&&\left. + \frac{1}2\partial_j g_{ik}(x)v^i(t,x)v^k(t,x) - \partial_k  g_{ij}(x)v^i(t,x)v^k(t,x) 
\right\}\nonumber\\
&=& -\int_0^Tdt\int_{\Omega}\mu \,u^jg_{ij}\left\{
\partial_tv^i + v^k\partial_kv^i + \frac12 g^{im}(\partial_kg_{lm}+\partial_lg_{km}-\partial_mg_{lk})v^kv^l
\right\}\nonumber\\
&=& -\int_0^Tdt\int_{\Omega}\mu \,u^jg_{ij}\left\{
\partial_tv^i + v^k\nabla_kv^i. 
\right\},\label{eqn_DAK}
\end{eqnarray}
Here $u^j(t,x)=\delta\varphi^j(t,\varphi_t^{-1}(x))$ and 
$\partial\equiv \partial_{a}$ denotes the partial derivative with respect the Lagrangian parameter
$a$ (initial position).
Next, for the first variation of 
$\delta \mathcal{A}_I$, using the definition of the volume form $\mu$ and
the following identities (see Appendix~\ref{ssec:CICF:PKD})
\begin{equation*}
  \frac{\partial \,{\rm det}( D_a\varphi)}{\partial (\partial_j \varphi^k)}=
  {\rm det} (D_a \varphi) ([D_a\varphi]^{-1})_k^j, \quad \partial_k \mathrm{g} = \mathrm{g}g^{ij}\partial_kg_{ij},
\end{equation*}  
we obtain
\begin{eqnarray*}
\delta  \mathcal{A}_I(\varphi,\lambda,\Omega)[\delta\varphi,\delta\lambda] &=&
\delta \int_0^T dt \int_{\Omega} da\, \lambda(t,a)\left(\sqrt{\mathrm{g}(\varphi(t,a))} {\rm det}(D_a\varphi(t,a))
-\sqrt{\mathrm{g}(a)} \right) \\
&=& \int_0^T dt \int_{\Omega} da\, \delta \lambda(t,a)\left(\sqrt{\mathrm{g}(\varphi(t,a))} {\rm det}(D_a\varphi(t,a))
-\sqrt{\mathrm{g}(a)} \right)\\
&& +
\int_0^T dt \int_{\Omega} da\, \lambda(t,a) \sqrt{\mathrm{g}(\varphi(t,a))}{\rm det}(D_a \varphi(t,a)) \\
&& \left(
\frac12 g^{ij}(\varphi(t,a))\partial_kg_{ij}(\varphi(t,a))\delta\varphi^k(t,a)
+ ([D_a\varphi(t,a)]^{-1})_k^j  \partial_j \delta \varphi^k(t,a)\right)\\
&=& \int_0^T dt \int_{\Omega} da\, \delta \lambda(t,a)\left(\sqrt{\mathrm{g}(\varphi(t,a))} {\rm det}(D_a\varphi(t,a))
-\sqrt{\mathrm{g}(a)} \right) \\
&& +
\int_0^T dt \int_{\Omega} \mu(x) \,p(t,x)\left(\frac12 g^{ij}(x)\partial_kg_{ij}(x)u^k(t,x)
+ \partial_k u^k(t,x)\right).
\end{eqnarray*}  
Here we have introduced the pressure function $p$ by setting $p(t,x)= \lambda(t,\varphi_t^{-1}(x))$.
Using an integration by parts in the last term of this equation, we finally obtain
\begin{eqnarray}
\delta  \mathcal{A}_I(\varphi,\lambda,\Omega) [\delta\varphi,\delta\lambda]&=&
\int_0^T dt \int_{\Omega} da\, \delta \lambda(t,a)\left(\sqrt{\mathrm{g}(\varphi(t,a))} {\rm det}(D_a\varphi(t,a))
-\sqrt{\mathrm{g}(a)} \right) \nonumber\\
&& -
\int_0^T dt \int_{\Omega} \mu(x) \,u^k(t,x)\partial_kp(t,x) \label{eqn_DAI},
\end{eqnarray}  
where we have used the boundary condition $(u,\nu)=0$ on $\partial \Omega$ for the infinitesimal variation
$u(t,x)=\delta\varphi(t,\varphi_t^{-1}(x))$.
Setting the first variation $\delta \mathcal{A}$ to zero, and using \eqref{action_T} and
\eqref{eqn_DAK}-\eqref{eqn_DAI}, we obtain the Euler equations \eqref{euler_eq_contra},
together with the volume-preserving condition $\varphi_t^\ast\mu=\mu$, which is
equivalent to the incompressibililty condition $\nabla_i v^i=0$ for the  velocity field.

\subsection{Derivation of the Cauchy invariants equation from the
  relabelling symmetry and a variational principle}
\label{sec:Rsym2form}
In this appendix, from the relabelling symmetry, i.e.
the invariance of the action under relabelling transformations, 
we recover the Cauchy invariants equation without appealing to Theorem~\ref{th:GCI}.
Here we follow the spirit of the proof given by \citet{FV14} and references therein for
the Euclidean case. The reader is also referred to this for historical discussion and description
of the use of different Hamiltonian principles or 
least action principles in Lagrangian coordinates.
Such a strategy  does not directly make use of Noether's theorem,
but  is reminiscent of its proof.
Before stating the result, we give the formal definition of a relabelling transformation.

\begin{definition}
\label{def:RelabSym}
A relabelling transformation is a map $\Omega \ni a\rightarrow \eta(a)\in \Omega$ such that
\[
\eta(a)=a + \delta a(a), \quad \delta a \in \mathfrak{g}, 
\]
i.e. with
\[
\nabla_i \delta a^i=0 \quad \mbox{and} \quad (\delta a, \nu)=0.
\]
In other words the vector field $\delta a $ is the infinitesimal generator of
a group of volume-preserving diffeomorphisms of $\Omega$ that leave the boundary $\partial \Omega$ invariant.
\end{definition}

\begin{theorem}(Cauchy invariants equation from the relabelling symmetry and variational principle). \label{th:InvCauchy}
  Let $\varphi_t$ be the Euler flow. We set $x=\varphi_t$ and $v=\dot{\varphi}_t$, with
  $v_0=\dot{\varphi}_0$. Then the invariance of the action \eqref{action_KE} of Appendix~\ref{sec:geodiffeuler} 
  under relabelling transformations of Definition \ref{def:RelabSym}
  implies the following Cauchy invariants conservation law:
  \begin{equation}
    \label{OrgCIE}
     dv_k\wedge dx^k = \omega_0:= dv_0^\flat.
\end{equation}  
\end{theorem}  

\begin{proof}
\label{sec:ECIRS}
The idea is first to compute the first-order variation of the action integral
\[
\mathcal{A}_K(\varphi,\Omega)=\frac12 \int_0^Tdt \int_\Omega \mu(a)
g_{ij}(\varphi_t(a)) \partial_t\varphi_t^i(a)
\partial_t\varphi_t^j(a),
\]
induced by the relabelling transformations of Definition~\ref{def:RelabSym}.
The variation of $\mathcal{A}_K(\varphi,\Omega)$ is given by
\begin{eqnarray}
\delta \mathcal{A}_K(\varphi,\Omega)[\delta\varphi] &=&
\frac{1}2 \delta \int_0^Tdt\int_{\Omega}\mu(a)\,
g_{ij}(\varphi_t(a))\partial_t\varphi_t^i(a) \partial_t\varphi_t^j(a) \nonumber\\
&=&  \frac{1}2 \int_0^Tdt\int_{\Omega}\mu(a)\,
\partial_l g_{ij}(\varphi_t(a))\delta \varphi_t^l(a) 
\partial_t\varphi_t^i(a) \partial_t\varphi_t^j(a)\nonumber\\
&&
+ \int_0^Tdt\int_{\Omega}\mu(a)\,
g_{ij}(\varphi_t(a))
\partial_t\delta\varphi_t^i(a) \partial_t\varphi_t^j(a).\label{eqn:VCI1}
\end{eqnarray}
The relabelling transformation of  Definition~\ref{def:RelabSym} induces
a change in the Lagrangian flow $\varphi_t$ at time $t$, given by
\begin{equation}
\label{eqn:deltaphi}
\delta \varphi_t = \frac{\partial \varphi_t}{\partial a^i} \delta\eta^i = \frac{\partial \varphi_t}{\partial a^i} \delta a^i.
\end{equation}
Substituting  \eqref{eqn:deltaphi} in \eqref{eqn:VCI1}, and using the product rule, we obtain
\begin{eqnarray}
\delta \mathcal{A}_K(\varphi,\Omega)[\delta a]
&=&\int_0^T\int_{\Omega}\mu(a)\,\left\{\frac12
\partial_l g_{ij}(\varphi_t(a)) \frac{\partial \varphi_t^l(a)}{\partial a^m} 
\partial_t\varphi_t^i(a) \partial_t\varphi_t^j(a) \delta a^m \right.\nonumber\\
&& \left.
+   g_{ij}(\varphi_t(a)) \partial_t\left(  \frac{\partial \varphi_t^i(a)}{\partial a^n}\right)
\partial_t\varphi_t^j(a) \delta a^n\right\}\nonumber
\\
&=&\int_0^T\int_{\Omega}\mu(a)\,\left\{\frac12
\partial_l g_{ij}(\varphi_t(a)) \frac{\partial \varphi_t^l(a)}{\partial a^m} 
\partial_t\varphi_t^i(a) \partial_t\varphi_t^j(a) \delta a^m \right.\nonumber\\
&& \left.
+   \partial_t\left( g_{ij}(\varphi_t(a)) \frac{\partial \varphi_t^i(a)}{\partial a^n}
\partial_t\varphi_t^j(a)\right) \delta a^n
- \partial_t\left( g_{ij}(\varphi_t(a)) \partial_t\varphi_t^j(a)\right)
\frac{\partial \varphi_t^i(a)}{\partial a^n}\delta a^n
\right\}\nonumber\\
&=&\int_0^T\int_{\Omega}\mu(a)\,\left\{\frac12
\partial_l g_{ij}(\varphi_t(a)) \frac{\partial \varphi_t^l(a)}{\partial a^m} 
\partial_t\varphi_t^i(a) \partial_t\varphi_t^j(a) \delta a^m\right.
\nonumber\\
&& \left.
- \partial_kg_{ij}(\varphi_t(a)) \partial_t\varphi_t^k(a)\partial_t\varphi_t^j(a)
\frac{\partial \varphi_t^i(a)}{\partial a^m}\delta a^m
-g_{ij}(\varphi_t(a)) \partial_t^2\varphi_t^j(a)
\frac{\partial \varphi_t^i(a)}{\partial a^m}\delta a^m
\right\}
\nonumber\\
&&
+\int_0^T\int_{\Omega}\mu(a)\, 
\partial_t\left( g_{ij}(\varphi_t(a)) \frac{\partial \varphi_t^i(a)}{\partial a^n}
\partial_t\varphi_t^j(a)\right) \delta a^n\nonumber\\
&=& I_1 + I_2.\label{eqn:VCI2}
\end{eqnarray}
First, we show that $I_1=0$. From \eqref{eqn:VCI2} and using the definition of the covariant derivative, we obtain
\begin{eqnarray*}
I_1 &=& 
\int_0^Tdt\int_{\Omega}\mu(a) \,\frac{\partial \varphi_t^j}{\partial a^m} \delta a^m\left\{
- g_{ij}(\varphi_t)\left[\partial_tv^i(t,\varphi_t) + v^k(t,\varphi_t)\partial_kv^i(t,\varphi_t)
 \right]\right.\nonumber\\
&&\left. + \frac{1}2\partial_j g_{ik}(\varphi_t)v^i(t,\varphi_t)v^k(t,\varphi_t) -
\partial_k  g_{ij}(\varphi_t)v^i(t,\varphi_t)v^k(t,\varphi_t) 
\right\}\nonumber\\
&=& -\int_0^Tdt\int_{\Omega}\mu(a) \,\frac{\partial \varphi_t^j}{\partial a^m} \delta a^m
g_{ij}(\varphi_t)\left\{
\partial_tv^i(t,\varphi_t) + v^k(t,\varphi_t)\partial_kv^i(t,\varphi_t)\right. \nonumber\\
&& \left.+ \frac12
g^{im}(t,\varphi_t)(\partial_kg_{lm}(t,\varphi_t)+\partial_lg_{km}(t,\varphi_t)-\partial_mg_{lk}(t,\varphi_t))
v^k(t,\varphi_t)v^l(t,\varphi_t)
\right\}\nonumber\\
&=& -\int_0^Tdt\int_{\Omega}\mu(a) \, \frac{\partial \varphi_t^j}{\partial a^m} \delta a^m g_{ij}(\varphi_t)\left\{
\partial_tv^i(t,\varphi_t)+ v^k(t,\varphi_t)\nabla_kv^i(t,\varphi_t) 
\right\}.\label{eqn:VCI3}
\end{eqnarray*}
Using the Euler equations \eqref{euler_eq_contra}, the term $I_1$ becomes
\begin{eqnarray*}
I_1&=&\int_0^Tdt\int_{\Omega}\mu(a) \,\delta a^m \frac{\partial \varphi_t^j}{\partial a^m}  g_{ij}(\varphi_t)
g^{ik}(\varphi_t)  \partial_k p (t,\varphi_t) 
=
\int_0^Tdt\int_{\Omega}\mu(a) \, \delta a^m \frac{\partial \varphi_t^j}{\partial a^m} 
\delta_j^k\partial_k p (t,\varphi_t) \nonumber\\
&=& \int_0^Tdt\int_{\Omega}\mu(a) \,\delta a^m \frac{\partial \varphi_t^k}{\partial a^m} 
\partial_k p (t,\varphi_t)
=\int_0^Tdt\int_{\Omega}\mu(a) \, \delta a^m \frac{\partial p}{\partial a^m}.
\label{eqn:VCI4}
\end{eqnarray*} 
Now, we recall that $\nabla_i\delta a^i=\mathrm{g}^{-1/2}\partial_i(\sqrt{\mathrm{g}}\delta a^i)=0$,  and $(\delta, \nu)=0$. Therefore,
using an integration by parts in space, the term $I_1$ becomes
\begin{equation*}
I_1=\int_0^Tdt\int_{\Omega}\mu(a) \, \delta a^i \frac{\partial p}{\partial a^i}
=- \int_0^Tdt\int_{\Omega}\mu(a) \, \nabla_i \delta a^i p +
\int_0^Tdt\int_{\partial \Omega} d\Gamma \sqrt{\mathrm{g}(a)} \, p\, (\delta a, \nu) 
=0.
\end{equation*}
Finally, we deal with the term $I_2$ defined in \eqref{eqn:VCI2}.
For this, we use the property that $\delta a\in \mathfrak{g}$, i.e.
$\nabla_n \delta a^n=0$ and $(\delta a,\nu)= \delta_{ij}\delta a^i \nu^j=0$.
Here, $\delta_{ij}$ is the metric tensor of an Euclidean space with cartesian
coordinates, i.e. $\delta_{ij}=0$ if $i\neq j$ and $\delta_{ij}=1$ if $i=j$. Such a vector
$\delta a$ can be constructed from a skew-symmetric $2$-contravariant tensor
$\xi^{ij}$ satisfying the following constraints:
\begin{equation}
\label{constonxi0}
\xi^{ij}+\xi^{ji}=0\ \ \mbox{on}\ \ \Omega, \quad
\delta_{ij}\xi^{ik}\nu^j=0\ \ \forall k \ \ \mbox{on}\ \ \partial \Omega, \quad \mbox{and} \quad
\delta_{ij}\xi^{ik}\partial_k\nu^j=0 \ \ \mbox{on}\ \ \partial \Omega.
\end{equation}
Indeed, if we set
\begin{equation}
\label{structdeltaeta0}
\delta a^i = \frac{1}{\sqrt{\mathrm{g}}} \partial_j \xi^{ij},
\end{equation}
then, using \eqref{constonxi0}, we find that $\nabla_i \delta a^i=0$ and
$(\delta a,\nu)= \delta_{ij}\delta a^i \nu^j=0$. We observe that a
skew-symmetric $2$-contravariant tensor $\xi^{ij}$ satisfying 
$\xi_{|_{\partial \Omega}}^{ij}=0$, satisfies also the boundary conditions \eqref{constonxi0}.   
Using \eqref{constonxi0}-\eqref{structdeltaeta0}, the term $I_2$ becomes
\begin{eqnarray*}
I_2&=& 
\int_0^T\int_{\Omega}\mu(a)\, 
\partial_t\left( g_{ij}(\varphi_t(a)) \frac{\partial \varphi_t^i(a)}{\partial a^n}
\partial_t\varphi_t^j(a)\right) \delta a^n\nonumber\\
&=&-\int_0^T\int_{\Omega}da\, \partial_k
\partial_t\left( g_{ij}(\varphi_t(a)) \frac{\partial \varphi_t^i(a)}{\partial a^n}
\partial_t\varphi_t^j(a)\right) \xi^{nk}  \nonumber\\
&&
+   \int_0^Tdt\int_{\partial \Omega} d\Gamma
\partial_t\left( g_{ij}(\varphi_t(a)) \frac{\partial \varphi_t^i(a)}{\partial a^n}
\partial_t\varphi_t^j(a)\right) \xi^{nk}\nu^m\delta_{km}
\nonumber\\
&=&  -\int_0^T\int_{\Omega}da\, 
\partial_t\partial_k\left( g_{ij}(\varphi_t(a)) \frac{\partial \varphi_t^i(a)}{\partial a^n}
\partial_t\varphi_t^j(a)\right) \xi^{nk}. 
\end{eqnarray*}  
The action $\mathcal{A}_K(\varphi,\Omega)$ should be invariant under relabelling transformations.
Thus the variation of the action integral, i.e.
$ \delta \mathcal{A}_K$, must vanish. Therefore we have $I_2=0$, i.e.
\[
-\int_0^T\int_{\Omega}da\, 
\partial_t\partial_k\left( g_{ij}(\varphi_t(a)) \frac{\partial \varphi_t^i(a)}{\partial a^n}
\partial_t\varphi_t^j(a)\right) \xi^{nk} = 0. 
\]
Since the $\xi^{nk}$'s are arbitrary, we obtain
\[
\frac{d}{dt}\partial_k\left( g_{ij}(\varphi_t(a)) \frac{\partial \varphi_t^i(a)}{\partial a^n}
\partial_t\varphi_t^j(a)\right) =0, \quad \forall k,\, n =1,\ldots, d.
\]
Integration in time of these equations leads to
\[
\partial_k\left(v_i(t,\varphi_t(a))\frac{\partial \varphi_t^i(a)}{\partial a^n}\right)
=\partial_{k}v_{0n}, \quad \forall k,\, n =1,\ldots, d. 
\]
Multiplying  these equalities by $da^k\wedge da^n$ and summing over the indices $k$ and $n$, we
obtain
\[
d(v_idx^i)=dv_0^\flat \quad \mbox{i.e.} \quad dv_i\wedge dx^i=\omega_0:=dv_0^\flat,
\]
which ends the proof.
\end{proof}  

\subsection{Conservation of the vorticity $2$-form, directly from Noether's theorem}
\label{sec:CV2FFNT}
As we shall now show, when Noether's theorem is literally applied to the
variational formulation of the Euler equations in conjunction with the
relabelling symmetry, it does not yield the Cauchy invariants but the
conservation (under pullback) of the vorticity $2$-form.

For this, we introduce
the Lagrangian density $\mathcal{L}_K$ associated to the action integral \eqref{action_KE}.
Since by definition we have
\begin{equation}
  \label{action_KE2}
\mathcal{A}_K(\varphi,\Omega)= \int_{0}^T dt \int_\Omega \mu(a)\, \mathcal{L}_K(a,\varphi,\partial \varphi),
\end{equation}
then, from \eqref{action_KE}, we obtain
\begin{equation}
\label{def:LK}  
\mathcal{L}_K(a,\varphi,\partial \varphi)=\frac12
g_{ij}(\varphi_t) \partial_t \varphi_t^i \partial_t \varphi_t^j. 
\end{equation}
In this definition of the Lagrangian density, $\partial \varphi$ denotes any first-order
partial derivative of $\varphi$ with respect to space or time variables.
Let us now define the energy-momentum tensor $T_\beta^\alpha$ by
\begin{equation}
\label{Tem}
T_\beta^\alpha=\frac{\partial\mathcal{L}_K}{\partial(\partial_\alpha \varphi_t^\gamma)}
\partial_\beta \varphi_t^\gamma-\mathcal{L}_K\delta_\beta^\alpha,
\end{equation}
where the contravariant (resp. covariant) index $\alpha$ (resp. $\beta$)
denotes space-time independent variables. The relabelling transformations, as
given in Definition~\ref{def:RelabSym} in Appendix~\ref{sec:Rsym2form}, lead us
to choosing the following functional variations
\[
\delta a, \ \ \mbox{such that}\ \ \nabla_i \delta a^i=0, \ \ \mbox{and}\ \ (\delta a, \nu)=0;
\quad \delta \varphi\equiv0, \quad \delta \partial\varphi\equiv0.
\]
Using these functional variations and the relabelling symmetry (i.e.
invariance of the action integral \eqref{action_KE2} under relabelling transformations), from Noether's
theorem \citep{Hil51, CH66, Lan70, JS98, GPS01, GH04}, 
we obtain the following conservation law
\begin{equation}
\label{conslaw}
\nabla_\alpha T^\alpha = 0  \quad \mbox{where} \quad  T^\alpha= T_i^\alpha \delta a^i.
\end{equation}
More precisely, using \eqref{def:LK}-\eqref{Tem} and the properties
of the Euler flow $\varphi_t$, 
the components of the covariant contraction of the energy-impulsion tensor $T^\alpha$ are
\begin{eqnarray*}
T^t &=&  T_l^\alpha \delta a^l 
    = \frac{\partial}{\partial(\partial_t \varphi_t^i)}
\left(\frac12 g_{mn}(\varphi_t)\partial_t \varphi_t^m\partial_t \varphi_t^n\right)
\partial_l \varphi_t^i \delta a^l
= g_{ij}(\varphi_t)\partial_t\varphi_t^i\partial_l\varphi_t^j \delta a^l,\\
T^{i}&=&-\mathcal{L}_K\delta a^i
= -\frac12 g_{jk}(\varphi_t) \partial_t \varphi_t^j \partial_t \varphi_t^k 
\delta a^i.
\end{eqnarray*}
Using this equality and the boundary condition
$(\delta a,\nu)=0$ (since $\delta a \in \mathfrak{g}$), we obtain the 
boundary condition $(T,\nu)=0$, where $T$ is the vector of components $T^i$. 
Integrating the conservation law \eqref{conslaw}  on $\Omega$ and using the boundary condition
$(T,\nu)=0$,  we obtain
\begin{eqnarray}
\int_\Omega \mu \,\nabla_\alpha T^\alpha &=&
\int_\Omega da\, \sqrt{\mathrm{g}} \,\nabla_\alpha T^\alpha 
= \int_\Omega da\, \partial_t(\sqrt{\mathrm{g}} T^t)
+ \int_\Omega da\, \partial_i(\sqrt{\mathrm{g}} T^i) \nonumber\\
&=& \int_\Omega da\, \partial_t(\sqrt{\mathrm{g}} T^t)
+ \int_{\partial \Omega} d\Gamma \sqrt{\mathrm{g}}\, (T,\nu) 
=\frac{d}{dt}\int_\Omega da\, \sqrt{\mathrm{g}} T^t:=0.
\label{inv_integral}
\end{eqnarray}
We now give details of the calculation of the time integral invariant \eqref{inv_integral}.
For this, we use the property that $\delta a\in \mathfrak{g}$, i.e.
$\nabla_i \delta a^i=0$ and $(\delta a,\nu)= \delta_{ij}\delta a^i \nu^j=0$.
Here, $\delta_{ij}$ is the metric tensor of an Euclidean space with cartesian
coordinates, i.e. $\delta_{ij}=0$ if $i\neq j$ and $\delta_{ij}=1$ if
$i=j$. Such a vector
$\delta a$ can be constructed from a skew-symmetric $2$-contravariant tensor
$\xi^{ij}$, which satisfies the following constraints
\begin{equation}
  \label{constonxi}
\xi^{ij}+\xi^{ji}=0\ \ \mbox{on}\ \ \Omega, \quad
\delta_{ij}\xi^{ik}\nu^j=0\ \ \forall k \ \ \mbox{on}\ \ \partial \Omega, \quad \mbox{and} \quad
\delta_{ij}\xi^{ik}\partial_k\nu^j=0 \ \ \mbox{on}\ \ \partial \Omega.
\end{equation}
Indeed, if we define $\delta a$ by
\begin{equation}
\label{structdeltaeta}
\delta a^i = \frac{1}{\sqrt{\mathrm{g}}} \partial_j \xi^{ij},
\end{equation}
then using \eqref{constonxi} we obtain that $\nabla_i \delta a^i=0$ and
$(\delta a,\nu)= \delta_{ij}\delta a^i \nu^j=0$. We note that a
skew-symmetric $2$-contravariant tensor $\xi^{ij}$, satisfying 
$\xi_{|_{\partial \Omega}}^{ij}=0$, also satisfies the boundary conditions \eqref{constonxi}. 
Using \eqref{constonxi}-\eqref{structdeltaeta}, and an integration by parts in space,
the integral invariant becomes
\begin{eqnarray*}
\int_\Omega da\, \sqrt{\mathrm{g}}T^t&=&\int_\Omega  da\, \sqrt{\mathrm{g}}
g_{ij}(\varphi_t)\partial_t\varphi_t^i\partial_l\varphi_t^j \delta a^l
=\int_\Omega da\,
g_{ij}(\varphi_t)\partial_t\varphi_t^i\partial_l\varphi_t^j \partial_k \xi^{kl}
=\int_\Omega da\, v_j(t,\varphi_t)\partial_l\varphi_t^j \partial_k \xi^{kl}\\
&=&- \int_\Omega da\,
\left\{ \partial_i v_j(t,\varphi_t)\partial_k\varphi_t^i \partial_l\varphi_t^j +
 v_j(t,\varphi_t)\partial_{kl}\varphi_t^j 
 \right\} \xi^{kl}
 + \int_{\partial\Omega} d\Gamma\, v_j(t,\varphi_t)\partial_l\varphi_t^j \xi^{kl}\nu^m\delta_{km}\\
&=&-\frac12\int_\Omega da\,\left(
 \partial_i v_j(t,\varphi_t) - \partial_j v_i(t,\varphi_t)
 \right)\partial_k\varphi_t^i \partial_l\varphi_t^j \xi^{kl}
-\int_\Omega da\, v_j(t,\varphi_t)\partial_{kl}\varphi_t^j \xi^{kl}\\
&=&-\frac12\int_\Omega da\,\left(
 \partial_i v_j(t,\varphi_t) - \partial_j v_i(t,\varphi_t)
 \right)\partial_k\varphi_t^i \partial_l\varphi_t^j \xi^{kl}.
\end{eqnarray*}
Therefore, we obtain
\begin{equation}
\label{eqn_int_inv}
\int_\Omega da\,\xi^{kl}\, \frac{d}{dt} \varphi_t^\ast\omega_{kl} =0,
\end{equation}
where we have defined the components of the vorticity $2$-form $\omega_{kl}$ as 
\begin{equation*}
\omega_{kl}(t,x)= \partial_k v_l(t,x) - \partial_l v_k(t,x).
\end{equation*}
Since the functions $\xi^{kl}$'s are arbitrary and smooth, equality \eqref{eqn_int_inv}
implies
\[
\frac{d}{dt} \varphi_t^\ast\omega_{kl} =0,
\]
which implies
\begin{equation}
\label{finnoether}
 \varphi_t^\ast\omega=\omega_0.
\end{equation}
Here,
\[
\omega(t,x)=\sum_{i<j}\omega_{ij}(t,x)dx^i\wedge dx^j, \ \ \mbox{and}  \ \
\omega_0(a)=\omega(0,a)=\sum_{i<j}\omega_{ij}(0,a)da^i\wedge da^j=
\sum_{i<j}\omega_{0ij}(a)da^i\wedge da^j.
\]
Eq.~\eqref{finnoether} establishes the invariance of the vorticity $2$-form
under pullback.

\subsection{About Cartan's formula}
\label{appendixCF} 
The aim of this appendix it to establish the formula
\begin{equation}
\label{eqn:CFRecast}
d{\rm i}_vv^\flat +{\rm i}_vdv^\flat = (\nabla_v v)^\flat + \frac12 d (v,v)_g.
\end{equation}
First, using definitions of the interior product ${\rm i}_v$ and the exterior
derivative $d$, given 
in Appendix~\ref{ssec:CICF:EDIP}, and the symmetry of the Christoffel symbols in
the definition of the covariant derivative (see Appendix~\ref{ssec:CICF:RCCD}), for a vector field
$X\in\mathcal{T}_0^1(M)$ and a $1$-form $X\in\mathcal{T}_1^0(M)$, we obtain
\begin{equation}
\label{eqn:ixda}
{\rm i}_Xd\alpha = {\rm i}_X\left(\frac12(\partial_i \alpha_j - \partial_j \alpha_i)dx^i\wedge dx^j\right)
=X^j(\partial_j \alpha_i - \partial_i \alpha_j)dx^i=X^j(\nabla_j \alpha_i - \nabla_i \alpha_j)dx^i.
\end{equation} 
Second, using the same properties as for deriving \eqref{eqn:ixda}, we obtain
\begin{equation}
\label{eqn:dixa} 
d{\rm i}_X\alpha= \partial_i(X^j\alpha_j) dx^i=\nabla_i(X^j\alpha_j)dx^i =(X^j\nabla_i\alpha_j+ \alpha_j\nabla_iX^j)dx^i. 
\end{equation} 
Adding \eqref{eqn:ixda} and \eqref{eqn:dixa}, we obtain
\begin{equation*}
\label{eqn:ixda+dixa}
d{\rm i}_X\alpha+ {\rm i}_Xd\alpha = (X^j\nabla_j \alpha_i +   \alpha_j\nabla_iX^j)dx^i = 
(X^j\partial_j \alpha_i +   \alpha_j\partial_iX^j)dx^i=\mathsterling_X \alpha.
\end{equation*}
Using this equation with $X=v$ and $\alpha=v^\flat$, the lowering-raising operators 
and the property $\nabla_i g_{jk}=0$, we obtain
\begin{eqnarray*}
v^j\nabla_j v_i^\flat +  v_j^\flat \nabla_i v^j &=&
 v^j\nabla_j v_i^\flat + v^k g_{jk} \nabla_i v^j =  v^j\nabla_j v_i^\flat +  \frac12 g_{jk} \nabla_i(v^jv^k ) \\
&=&  v^j\nabla_j v_i^\flat +  \frac12 \nabla_i( g_{jk}  v^j v^k) =  v^j\nabla_j v_i^\flat +  \frac12 \partial_i( g_{jk} v^j v^k),
\end{eqnarray*}
which reexpresses \eqref{eqn:CFRecast} in terms of components.
For more details see, e.g., \citet[][Chap.IV, pp. 202--204]{AK98}.

\subsection{Proof of a commutation relation needed for the Lie-advection of the vorticity vector}
\label{appendixCR}
In Section~\ref{ssec:CFVE}, to establish the Lie-advection equation for the
vorticty vector, we have used a result on the commutation of the composition of
the raising operator with the Hodge dual operator and  the Lie
derivative. Here, we give a proof of the commutation relation
$[\sharp^{d-p}\,\star,\mathsterling_v]=0$ with the condition $\nabla_i
v^i=0$. We are also motivated by the observation  that we were not able to  
find a proof in the published literature. 

Let $\omega$
be a $p$-form. Using the definitions of the raising operator (see
Appendix~\ref{ssec:CICF:RM}) and of the Hodge dual operator (see
Appendix~\ref{ssec:CICF:HSECD}), and recognising the determinant of the metric
tensor in the following expression, we obtain
\begin{eqnarray}
\omega^{i_1 \ldots i_{d-p}}=\big([\star \omega]^{\sharp^{d-p}}\big)^{i_1 \ldots i_{d-p}}
&=&   \frac{1}{p!}\sqrt{g}\varepsilon^{j_1\ldots j_pl_1\ldots l_{d-p}} g^{j_1k_1} \ldots g^{j_pk_p} g^{i_1l_1} \ldots g^{i_{d-p}l_{d-p}} \omega_{k_1\ldots k_p} \nonumber\\
&=& \frac{1}{p!}\frac{1}{\sqrt{g}}\varepsilon^{k_1\ldots k_pi_1\ldots i_{d-p}}\omega_{k_1\ldots k_p}.\label{ACR1}
\end{eqnarray}
Using definitions of the Lie derivative (see Appendix~\ref{ssec:CICF:EADF}), of the raising and Hodge star operators,
and using the product rule to reveal the divergence of the vector field $v$ and the term 
$\partial_l \omega^{i_1 \ldots i_{d-p}}$ in the next expression, we obtain
\begin{eqnarray*}
  \big([\star \mathsterling_v\omega]^{\sharp^{d-p}} \big)^{i_1 \ldots i_{d-p}}&=&
  \frac{1}{p!}\frac{1}{\sqrt{g}}\varepsilon^{k_1\ldots k_pi_1\ldots i_{d-p}}
  (v^l\partial_l \omega_{k_1 \ldots k_{p}}+p\omega_{lk_2\ldots k_p}\partial_{k_1}v^l)\nonumber\\ 
  &=&v^l\partial_l \omega^{i_1 \ldots i_{d-p}} + \omega^{i_1 \ldots i_{d-p}}(\nabla_l v^l -\partial_l v^l)
  +\frac{1}{(p-1)!}\frac{1}{\sqrt{g}}\varepsilon^{k_1\ldots k_pi_1\ldots i_{d-p}}\omega_{lk_2\ldots k_p}\partial_{k_1}v^l\nonumber\\
  &=& T_1 + \frac{1}{(p-1)!}\frac{1}{\sqrt{g}}\varepsilon^{k_1\ldots k_pi_1\ldots i_{d-p}}\omega_{lk_2\ldots k_p}\partial_{k_1}v^l = T_1+T_2. \label{ACR2}
\end{eqnarray*}  
Using \eqref{ACR1}, the antisymmetry of $\omega$, and properties of generalised Kronecker symbols (see Appendix~\ref{ssec:CICF:PKD}), we obtain
\begin{eqnarray}
 \frac{(-1)^p\sqrt{g}}{(d-p)!}\varepsilon_{i_1\ldots i_{d-p}j_1\ldots j_{p}}\omega^{i_1\ldots i_{d-p}}
  &=&\frac{1}{p!(d-p)!} \varepsilon_{j_1\ldots j_pi_1\ldots i_{d-p}}\varepsilon^{k_1\ldots k_pi_1\ldots i_{d-p}}\omega_{k_1\ldots k_p}\nonumber\\
  &=&\frac{1}{p!}\delta_{j_1\ldots j_p}^{k_1\ldots k_p}\omega_{k_1\ldots k_p}=\omega_{j_1\ldots j_p}. \label{ACR3}
\end{eqnarray}  
Substituting \eqref{ACR3} in $T_2$, and using  properties of generalised Kronecker symbols, we obtain
\begin{eqnarray}
  T_2 &=& \frac{(-1)^p}{(p-1)!} \frac{1}{(d-p)! }
  \varepsilon^{k_1\ldots k_pi_1\ldots i_{d-p}}\varepsilon_{l_1\ldots l_{d-p} l k_2\ldots k_p}\omega^{l_1\ldots l_{d-p}}\partial_{k_1}v^l\nonumber\\
  &=& \frac{1}{(d-p)!}\delta_{l_1\ldots l_{d-p}l}^{i_1\ldots i_{d-p} k_1}\omega^{l_1\ldots l_{d-p}}\partial_{k_1}v^l\nonumber\\
 &=&\frac{1}{(d-p)!}\sum_{k=1}^{d-p+1}(-1)^{d-p+k+1}\delta_{l_k}^{k_1}\delta_{l_1\ldots \hat{l}_k \ldots l_{d-p+1}}^{i_1\ldots i_{d-p}}
  \omega^{l_1\ldots l_{d-p}}\partial_{k_1}v^{l_{d-p+1}},\label{ACR4}
\end{eqnarray}  
where we have set $l_{d-p+1}:=l$, and where the hat character $\,\hat{}\,$  indicates an index that is omitted from the sequence.
Using the antisymmetry of $\omega$, equation \eqref{ACR4} becomes
\[
T_2=\partial_l v^l\omega^{i_1\ldots i_{d-p}} +
\frac{1}{(d-p)!}\sum_{k=1}^{d-p}(-1)^{d-p+k+1}\delta_{l_1\ldots \hat{l}_k \ldots l_{d-p}l}^{i_1\ldots i_{d-p}}
  \omega^{l_1\ldots l_{d-p}}\partial_{l_k}v^l=T_{21} +T_{22}.
\]
Using properties of generalised Kronecker symbols, the antisymmetry of $\omega$, and relabeling some indices, we obtain
\begin{eqnarray*}
T_{22}&=&-\frac{1}{(d-p)!}\delta_{l_1\ldots l_{d-p}}^{i_1\ldots i_{d-p}}
\sum_{k=1}^{d-p}\omega^{l_1\ldots n_k \hat{l}_k\ldots l_{d-p}}\partial_{n_k}v^{l_k} \\
&=&-\frac{1}{(d-p)!} \sum_{k=1}^{d-p} \left( (-1)^{k+1} \delta_{l_1}^{i_k}\delta_{\hat{l}_1l_2\ldots l_{d-p}}^{i_1\ldots  \hat{i}_k \ldots i_{d-p}}
\partial_{n_1}v^{l_1}\omega^{n_1l_2\ldots l_{d-p}} \ + \ \ldots \ +\right.\\
&& \hspace{2.5cm}\left. (-1)^{k+d-p} \delta_{l_{d-p}}^{i_k}\delta_{l_1\ldots \hat{l}_{d-p}}^{i_1\ldots  \hat{i}_k \ldots i_{d-p}}
\partial_{n_{d-p}}v^{l_{d-p}}\omega^{l_1\ldots l_{d-p-1}n_{d-p}} \right)\\
&=&-\frac{1}{d-p} \sum_{k=1}^{d-p} \left( (-1)^{k+1}\partial_{l}v^{i_k}\omega^{li_1\ldots \hat{i}_k\ldots i_{d-p}} \ +\  \ldots\ +\ 
(-1)^{k+d-p}\partial_{l}v^{i_k}\omega^{i_1\ldots \hat{i}_{k} \ldots i_{d-p}l} \right)\\
&=&-\partial_{l}v^{i_1}\omega^{li_2\ldots i_{d-p}} \ -\ \ldots\ -\ \partial_{l}v^{i_{d-p}}\omega^{i_1\ldots i_{d-p-1}l}.
\end{eqnarray*}
Finally, putting all the terms together, using the condition $\nabla_i v^i=0$,
and remembering the definition of Lie derivative for tensors (see Appendix~\ref{ssec:CICF:LD}), we obtain
\begin{eqnarray*}
  \big([\star \mathsterling_v\omega]^{\sharp^{d-p}} \big)^{i_1 \ldots i_{d-p}}&=& T_1+T_2=T_1 + T_{21} +T_{22}\\
  &=& v^l\partial_l \omega^{i_1 \ldots i_{d-p}} - \partial_{l}v^{i_1}\omega^{li_2\ldots i_{d-p}} - \ldots -\partial_{l}v^{i_{d-p}}\omega^{i_1\ldots i_{d-p-1}l}\\
  &=&\big(\mathsterling_v[\star \omega]^{\sharp^{d-p}} \big)^{i_1 \ldots i_{d-p}},  
\end{eqnarray*}
which ends the proof.

\section{Differential geometry in a nutshell}
\label{sec:CICF} 
In this appendix we recall some notions of differential geometry.
There exist many classical textbooks of differential geometry on
manifolds, for example 
\citet{AMR88, Arn89, CB68, CDD77, DeR84, Fec06, Fla63, Fra12, Hel62, KN63, LR89, Sch80, Spi79, Ste64}.
This appendix is based on textbooks that we find pedagogical for our
intended readership
\citep{AMR88, Arn89, CDD77, DeR84, Fec06, Fra12}, to which we give precise references.

\subsection{Manifolds, tangent and cotangent bundles}
\label{ssec:CICF:Ma}
A manifold is a generalisation of the notion of a smooth surface in
Euclidean space.  The concept of manifold has proved to be useful
because they occur frequently, and not just as subsets embedded in an
Euclidean space. Indeed such a generalisation, eliminating the need for a
containing Euclidean space, makes the construction intrinsic to the manifold itself.
Usually a differentiable (smooth) manifold $M$ of dimension $d$ is
defined through a differentiable parametric representation, called
an atlas, which can be seen as a collection of charts
${(U_i,\phi_i)}_{i\in I}$ such that $M=\cup_{i\in I} U_i$.  A chart
$(U_i,\phi_i)$ is a local subset $U_i \subset M$ and local smooth
bijection $\phi_i$ from $U_i$ to an open subset of Banach space
(typically $\R^d$). The manifold $M$ is then constructed by patching
smoothly such objects together. For a formal definition of a
differentiable manifold we refer the reader to \citet[][Sec. III.A.1,
 pp. 111]{CDD77}, \citet[][Sec. 3.1, pp. 141]{AMR88} and
\citet[][Sec. 1.2c, pp. 19]{Fra12}.

The set of tangent vectors to $M$ at $a\in M$ forms a vector space $TM_a$. This space is called
the tangent space to $M$ at $a$. The union of the tangent spaces to $M$ at the various point of $M$,
i.e. $TM:= \cup_{a\in M}TM_a$, has a natural differentiable manifold structure, the dimension of which
is twice the dimension of $M$. This manifold is called the tangent bundle of $M$ and is denoted by
$TM$. The mapping $\pi:TM\rightarrow M$, which takes a tangent vector $V$ to the point $a\in M$
at which the vector is tangent to $M$ (i.e. $X\in TM_a$), is called the natural projection. The inverse
image of a point $a\in M$ under the natural projection, i.e. $\pi^{-1}(a)$, is the tangent space $TM_a$.
This space is called the fiber of the tangent bundle over the point $a$. A vector field on $M$ is
a (cross-)section of $TM$. A (cross-)section of a vector bundle assigns
to each base point $a\in M$ a vector in the fiber $\pi^{-1}(a)$ over $a$  and the addition and
scalar multiplication of sections takes place within each fiber
\citep[see, e.g.,][Sec. 2.2, pp. 48 and  III.B.3, pp. 132 in \citet{CDD77}]{Fra12}.

As for ordinary vector spaces, one can define the dual of the tangent
bundle, noted $T^\ast M$, which can be constructed through linear
forms, called $1$-forms or cotangent vectors, acting on vectors of the
tangent bundle $TM$.  The cotangent space to $M$ at $a$, noted $T^\ast M_a$, 
is the set of all cotangent vectors to $M$ at $a$. The cotangent
bundle is the union of the cotangent spaces to the manifold $M$ at all its points,
that is $T^\ast M:=\cup_{a\in M} T^\ast M_a$.  The cotangent
bundle $T^\ast M$ has a natural differentiable manifold structure, 
the dimension of which is twice the dimension of $M$.

Finally we introduce the notion of contractible manifolds.
Let $c:[0,1]\rightarrow M$ be a continuous map such that $c(0)=c(1)=a\in M$. We call $c$ a loop in $M$ 
at the point $a$. The loop is called contractible if there is a continuous map $H:[0,1]\times[0,1]\rightarrow M$
such that $H(t,0) = c(t)$ and $H(0,s) = H(1,s) = H(t,1) = a$ for all $t\in [0,1]$. Indeed $c_s(t) = H(t,s)$ has to be
viewed as a family of arcs connecting $c_0 = c$ to $c_1$, a constant arc. Roughly speaking, a loop is contractible when 
it can be shrunk continuously to the point $a$ by loops beginning and ending at $a$. The manifold $M$ is 
contractible to a point $a$, if every loop in $M$, which starts and ends at the point $a$ is contractible. 
In other words the manifold $M$ is contractible if there exists a vector field $u$ on $M$ which generates 
a flow $\eta_t: M\rightarrow M$, with $t\in[0,1]$, that gradually and smoothly shrinks the whole manifold 
$M$ to the point $a$, i.e. $\eta_0 = {\rm Id}_M$ and $\eta_1(x) = a$, $\forall x \in M$,
where the point $a$ is fixed and independent of $x$. For more details see \citet[][Sec. 1.6, pp. 33]{AMR88}
and \citet[][Sec. 9, pp. 192]{Fec06}.

\subsection{Tensors}
\label{ssec:CICF:Te}
Let $\{E_i\}_{i\in\N^*}$, $F$ be finite-dimensional vector spaces.
Let $\mathcal{L}^k(E_1,\ldots,E_k;F)$
be the vector space of continuous $k$-multilinear maps of $E_1\times \ldots \times E_k$ to $F$.
The special case of the linear form on $E$, i.e.  $\mathcal{L}(E,\R)$, is denoted $E^\ast$, the dual space of $E$.
If $\{e_1, . . . , e_d\}$ is an ordered basis of $E$, there
is a unique ordered basis of $E^\ast$, the dual basis $\{e^1, \ldots, e^d\}$, such that
$\langle e^j , e_i \rangle := e^j(e_i)= \delta_i^j$  where $\delta_i^j=1$ if  $i=j$
and $0$ otherwise. Here $\langle \cdot , \cdot \rangle$ denotes the natural pairing between
$E$ and $E^\ast$. Furthermore, for each $v\in E$,  $v=\langle e^i , v \rangle e_i$ and
for each and $\alpha \in E^\ast$, $\alpha=\langle \alpha , e_i \rangle e^i$. 

For a vector space $E$ we define
\[
\mathrm{T}_p^q(E) = \mathcal{L}^{q+p}(E^\ast,\ldots ,E^\ast,E, \ldots,E;\R),
\]
($q$ copies of $E^\ast$ and $p$ copies of $E$). Elements of $\mathrm{T}_p^q (E)$ are called tensors on $E$,
contravariant of order $q$ and covariant of order $p$; or simply of type $(q,p)$. Given
$ \Uptheta_1\in \mathrm{T}_{p_1}^{q_1} (E)$ and $ \Uptheta_2\in \mathrm{T}_{p_2}^{q_2} (E)$, the tensor product of
$ \Uptheta_1$ and $ \Uptheta_2$ is the tensor $ \Uptheta_1\otimes  \Uptheta_2\in \mathrm{T}_{p_1+p_2}^{q_1+q_2} (E)$ defined by
\begin{multline*}
( \Uptheta_1 \otimes  \Uptheta_2)(\alpha^1, \ldots, \alpha^{q_1}, \beta^1, \ldots, \beta^{q_2}, f_1,
\ldots, f_{p_1}, g_1, \ldots, g_{p_2})\\=  \Uptheta_1(\alpha^1, \ldots, \alpha^{q_1}, f_1,
\ldots, f_{p_1}) \Uptheta_2(\beta^1, \ldots, \beta^{q_2}, g_1, \ldots, g_{p_2}),
\end{multline*}
where $\alpha^j,\ \beta^j \in E^\ast$, and $f_j,\, g_j \in E$. The natural basis of $\mathrm{T}_p^q(E)$
of dimension $d^{p+q}$ is given by
\[
\{e_{i_1}\otimes\ldots\otimes e_{i_q}\otimes e^{j_1}\otimes\ldots\otimes e^{j_p}\ | \
{i_1}, \ldots,{i_q}, {j_1}, \ldots {j_p}=1, \ldots, d
\}.
\]
In this basis any tensor $ \Uptheta\in \mathrm{T}_p^q(E)$ reads
\[
 \Uptheta= \Uptheta_{j_1\ldots j_p}^{i_1\ldots i_q}e_{i_1}\otimes\ldots\otimes e_{i_q}\otimes e^{j_1}\otimes\ldots\otimes e^{j_p},
\]
where the components of $ \Uptheta$ are given by
\[
 \Uptheta_{j_1\ldots j_p}^{i_1,\ldots i_q}= \Uptheta(
e^{i_1}, \ldots, e^{i_q}, e_{j_1}, \ldots, e_{j_p}
).
\]
We refer to \citet[][Sec. 5.1, pp. 341]{AMR88} for the definition of standard operations
(linear combination, contraction, contracted product, interior product, change of basis formula,
tensoriality criterion, ...) on tensors.

Let $M$ be a manifold and $TM$ its tangent bundle. We call
$\mathrm{T}_p^q(M):=\mathrm{T}_p^q(TM)=\cup_{a\in M}\mathrm{T}_p^q(TM_a)$
the vector bundle of tensors contravariant of order $q$ and covariant of order $p$, or simply
of type $(q,p)$. We identify $\mathrm{T}_0^1(M)$ with the tangent bundle $TM$ and  call
$\mathrm{T}_1^0(M)$ the cotangent bundle of $M$, also denoted $T^\ast M$ (i.e. the set of linear forms
on $TM$). The zero section
of $\mathrm{T}_p^q(M)$ is identified with $M$. Recall that a section of a vector bundle assigns
to each base point $a\in M$ a vector in the fiber $\pi^{-1}(a)$ over $a$  and the addition and
scalar multiplication of sections takes place within each fiber. In the
case of $\mathrm{T}_p^q(M)$ these vectors are called tensors. The $\mathscr{C}^\infty$ sections
of $E$ are denoted by $\Gamma^\infty(E)$. Recall that a vector field on $M$ is a $\mathscr{C}^\infty$ section
of $TM$, i.e. an element of  $\Gamma^\infty(TM)$. Therefore a tensor field of type $(q,p)$ on
a manifold $M$ is a $\mathscr{C}^\infty$ section of $\mathrm{T}_p^q(M)$. We denote by $\mathcal{T}_p^q(M)$
the set  $\Gamma^\infty(\mathrm{T}_p^q(M))$. A covector field or a differential $1$-form is an
element of $\mathcal{T}_1^0(M)$.

For the tangent bundle $TM$, a natural chart is obtained by taking the
vector bundle (or tangent) map $T\phi:TM\rightarrow T\R^d=\R^d$, where
$\phi$ is an admissible chart of $M$.
This in turn induces a tensor bundle map
$(T\phi)_\ast:\mathrm{T}_p^q(M)\rightarrow \mathrm{T}_p^q(\R^d)$,
which constitutes a natural chart  on $\mathrm{T}_p^q(M)$. Indeed
let  $\phi: U\ni a \rightarrow x=\phi(a) \in U'\subset \R^d$ a chart on $M$.
Let $\{e_i\}_{1\leq i\leq d}$ (resp. $\{e^i\}_{1\leq i\leq d}$)
be a (resp. dual) basis of $\R_x^d$. Then 
$\partial/\partial {a^i}=\phi^\ast e_i =(T\phi)^{-1}\circ e_i \circ \phi = (\partial a^j/\partial x^i) e_j$ is
a basis of  $\mathcal{T}_0^1(U)$. The vector field $\partial/\partial {a^i}$ corresponds to the differentiation
$f \mapsto \partial f/\partial {a^i}$. 
In the same way the $1$-forms $da^i=\phi^\ast e^i=(\partial x^i/\partial a^j) e^j$
is a basis of $\mathcal{T}_1^0(U)$. Since
\[
\langle da^i, \partial/ \partial{a^j}\rangle :=
da^i(\partial/ \partial{a^j})=\frac{\partial x^i}{\partial a^l}e^l\left(\frac{\partial a^k}{\partial x^j}e_k \right)=
\frac{\partial x^i}{\partial a^l}\frac{\partial a^k}{\partial x^j}e^l(e_k)=
\frac{\partial x^i}{\partial a^l}\frac{\partial a^k}{\partial x^j}\delta_k^l=\frac{\partial x^i}{\partial x^j}=\delta_j^i,
\]
$\{da^i\}_i$ is the dual basis of $\{\partial/\partial {a^i}\}_i$ at every point of $U$. Let
\[
 \Uptheta_{j_1\ldots j_p}^{i_1\ldots i_q}= \Uptheta(da^{i_1}, \ldots, da^{i_q},\partial/ \partial {a^{j_1}},
\ldots, \partial/\partial {a^{j_p}}) \in \mathcal{F}(U),
\]
where $\mathcal{F}(U)$ is the set of mappings from $U$ into $\R$ that are of class $\mathscr{C}^\infty$.
Then at every point $a$ of $U$ the coordinate expression of a $(q,p)$-tensor field $ \Uptheta\in \mathcal{T}_p^q(M)$ is
\[
 \Uptheta_{|_{U}}= \Uptheta_{j_1\ldots j_p}^{i_1\ldots i_q}(a)\frac{\partial}{\partial {a^{i_1}}} \otimes\ldots\otimes \frac{\partial}{\partial{a^{i_q}}}
\otimes da^{j_1}\otimes \ldots \otimes da^{j_p}.
\]
For more details  see \citet[][Sec. 5.2, pp. 352]{AMR88},
\citet[][Sec. 2.5, pp. 47]{Fec06} and
\citet[][Sec. III.B.1, pp. 117 and Sec. III.B.4, pp. 135]{CDD77}.

\subsection{Riemannian manifolds}
\label{ssec:CICF:RM}
Sometimes when dealing with manifolds it is useful to quantify
geometric notions such as length, angles and volumes.  All such
quantities are expressed by means of the lengths of tangent vectors,
that is, as the square root of a positive definite quadratic form given
on every tangent space.

A Riemannian manifold
is a differentiable manifold $M$ together with a differentiable $2$-covariant tensor field
$g\in \mathcal{T}_2^0(M)$, called the metric tensor, such that: i) $g$ is symmetric, ii) for each $a\in M$, the 
bilinear form ${g_a}$ (this notation emphasises that $g$ is evaluated in $a$) is non-degenerate,
i.e. $g_a(v,w)=0$ for all $v\in TM_a$ if and only if $w=0$. Such a manifold is said to possess
a Riemannian structure. A Riemannian manifold (Riemannian structure) is called proper if
$g_a$ is a positive definite quadratic form  on every tangent space, i.e. 
$g_a(v,v)>0, \quad \forall v \in TM_a, \ \ v\neq 0, \ \ a\in M$. Otherwise the manifold is
called pseudo-Riemannian or is said to possess an indefinite metric. The tensor $g$ allows one
to define a metric on $M$ for measuring distances between two points on $M$. The Riemannian
metric is given by the infinitesimal line element $ds^2$ which is defined by the metric tensor $g$:
\[
ds^2=g=g_{ij}da^ida^j=g_{ij}(a)da^i\otimes da^j.
\]
The tensor $g$ endows each tangent vector space $TM_a$ with an inner or scalar product, 
$(\cdot, \cdot)_{g_a}$ called
also Riemannian metric and defined by: $\forall a \in M$
\[
\begin{tabular}{rcll}
$(\cdot,\cdot)_{g_a} :$&$TM_a \ \ \times \  \ TM_a$&$\rightarrow $&$\R $   \\ 
  &$\left(v=v^i\frac{\partial}{\partial a^i},\ w=w^i\frac{\partial}{\partial a^i}\right)$
  &$ \mapsto   $&$ (v,w)_{g_a}=g_{ij}(a)v^i(a)w^j(a)$,
\end{tabular}
\]
where the notation $(\cdot,\cdot)_{g_a}$ is to emphasise that the quadratic form is local, i.e. 
evaluated at the point $a\in M$; but most of the time it is omitted to simplify the notation into
 $(\cdot,\cdot)_{g}$. 
The components of $g$  are differentiable on $M$ and are given by
\[
g_{ij}(a)=\left(\frac{\partial}{\partial a^i},\frac{\partial}{\partial a^j}\right)_{g_a}
=\frac{\partial a^k}{\partial x^i}\frac{\partial a^l}{\partial x^j}(e_k,e_l).
\]
where $(\cdot,\cdot)$ denotes the usual scalar product in the Euclidean space, i.e. induced
by the constant diagonal metric $\delta_{ij}$, with unity on the diagonal. 
Therefore, using the inner product $(\cdot, \cdot)_{g}$, we get an isomorphism between
the tangent bundle $TM$ and the cotangent bundle $T^\ast M$. In particular, it induces 
an isomorphism of the spaces of sections, which is called the raising operator
 $(\cdot)^\sharp: \mathcal{T}_1^0(M) \rightarrow  \mathcal{T}_0^1(M) $, with its inverse, named
the lowering operator  $(\cdot)^\sharp: \mathcal{T}_0^1(M) \rightarrow  \mathcal{T}_1^0(M)$.
More precisely, such operators are defined by  
\[
\begin{tabular}{rll}
$(\cdot)^\sharp : \mathcal{T}_1^0(M)  $ &$  \rightarrow  $& $\mathcal{T}_0^1(M) $   \\ 
  $ \alpha $& $ \mapsto $  &$\alpha^\sharp=( \alpha_ida^i)^\sharp=(\alpha^\sharp)^i\frac{\partial}{\partial a^i},
  \ \ (\alpha^\sharp)^i=g^{ij}\alpha_j$,
\end{tabular}
\]
\[
\begin{tabular}{rll}
$(\cdot)^\flat : \mathcal{T}_0^1(M)  $ &$\rightarrow  $& $\mathcal{T}_1^0(M) $   \\ 
  $ v $& $ \mapsto $   &$ v^\flat=\left(v^i \frac{\partial}{\partial a^i}\right)^\flat= (v^\flat)_i da^i,
  \ \ (v^\flat)_i=g_{ij}v^j$,
\end{tabular}
\]
where $g_{ik}g^{kj}=\delta_{i}^j$. For more details we refer the reader to \citet[][Sec. V.A.1, pp. 285]{CDD77}.

\subsection{Pullback and pushforward}
\label{ssec:CICF:PP}
Let $M$, $N$ and $P$ be differentiable manifolds. Let $\varphi: M\ni a\rightarrow x=\varphi(a)\in N$ and $\psi: N\rightarrow P$
be diffeomorphisms. The pullback of $ \Uptheta\in \mathcal{T}_p^0(N)$ by $\varphi$ is defined by
\[
(\varphi^\ast  \Uptheta)(a)(v_1,\ldots,v_p)= \Uptheta(\varphi(a)) (T_a\varphi (v_1), \ldots, T_a\varphi(v_p)),
\]
for all $a\in M$, and $v_1,\ldots, v_p \in TM_a$. The map $T_a\varphi:TM_a\rightarrow TN_{x=\varphi(a)}$ is the tangent map
of $\varphi$ at $a\in M$, i.e. the Jacobian matrix $J_{\varphi}(a)=J(\varphi)(a)=(\partial \varphi /\partial a)(a)$.
The pullback $\varphi^\ast: \mathcal{T}_p^0(N) \rightarrow  \mathcal{T}_p^0(M)$ is a linear isomorphism, which
satisfies $\varphi^\ast( \Uptheta_1\otimes \Uptheta_2)=(\varphi^\ast \Uptheta_1)\otimes(\varphi^\ast \Uptheta_2)$ for any
$ \Uptheta_1\in \mathcal{T}_{p_1}^0(N)$ and $ \Uptheta_2\in
\mathcal{T}_{p_2}^0(N)$. The pullback, applied to the composition of two
maps, $\psi\circ\varphi$, satisfies the following rule: $(\psi\circ\varphi)^\ast=\varphi^\ast  \psi^\ast$. Since $\varphi$ is
a diffeomorphism, $\varphi^\ast$ is an isomorphism with inverse $(\varphi^\ast)^{-1}:=(\varphi^{-1})^{\ast}$.

The pushforward of $ \Uptheta\in \mathcal{T}_p^q(M)$ by $\varphi$ is defined by
\[
(\varphi_\ast  \Uptheta)(x)(\alpha^1,\ldots,\alpha^q,f_1,\ldots,f_p)
= \Uptheta(\varphi^{-1}(x))(\varphi^\ast \alpha^1,\ldots,\alpha^q,(T\varphi)^{-1}(f_1),\ldots,(T\varphi)^{-1}(f_p)),
\]
where $\alpha^i\in T^\ast N_x$ and $f_i\in TN_x$. Using the tensor bundle map
$(T\varphi)_\ast:\mathrm{T}_p^q(M)\rightarrow \mathrm{T}_p^q(N)$, the pushforward can be written
in compact form as $\varphi_\ast  \Uptheta:=(T\varphi)_\ast \circ  \Uptheta \circ \varphi^{-1}$.
The pushforward $\varphi_\ast: \mathcal{T}_p^q(M) \rightarrow  \mathcal{T}_p^q(N)$ is a linear isomorphism, which
satisfies $\varphi_\ast( \Uptheta_1\otimes \Uptheta_2)=(\varphi_\ast \Uptheta_1)\otimes(\varphi_\ast \Uptheta_2)$ for any
$ \Uptheta_1\in \mathcal{T}_{p_1}^{q_1}(M)$ and $ \Uptheta_2\in \mathcal{T}_{p_2}^{q_2}(M)$. The pushforward of map composition
verifies the following rule: $(\psi\circ\varphi)_\ast=\psi_\ast  \varphi_\ast$. Since $\varphi$ is
a diffeomorphism, $\varphi^\ast$ is an isomorphism with inverse $(\varphi_\ast)^{-1}:=(\varphi^{-1})_{\ast}$.
The pullback  of $ \Uptheta\in \mathcal{T}_p^q(N)$ by $\varphi$ is given by $\varphi^\ast \Uptheta=(\varphi^{-1})_\ast  \Uptheta$.
In other words we have $\varphi^\ast=(\varphi_\ast)^{-1}=(\varphi^{-1})_\ast$ and $\varphi_\ast=(\varphi^\ast)^{-1}=(\varphi^{-1})^\ast$.

For finite-dimensional manifolds, pullback and pushforward can be expressed in terms of coordinates. Setting
$m={\rm dim}(M)$ and $n={\rm dim}(N)$, the maps $x^j=\varphi^j(a^1,\ldots,a^m)$, with $j=1,\ldots,n$ denote
the local expression of the diffeomorphism $\varphi:M\rightarrow N$ relative to charts.
Taking into account that the tangent map $T\varphi$ of $\varphi$ is given locally by the
Jacobian matrix $J_\varphi=J(\varphi)=(\partial\varphi /\partial a)$, we obtain the following coordinate expressions
of the pushforward and the pullback.\\
If $ \Uptheta \in  \mathcal{T}_p^q(M)$ and $\varphi$ a diffeomorphism, the coordinates of the
  pushforward of $\varphi_\ast  \Uptheta$ are
  \begin{equation}
    \label{DefOfPushForward}
  (\varphi_\ast  \Uptheta)_{j_1\ldots j_p}^{i_1\ldots i_q}= \left(\frac{\partial x^{i_1}}{\partial a^{k_1}}\circ\varphi^{-1}\right)
  \ldots \left(\frac{\partial x^{i_q}}{\partial a^{k_q}}\circ\varphi^{-1}\right)
  \frac{\partial a^{l_1}}{\partial x^{j_1}} \ldots \frac{\partial a^{l_p}}{\partial x^{j_p}}
    \Uptheta_{l_1\ldots l_p}^{k_1\ldots k_q} \circ \varphi^{-1}.
   \end{equation}
If $ \Uptheta \in  \mathcal{T}_p^q(N)$ and $\varphi$ is a diffeomorphism, the coordinates of the
  pullback of $\varphi^\ast  \Uptheta$ are
  \[
  (\varphi^\ast  \Uptheta)_{j_1\ldots j_p}^{i_1\ldots i_q}= \left(\frac{\partial a^{i_1}}{\partial x^{l_1}}\circ\varphi\right)
  \ldots \left(\frac{\partial a^{i_q}}{\partial a^{l_q}}\circ\varphi\right)
  \frac{\partial x^{k_1}}{\partial a^{j_1}} \ldots \frac{\partial x^{k_p}}{\partial a^{j_p}}
    \Uptheta_{k_1\ldots k_p}^{l_1\ldots l_q} \circ \varphi.
   \]
In particular, if $ \Uptheta \in  \mathcal{T}_p^0(N)$ the coordinates of the
  pullback of $\varphi^\ast  \Uptheta$ are
  \[
  (\varphi^\ast \Uptheta)_{j_1\ldots j_p}= 
  \frac{\partial x^{k_1}}{\partial a^{j_1}} \ldots \frac{\partial x^{k_p}}{\partial a^{j_p}}
    \Uptheta_{k_1\ldots k_p}\circ \varphi.
   \]
   If $v=v^i(\partial/\partial a^i) \in  \mathcal{T}_0^1(M)$
   (resp.  $\alpha=\alpha_id a^i \in  \mathcal{T}_1^0(M)$) then
   $\varphi_\ast(v^i(\partial/\partial a^i))=v^j(\partial x^i/\partial a^j) (\partial/\partial x^i)$
   (resp.  $\varphi^\ast(\alpha_i d a^i)=\alpha_j(\partial x^j/\partial a^i) d a^i$).
   Therefore, using the map $g:N\rightarrow \R$, we obtain
   \[
   (\varphi_\ast v)g
   = \left(v^j\frac{\partial x^i}{\partial a^j}\frac{\partial }{\partial x^i}\right)g
   =v^j\frac{\partial x^i}{\partial a^j}\frac{\partial g}{\partial x^i}
   = v^j\frac{\partial}{\partial a^j}g(\varphi(a))
   =\left(v^j\frac{\partial}{\partial a^j}\right)\varphi^\ast g= v(\varphi^\ast g).
   \] 
From the above formula we see that the pullback of covariant tensors can be defined even
for maps that are not diffeomorphisms but only differentiable maps, i.e.  of class $\mathscr{C}^1$
\citep[see, e.g.,][Sec. 5.2, pp. 355; see also Sec. 3.1, pp. 54 in \citet{Fec06}]{AMR88}.

\subsection{Lie derivative}
\label{ssec:CICF:LD}

Concepts of Lie derivative and Lie advection have been presented in Sec.~\ref{sec:CLT}, where
the Lie-derivative theorem has also been stated. Here we give additional properties of
the Lie differentiation process.

From an algebraic point of view, the local coordinate expression of 
the Lie derivative $\mathsterling_{v}:\mathcal{T}_p^q(M) \rightarrow \mathcal{T}_p^q(M)$ of an
arbitrary tensor $ \Uptheta\in \mathcal{T}_p^q(M)$ is
\citep[see, e.g.,][Sec. 5.3, pp. 359; see also Sec. 4.3, pp. 72 in \citet{Fec06}]{AMR88}
\[
\mathsterling_{v} \Uptheta=(\mathsterling_{v} \Uptheta)_{j_1\ldots j_p}^{i_1\ldots i_q}
\frac{\partial}{\partial {a^{i_1}}} \otimes\ldots\otimes \frac{\partial}{\partial{a^{i_q}}}
\otimes da^{j_1}\otimes \ldots \otimes da^{j_p},
\]
where
\[
(\mathsterling_{v} \Uptheta)_{j_1\ldots j_p}^{i_1\ldots i_q}=\, v^\ell\partial_\ell \Uptheta_{j_1\ldots j_p}^{i_1\ldots i_q} \quad 
- \  \Uptheta_{j_1\ldots j_p}^{ki_2\ldots i_q}\partial_k v^{i_1} \, -\,  (\mbox{all upper indices})
\quad +\  \Uptheta_{l j_2\ldots j_p}^{i_1\ldots i_q}\partial_{j_1} v^{l} \, +\,  (\mbox{all lower indices}).
\]
Moreover the Lie derivative is a linear operator, a derivation (i.e. it satisfies the Leibniz rule):
\[
 \mathsterling_{v}(\omega + \lambda \theta)=\mathsterling_{v}\omega +
 \lambda \mathsterling_{v} \theta, \quad \mathsterling_{v}(\omega
 \otimes \gamma)= \mathsterling_{v}\omega \otimes \gamma +
 \omega\otimes \mathsterling_{v}\gamma, \quad\lambda\in \R,\,\, v \in \mathcal{T}_0^1(M), \,\,\gamma \in \mathcal{T}_s^r(M) ,\,\, \theta, \, \omega \in \mathcal{T}_p^q(M).
\] 
Furthermore, the Lie derivative  is natural with respect to the pushforward and pullback by any
diffeomorphism $\varphi:M\rightarrow N$, in the following sense
\[
\varphi_{t\ast}
 \mathsterling_{v} = \mathsterling_{\varphi_{t\ast} v}\varphi_{t\ast},
 \quad \varphi_{t}^{\ast} \mathsterling_{v} =
 \mathsterling_{\varphi_{t}^{\ast} v}\varphi_{t}^{\ast}.
\]
\subsection{Permutations, generalised Kronecker symbols and determinants}
\label{ssec:CICF:PKD}
The set $\mathfrak{S}_k$ is the permutation group on $k$ elements,
which consists of all bijections $\sigma:\{1,\ldots,k\}\rightarrow\{1,\ldots,k\} $,
usually given in the form a table
\[
\left(
\begin{tabular}{ccc}
$1$ &$\ldots $ & $k$\\
$\sigma(1)$ & $\ldots$ & $\sigma(k)$
\end{tabular}
\right),
\]
with the structure of a group under composition of maps. A transposition is a
permutation which swaps two elements of $\{1,\ldots,k\}$. 
A permutation is even (resp. odd) when it can be written as the product 
of an even (resp. odd) number of transpositions. When a permutation is even (resp.
odd) $\mathrm{sign}\, \sigma =+1$ (resp. $\mathrm{sign}\, \sigma =-1$) and 
$\mathrm{sign}(\sigma \circ  \tau)=(\mathrm{sign}\, \sigma)(\mathrm{sign}\,  \tau)$.
The dimension of $\mathfrak{S}_k$ is ${\rm dim}(\mathfrak{S}_k)=k!$. 

Let $\delta_{ij}$, $\delta_j^i$ and $\delta^{ij}$ be the first Kronecker symbols
defined by
\[
\delta_{ij}=\delta_j^i=\delta^{ij}=\left\{
\begin{tabular}{lll}
  $0$ & $\mbox{if}$ & $i\neq j$ \\
  $1$ & $\mbox{if}$ & $i = j$.
\end{tabular}  
\right.
\]
The generalised Kronecker symbol $\delta_{j_1\ldots j_p}^{i_1\ldots i_p}$ (also noted $\varepsilon_{j_1\ldots j_p}^{i_1\ldots i_p}$)
is defined by
\[
\delta_{j_1\ldots j_p}^{i_1\ldots i_p}=\left\{
\begin{tabular}{rl}
  $0$ & $\mbox{if } (i_1 \ldots i_p) \mbox{ is not a permutation of }  (j_1 \ldots j_p)$ \\
  $+1$ &  $\mbox{if } (i_1 \ldots i_p) \mbox{ is an even permutation of }  (j_1 \ldots j_p)$ \\
  $-1$ &  $\mbox{if } (i_1 \ldots i_p) \mbox{ is an odd permutation of }  (j_1 \ldots j_p)$.
\end{tabular}  
\right.
\]
Using the Laplace
expansion of determinant, the generalised Kronecker symbol $\delta_{j_1\ldots j_p}^{i_1\ldots i_p}$
can be recast in different forms:
\begin{eqnarray*}
\delta_{i_1\ldots i_{p}}^{j_1\ldots j_{p}}&=&
\left |
\begin{tabular}{lll}
  $\delta_{i_1}^{j_1}$ & $\ldots $ & $ \delta_{i_p}^{j_1} $\\
  $\vdots$ & $\ddots$ & $\vdots$\\
   $\delta_{i_1}^{j_p}$ &$\ldots $ &  $ \delta_{i_p}^{j_p} $
\end{tabular}
\right|  
  =\sum_{k=1}^p(-1)^{p+k}
  \delta_{i_k}^{j_p}\delta_{i_1\ldots \widehat{i}_k \ldots i_{p}}^{j_1\ldots j_{k}\ldots\widehat{j}_p},\\
  &=&\sum_{\sigma \in \mathfrak{S}_p}{\rm sign}(\sigma) \delta_{j_{\sigma(1)}}^{i_1}\ldots \delta_{j_{\sigma(p)}}^{i_p}
  =\sum_{\sigma \in \mathfrak{S}_p}{\rm sign}(\sigma) \delta_{j_1}^{i_{\sigma(1)}}\ldots \delta_{j_p}^{i_{\sigma(p)}},
\end{eqnarray*}
where the hat character $\,\hat{}\,$ indicates an index that is omitted from the sequence.
Moreover, the generalised Kronecker symbol $\delta_{j_1\ldots j_p}^{i_1\ldots i_p}$ satisfies the properties \citep[][Sec. 5.6, pp. 107]{Fec06}
\[
\frac{1}{p!} \delta_{k_1\ldots k_p}^{i_1\ldots i_p}\delta_{j_1\ldots j_p}^{k_1\ldots k_p}=\delta_{j_1\ldots j_p}^{i_1\ldots i_p}, \quad
\mbox{and} \quad 
\delta_{j_1\ldots j_p i_{p+1}\ldots i_{q}}^{i_1\ldots i_p i_{p+1}\ldots i_{q}}=\frac{(d-p)!}{(d-q)!}\delta_{j_1\ldots j_p}^{i_1\ldots i_p}.
\]
We also define the second Kronecker symbols $\varepsilon_{j_1\ldots j_p}$ and  $\varepsilon^{i_1\ldots i_p}$ by
\[
\varepsilon_{j_1\ldots j_p}=\delta_{j_1\ldots j_p}^{1\ldots p}, \quad 
\varepsilon^{i_1\ldots i_p}=\delta_{1\ldots p}^{i_1\ldots i_p} \ \
\mbox{and thus } \ \
\delta_{j_1\ldots j_p}^{i_1\ldots i_p}=\frac{1}{(d-p)!}\varepsilon^{i_1\ldots i_p k_{p+1}\ldots k_{d}}\varepsilon_{j_1\ldots j_p k_{p+1}\ldots k_{d}}.
\]
Finally let $d={\rm dim}(M)$, and $\varphi:M\rightarrow M$  be of class $\mathscr{C}^1$.
The determinant of the linear mapping (tangent map at the point $a$) $T_a\varphi:TM_a\rightarrow TM_a$, 
is noted $\mathrm{det}(T_a\varphi)=\mathrm{det}(\partial \varphi/\partial a)$ and is given by
\begin{multline*}
\mathrm{det}(T_a\varphi)
= \sum_{\sigma \in \mathfrak{S}_d}{\rm sign}(\sigma)  \frac{\partial \varphi^1}{\partial a^{\sigma(1)}}
\ldots  \frac{\partial \varphi^d}{\partial a^{\sigma(d)}}= 
\sum_{\sigma \in \mathfrak{S}_d}{\rm sign}(\sigma)  \frac{\partial \varphi^{\sigma(1)}}{\partial a^{1}}
\ldots  \frac{\partial \varphi^{\sigma(d)}}{\partial a^d}
= \varepsilon^{i_1\ldots i_d}  \frac{\partial\varphi^1}{\partial a^{i_1}}
\ldots  \frac{\partial\varphi^d}{\partial a^{i_d}}\\
=\varepsilon_{j_1\ldots j_d}  \frac{\partial\varphi^{j_1}}{\partial a^{1}}
\ldots  \frac{\partial\varphi^{j_d}}{\partial a^{d}}
=\frac{1}{d!}\varepsilon^{i_1\ldots i_d} \varepsilon_{i_1\ldots i_d} \mathrm{det}(T_a\varphi)
=\frac{1}{d!}\varepsilon^{i_1\ldots i_d} \varepsilon_{j_1\ldots j_d}  \frac{\partial\varphi^{j_1}}{\partial a^{i_1}}
\ldots  \frac{\partial\varphi^{j_d}}{\partial a^{i_d}}
=\frac{1}{d!}\delta^{i_1\ldots i_d}_{j_1\ldots j_d}  \frac{\partial\varphi^{j_1}}{\partial a^{i_1}}
\ldots  \frac{\partial\varphi^{j_d}}{\partial a^{i_d}}.
\end{multline*}

The inverse matrix components of an invertible matrix $A$ is given by $(A^{-1})_i^j=({\rm det}(A)^{-1})\Delta_i^j $,
where $\Delta_i^j$ is the $(i,j)$th minor, i.e. the determinant of a matrix which it is obtained from $A$
when the $i$th row and $j$th column are deleted. In other words we have ${\rm det}(A)\delta_k^j= A_i^j\Delta_k^i=\Delta_i^jA_k^i$.
Therefore, we obtain
\[
\frac{\partial\,{\rm det}(A)}{\partial A_i^j}=\Delta_j^i={\rm det}(A)(A^{-1})_j^i
\quad \mbox{so that} \quad d({\rm det}(A))={\rm det}(A) {\rm Tr}(A^{-1}dA),
\]
where ${\rm Tr(A)}$ denotes the trace of $A$, i.e. $\sum_i A_i^i$. Now, we consider the metric
tensor $g\in \mathcal{T}_2^0$ which can be identified to a matrix. We define the minor $a_{ij}:={\rm g} g^{ij}$,
with ${\rm g}=\sqrt{{\rm det}(g_{ij})}$. It then follows that
the differential of the determinant ${\rm g}$ is  $d{\rm g}=a_{ij}dg_{ij}={\rm g} g^{ij}dg_{ij}$. 
Furthermore, using partial derivatives, the differential of  ${\rm g}$ is 
$d{\rm g}=\partial_k {\rm g}da^k={\rm g} g^{ij}\partial_k g_{ij}da^k$, from which we infer by identification that
\[
\partial_k{\rm g}={\rm g}g^{ij}\partial_kg_{ij}=-{\rm g}g_{ij}\partial_kg^{ij}.
\]

\subsection{Exterior algebra and differential forms}
\label{ssec:CICF:EADF}
Let $E$ be a finite-dimensional vector space.
The space $\bigwedge^p(E)$, is the subspace of all
skew symmetric elements of $\mathcal{L}^p(E)$ or $\mathrm{T}_p^0(E)$,
i.e. all antisymmetric covariant $p$-tensors on $E$. An element of  $\bigwedge^p(E)$
is called an exterior $p$-form. The exterior product $\wedge$ (wedge or Grassmann product)
of a $p$-form and a $q$-form is a mapping
\[
\begin{tabular}{rclll}
$\wedge :$ &$\bigwedge^p(E) \ \times \ \bigwedge^q(E)$  &$  \rightarrow  $& $\bigwedge^{p+q}(E)$   \\ 
 & $(\alpha, \ \beta )$& $ \mapsto $   & $\alpha\wedge \beta$
\end{tabular}
\]
with $\alpha\wedge \beta $ defined by
\[
(\alpha\wedge \beta) (v_1,\ldots,v_{p+q})=\frac{1}{p!q!}\sum_{\sigma \in \mathfrak{S}_{p+q}}  
\mathrm{sign}(\sigma) \alpha(v_{\sigma(1)},\ldots,v_{\sigma(p)})
\beta(v_{\sigma(p+1)},\ldots,v_{\sigma(p+q)}),
\]
where $v_i\in E$ and $\mathfrak{S}_p$ is the permutation group on $p$ elements. 
Componentwise it is defined as
\[
(\alpha\wedge \beta)_{i_1\ldots i_{p+q}}=\frac{1}{p!q!}\delta_{i_1\ldots i_{p+q}}^{j_1\ldots j_pk_1\ldots k_q}
\alpha_{j_1\ldots j_{p}} \beta_{k_1\ldots k_{q}}.
\]
In particular, if $\alpha$ and $\beta$ are $1$-forms then $\alpha
\wedge \beta=\alpha \otimes \beta - \beta \otimes \alpha$. It follows
from the definition that the exterior product is $i)$ associative:
$(\alpha\wedge \beta) \wedge \gamma = \alpha\wedge (\beta\wedge
\gamma)$, $ii)$ bilinear: $\alpha\wedge (\beta +\gamma) = \alpha\wedge
\beta +\alpha\wedge \gamma$ and $\lambda (\alpha\wedge
\beta)=\lambda\alpha\wedge \beta=\alpha\wedge \lambda\beta$, with
$\lambda\in \R$, $iii)$ not commutative in general: $\alpha\wedge
\beta=(-1)^{pq}\beta\wedge \alpha$ if $\alpha\in\bigwedge^p(E)$,
$\beta\in\bigwedge^q(E)$. From the property $iii)$ it follows that
$(\wedge\, \alpha)^k$ is identically zero if the degree of $\alpha$ is
odd; but not otherwise. If $E$ is finite dimensional with $d={\rm
  dim}(E)$, then for $p>d$, $\bigwedge^p(E)=\{ 0\}$. Indeed the only
non zero components of a totally antisymmetric covariant $p$-tensor
are those in which all indices are different, a situation which can
never exist if $p>d$. For $0< p \leq d$, $ \bigwedge^p(E)$ has
dimension $d!/(p!(d-p)!)$. If $\{e_1, . . . , e_d\}$ is an ordered
basis of $E$ and its dual basis $\{e^1, \ldots, e^d\}$, a basis for $
\bigwedge^p(E)$ is
\[
\{e^{i_1}, \ldots, e^{i_p} \ | \ 1\leq i_1 < i_2<\ldots <i_p\leq d\}.
\]
Therefore any $\alpha\in  \bigwedge^p(E)$, can be expanded as
\[
\alpha= \sum_{i_1<\ldots < i_p}
\alpha_{i_1\ldots i_p}e^{i_1}\wedge \ldots \wedge e^{i_p}
=\frac{1}{p!}
\alpha_{i_1\ldots i_p}e^{i_1}\wedge \ldots \wedge e^{i_p}=
\frac{1}{p!}\delta_{j_1\ldots j_p}^{i_1\ldots i_p}
\alpha_{i_1\ldots i_p}e^{j_1}\otimes \ldots \otimes e^{j_p}.
\]
Given the tangent vector bundle $TM$ of a manifold $M$, we can construct fiberwise the vector bundle
$\bigwedge^p(M)$ of exterior differential $p$-form on the tangent spaces of $M$, as
\[
\bigwedge\nolimits^{\! p}(M)=\bigwedge\nolimits^{\! p}(TM)=\cup_{a\in M} \bigwedge\nolimits^{\! p}(TM_a).
\] 
The field of exterior differential $p$-form on a manifold $M$, denoted
 $\Lambda^p(M)$, is defined as the $\mathscr{C}^\infty$ section of $\bigwedge^p(M)$,
i.e.  $\Lambda^p(M)=\Gamma^\infty(\bigwedge^p(M))$. We have the following identifications:
$\Lambda^1(M)=\mathcal{T}_1^0(M)$ and  $\Lambda^0(M)=\mathcal{F}(M)$, where
$\mathcal{F}(M)$ is the set of mappings from $M$ into $\R$ that are of class
$\mathscr{C}^\infty$. As for the definition of tensors on a manifold, 
given $(U,\phi)$, an admissible local chart on $M$, the local expression on
$U$ of $\alpha \in \bigwedge^k(M)$ is given by
\[
\alpha_{|_U}= \sum_{i_1<\ldots < i_p}
\alpha_{i_1\ldots i_p}(a)da^{i_1}\wedge \ldots \wedge da^{i_p}=
\frac{1}{p!}
\alpha_{i_1\ldots i_p}(a) da^{i_1}\wedge \ldots \wedge da^{i_p}
=\frac{1}{p!}\delta_{j_1\ldots j_p}^{i_1\ldots i_p} 
\alpha_{i_1\ldots i_p}(a) da^{j_1}\otimes \ldots \otimes da^{j_p}.
\]
The differential $p$-form is of class $\mathscr{C}^k$, when 
the component maps $\alpha_{i_1\ldots i_p}: U\ni a \rightarrow \alpha_{i_1\ldots i_p}(a) \in \R$
are $k$ times continuously differentiable on $U$ or are differentiable functions of $a$
of class  $\mathscr{C}^k(U)$.

Pullback and pushforward of $p$-forms are just  special 
cases of general definitions given for tensors (see Appendix~\ref{ssec:CICF:Te}) since a $p$-form field 
is a totally antisymmetric covariant $p$-tensor field. Moreover we have the following
properties. Let $\varphi:M\rightarrow N$ be of class $\mathscr{C}^1$. Then
$\varphi^\ast:\Lambda^k(N)\rightarrow\Lambda^k(M)$ is
a homeomorphism of differential algebras, that is
\[
\varphi^\ast(\alpha\wedge \beta) = \varphi^\ast\alpha\wedge \varphi^\ast\beta, \quad
\varphi^\ast(\alpha +\lambda \gamma)=\varphi^\ast \alpha +\lambda \varphi^\ast \gamma,
\quad \alpha,\, \gamma\in \Lambda^p(N), \ \  \beta\in \Lambda^q(N), \ \ \lambda\in\R.
\]
Of course similar formulas hold also for the pushforward operator when $\varphi:M\rightarrow N$
is a diffeomorphism. The Lie derivative
is a derivative on $\Lambda^p(M)$, since it satisfies the Leibniz rule:
\[
\mathsterling_v(\alpha \wedge \beta)=\mathsterling_v\alpha \wedge \beta+ \alpha \wedge \mathsterling_v\beta, \quad
\alpha\in \Lambda^p(M), \ \  \beta\in \Lambda^q(M).
\]
From the definition of Lie derivative for tensors, the coordinate expression
for the Lie derivative $\mathsterling_v:\Lambda^p(M) \rightarrow \Lambda^p(M)$ 
of a $p$-form $\alpha $ is
\[
\mathsterling_v\alpha=\frac{1}{p!}v^l\partial_l\alpha_{i_1\ldots i_p} da^{i_1}\wedge \ldots \wedge da^{i_p}
+ \frac{1}{p!}\alpha_{i_1\ldots i_p} (
\partial_lv^{i_1}da^{l}\wedge da^{i_2}\wedge \ldots \wedge da^{i_p}
+ \ldots + 
 \partial_lv^{i_p}da^{i_1}\wedge \ldots \wedge da^{i_{p-1}}\wedge da^{l}).
\]
This can also be recast in a simpler form, which however is not antisymmetric, namely
\[
\mathsterling_v\alpha=\frac{1}{p!}
(v^k\partial_k\alpha_{i_1\ldots i_p} +p\alpha_{ki_2\ldots i_p}\partial_{i_1}v^{k})\,
da^{i_1}\wedge \ldots \wedge da^{i_p}.
\]
For more details we refer the reader to \citet[][Sec. 6.1, pp. 392; Sec. 6.3, pp. 417]{AMR88},
\citet[][Sec. IV.A.1, pp. 195]{CDD77} and \citet[][Sec. 5.3, pp. 102]{Fec06}.

\subsection{Exterior derivative and interior product}
\label{ssec:CICF:EDIP}
The exterior differentiation operator $d:\Lambda^p(M)\rightarrow\Lambda^{p+1}(M)$
maps a $p$-form $\alpha$ of class $\mathscr{C}^k$ into a $(p+1)$-form $d\alpha$
of class $\mathscr{C}^{k-1}$, called the exterior derivative of $\alpha$.
The operator $d$ is uniquely defined by the following properties: 
\begin{itemize}
\item[1.] $d$ is linear: $d(\alpha +\lambda \beta)=d\alpha +\lambda d\beta, \quad \lambda\in\R$,
  $\ \ \alpha,\, \beta\in\Lambda^p(M)$.
\item[2.] $d$ is an antiderivative; that is, $d$ is $\R$-linear and for
  $ \alpha\in\Lambda^p(M)$, and  $\beta\in\Lambda^q(M)$:
\item[] $ \ \  d(\alpha\wedge\beta)=d\alpha\wedge\beta + (-1)^p\alpha\wedge d\beta, \quad$ (``antiLeibniz'' product rule).
\item[3.] $d^2=dd=0$.
\item[4.] If $f\in\mathcal{F}(M)$ is a $0$-form, then $df$
  is the ordinary differential of $f$, i.e.  $df={\partial_i f} da^i$.
\item[5.] The operation $d$ is local: if $\alpha$ and $\beta$ coincide on an open set $U$,
  $d\alpha =d\beta$ on $U$; that is, the behaviour
\item[] $ \ \ $   of $\alpha$ outside $U$ does not affect
  $d\alpha_{|_U}$, i.e. $d(\alpha_{|_U})=(d\alpha)_{|_U}$.
\end{itemize}

Let $\varphi:M\rightarrow N$ be a diffeomorphism. Let $v\in\mathcal{T}_0^1(M)$,
 $\ \alpha\in \Lambda^p(M)$ and $\beta\in \Lambda^p(N)$. We have the properties: 
\[
\varphi^\ast (d\beta)=d(\varphi^\ast \beta), \quad
\varphi_\ast (d\alpha)=d(\varphi_\ast \alpha), \quad
d \mathsterling_v \alpha = \mathsterling_v d\alpha.
\]

The contracted multiplication or interior product (also called inner product) of a $p$-form $\alpha \in \Lambda^p(M)$
and a vector $v\in\mathcal{T}_0^1(M)$ is denoted ${\rm i}_v \alpha$. The operator
${\rm i}_v:\Lambda^{p}(M) \rightarrow \Lambda^{p-1}(M)$ is defined as follows.
\begin{itemize}
\item[1.]${\rm i}_v$ is an antiderivative; that is, ${\rm i}_v$ is $\R$-linear and for
  $ \alpha\in\Lambda^p(M)$, and  $\beta\in\Lambda^q(M)$:
\item[] $\, \ \  {\rm i}_v(\alpha\wedge\beta)=({\rm i}_v\alpha)\wedge\beta +
  (-1)^p\alpha\wedge ({\rm i}_v\beta), \quad $ (``antiLeibniz'' product rule).
\item[2.] ${\rm i}_vf=0$, $\ f\in\mathcal{F}(M)$; $\quad {\rm i}_vda^i=v^i$.  
\end{itemize}
Then by the ``antiLeibniz'' rule, the coordinate expression of the interior product of a $p$-form $\alpha$ is
\[
  {\rm i}_v\alpha=\frac{1}{(p-1)!}v^k\alpha_{ki_2\ldots i_p}da^{i_2}\wedge\ldots\wedge da^{i_p}
  =\sum_{i_1<\ldots< i_{p-1}}v^k\alpha_{ki_1\ldots i_{p-1}} da^{i_1}\wedge\ldots\wedge da^{i_{p-1}}.
\]
Let $\varphi:M\rightarrow N$ be a diffeomorphism.  Let $v, \, w \in\mathcal{T}_0^1(M)$,
$\ u\in\mathcal{T}_0^1(N)$, $\ \alpha\in \Lambda^p(M)$,
$\ \beta\in \Lambda^p(N)$,  $\ \gamma\in\Lambda^1(M)$,  and $f\in\mathcal{F}(M)$.
Using the commutator notation $[A,B]=AB-BA$, we have the properties:
\begin{itemize}
\item[1.] ${\rm i}_v^2={\rm i}_v{\rm i}_v=0$.
\item[2.] $\mathsterling_v \alpha= {\rm i}_v d\alpha + d{\rm i}_v\alpha$, (Cartan formula)
\item[3.] ${\rm i}_{fv}\alpha= f{\rm i}_v\alpha={\rm i}_vf\alpha$,
  $\ \ \ {\rm i}_v df=\mathsterling_vf$,
  $\ \ \ \mathsterling_{fv}\alpha=  f\mathsterling_{v}\alpha + df \wedge {\rm i}_v\alpha$.
\item[4.] $[\mathsterling_{v},{\rm i}_w]\alpha ={\rm i}_{[v,w]}\alpha$,
  $\ \ \  [\mathsterling_{v},\mathsterling_{w}]\alpha=\mathsterling_{[v,w]}\alpha$,
  $\ \ \ {\rm i}_v\mathsterling_{v} \alpha=\mathsterling_{v}{\rm i}_v \alpha$,
  $\ \ \ {\rm i}_v{\rm i}_wd \gamma=\mathsterling_v{\rm i}_w\gamma - \mathsterling_w{\rm i}_v\gamma -{\rm i}_{[v,w]}\gamma$.
\item[5.] 
  $ \varphi^\ast{\rm i}_u\beta={\rm i}_{\varphi^\ast u}  \, \varphi^\ast\beta$,
  $\ \ \ \varphi_\ast{\rm i}_v\alpha={\rm i}_{\varphi_\ast v}  \, \varphi_\ast\alpha$.
\end{itemize}  
The last formula of point $4.$, which expresses the exterior
derivative in terms of the Lie derivative, can be extended
to high-order form \citep[see, e.g.,][Sec. 6.4, pp. 431]{AMR88}. For more details about exterior
derivative and interior product, we refer the reader to
\citet[][Sec. IV.A.2 to Sec. IV.A.4, pp. 200]{CDD77} and  \citet[][Sec. 6.4, pp. 423]{AMR88}.

\subsection{Hodge dual operator and  exterior coderivative}
\label{ssec:CICF:HSECD}

Let $(M,g)$ be a $d$-dimensional Riemannian manifold with the volume form $\mu$.
The Hodge dual operator is defined as the unique isomorphism $\star:\Lambda^p(M)\rightarrow \Lambda^{d-p}(M)$,
which satisfies \citep[see, e.g.,][Sec. 6.2, pp. 411]{AMR88}
\begin{equation}
  \label{DefStarHodge}
\alpha \wedge \star \beta = (\!( \alpha, \beta )\!)_g \,\mu, \quad \alpha,\, \beta \in \Lambda^p(M).
\end{equation}
with
\[
(\!( \alpha, \beta )\!)_g=\frac{1}{p!}\,\alpha_{i_1\ldots i_p} \beta^{i_1\ldots i_p}=
\frac{1}{p!}\,\alpha_{i_1\ldots i_p} \beta_{j_1\ldots j_p}g^{i_1j_1} \ldots g^{i_pj_p}. 
\]
Using \eqref{DefStarHodge} with $\beta=da^{i_1}\wedge \ldots \wedge da^{i_p}$
and $\alpha=da^{j_1}\wedge \ldots \wedge da^{j_p}$ where $\{j_1,\ldots,j_p \}$
is the complementary set  of indices to  $\{j_{p+1},\ldots,j_d \}$, we obtain
\[
\star(da^{i_1}\wedge \ldots \wedge da^{i_p})
=\frac{1}{(d-p)!}\sqrt{\mathrm{g}} \varepsilon_{j_1\ldots j_d}g^{i_1j_1}\ldots g^{i_pj_p}da^{j_{p+1}}
\wedge \ldots \wedge da^{j_{d}}.
\]
Then the coordinate expression of the $(d-p)$-form $\star \alpha$, where $\alpha\in\Lambda^p(M)$,
is
\[
\star \alpha =\frac{1}{(d-p)!}(\star \alpha)_{i_1\ldots i_{d-p}} da^{i_{1}}
\wedge \ldots \wedge da^{i_{d-p}},
\]
with
\[
(\star \alpha)_{i_1\ldots i_{d-p}}
= \frac{1}{p!}\sqrt{\mathrm{g}}\,\varepsilon_{j_1\ldots j_p i_1\ldots i_{d-p}}\alpha^{j_1\ldots j_p} 
= \frac{1}{p!}\sqrt{\mathrm{g}}\, \varepsilon_{j_1\ldots j_p i_1\ldots i_{d-p}}
g^{j_1k_1}\ldots g^{j_pk_p} \alpha_{k_1\ldots k_p}.
\]
Let $\alpha,\, \beta\in\Lambda^p(M)$. Then the Hodge dual operator satisfies
$\alpha \wedge \star \beta=  \beta \wedge \star \alpha= (\!(\alpha,\beta )\!)_g\, \mu$,
$\star 1=\mu$, $\star \mu =1$, $\star\star\alpha= (-1)^{p(d-p)}\alpha$,
$(\!(\alpha,\beta )\!)_g= (\!(\star\alpha,\star\beta )\!)_g$. The Hodge dual is an
$\R$-linear operator, i.e. $\star(\alpha +\lambda \beta)=\star \alpha +\lambda \star\beta$,
$\lambda\in\R$.
In particular if $v$ and $w$ are two vectors of $\R^3$, and if $M=\R^3$, then
$v\times w =[\star(v^\flat\wedge w^\flat)]^\sharp$ and
$v\cdot w = \star(v^\flat\wedge \star w^\flat)$.

The codifferential operator (or exterior coderivative) $d^\star: \Lambda^{p}(M) \rightarrow \Lambda^{p-1}(M)$,
is an $\R$-linear operator which is defined by  \citep[see, e.g.,][Sec. 6.5, pp. 457]{AMR88}
\[
d^\star \alpha = (-1)^{d(p-1)+1}\star d\star\alpha.
\]
Since $d^2=0$, then $(d^\star)^2=d^\star \circ d^\star=0$.

Let $v\in \mathcal{T}_0^1(M)$ be a vector field on $M$. Then the unique function
$\mathrm{div}_\mu v\in \mathcal{F}(M)$ such that 
\[
\mathsterling_{v} \mu =:(\mathrm{div}_\mu v) \mu,
\]
is by definition called the divergence of $v$ \citep[see, e.g.,][Sec. 6.5, pp. 455]{AMR88}. Let $f, \, h\in \mathcal{F}(M)$,
with $f(a)\neq 0$, $\forall a\in M$. Then we have the formula
\[
\mathrm{div}_{f\mu}v=\mathrm{div}_{\mu}v+ f^{-1}{\mathsterling_v f}, \quad
\mathrm{div}_{\mu}(hv)= h\mathrm{div}_{\mu } v + \mathsterling_v h.
\]
Here, for a Riemannian manifold $(M,g)$ with an oriented chart $(a^1,\ldots,a^d)$ on $M$, the volume
form $\mu$ is given by \citep[see, e.g.,][Sec. 6.5, pp. 457]{AMR88}
\[
\mu(a)=\sqrt{\mathrm{g}(a)}\,da^1\wedge \ldots \wedge da^d
=
\frac{1}{d!}\delta_{i_1\ldots  i_d}^{1\ldots.. d}\sqrt{\mathrm{g}(a)}da^{i_1}\wedge\ldots\wedge da^{i_d},
 \quad \mbox{where}\ \ \mathrm{g}=\mathrm{det}(g_{ij}).
\]
Using the relation $i_{v}\mu=\star v^\flat$ and the Cartan formula we obtain
$(\mathrm{div}_\mu v)\mu:=\mathsterling_{v} \mu:= di_{v}\mu=
d\star v^\flat=-\star d^\star v^\flat=-(d^\star v^\flat)\star 1=
-(d^\star v^\flat)\mu$. Therefore
\[
\mathrm{div}_\mu v = -d^\star v^\flat=\frac{1}{\sqrt{g}}\partial_i(\sqrt{g}v^i).
\]

Let $\mathrm{Op}=\mathrm{Op}_ \Uptheta$ be an operator that depends on a tensor field $ \Uptheta$.
The operator $\mathrm{Op}$ is called natural with respect to the diffeomorphism $\varphi:M\rightarrow N$,\, if\,
$\varphi^\ast\mathrm{Op}_ \Uptheta=\mathrm{Op}_{\varphi^\ast \Uptheta}\varphi^\ast$. Of course we have
a similar definition with the pushforward operator since $\varphi_\ast=(\varphi^{-1})^\ast$. In the previous section, we have
seen that the Lie derivative, the interior product and the exterior derivative 
are natural with respect to diffeomorphisms.
For convenience we use now the following
notation: $\flat_g\equiv (\cdot)^\flat$, $\sharp_g\equiv (\cdot)^\sharp$, $\star_g\equiv\star$
and $d_g^\star\equiv d^\star$. All these operators are natural with respect to diffeomorphisms,
i.e.
\begin{equation}
  \label{naturalnessOf4operator}
  \varphi^\ast\flat_g = \flat_{\varphi^\ast g}\varphi^\ast, \quad
   \varphi^\ast\sharp_g = \sharp_{\varphi^\ast g}\varphi^\ast, \quad
   \varphi^\ast\ast_g = \ast_{\varphi^\ast g}\varphi^\ast, \quad
    \varphi^\ast d_g^\star = d_{\varphi^\ast g}^\star\varphi^\ast.
\end{equation}

Let $(M,g)$ and $(N,h)$ be two Riemanian manifolds, and $\varphi:M \rightarrow N$ a diffeomorphism.
The mapping $\varphi$ is called an isometry if $\varphi^\star h=g$ \citep[see, e.g.,][Sec. V.A.5, pp. 298]{CDD77}.
If $\varphi$ is an isometry, using \eqref{naturalnessOf4operator}, we then observe
that the commutators $[\varphi^\ast, \mathrm{Op}]$ with
$ \mathrm{Op}\in\{ \flat,\, \sharp,\, \star, \, d^\star \}$ vanish.

Let $\kappa\in \mathcal{T}_0^1(M)$. The vector field $\kappa$ on $(M,g)$
is called a Killing vector field if $\mathsterling_\kappa{g}=0$, that is it satisfies
the Killing equations
\begin{equation}
\label{eqn:killing:1}
(\mathsterling_{\kappa}g)_{ij}= \kappa^k\frac{\partial g_{ij}}{\partial a^k}
+ g_{kj}\frac{\partial \kappa^k}{\partial a^i} +  g_{ik}\frac{\partial \kappa^k}{\partial a^j}=0.
\end{equation}
Using the covariant derivative, the Killing equations \eqref{eqn:killing:1} can be recast as
\begin{equation}
\label{eqn:killing:2}
(\mathsterling_{\kappa}g)_{ij}= (\nabla_i \kappa^k)g_{kj}+ (\nabla_j \kappa^k)g_{ik}=  (\nabla_i \kappa_j)+ (\nabla_j \kappa_i)=0.
\end{equation}
Let us note that a Killing vector is always divergence-free, since the contraction
of the $2$-contravariant metric tensor $g^{ij}$ with the $2$-covariant tensor appearing in
\eqref{eqn:killing:2} gives $2\nabla_i\kappa^i=0$.
The Lie derivative theorem (see Sec.~\ref{sec:CLT}) implies that the vector field $\kappa$ generates a flow $f_t:M\rightarrow M$,
which leaves invariant the metric $g$, since $f_t^\ast g=g$. Thus
the flow $f_t$, induced by the  Killing vector field $\kappa$,
generates a family of isometries.
Since the operators $\flat_g$, $\sharp_g$, $\star_g$,  and $d_g^\star$ are natural
with respect to diffeomorphism we obtain $[f_t^\ast,\mathrm{Op}]=0$, with
$ \mathrm{Op}\in\{ \flat_g,\, \sharp_g,\, \star_g, \, d_g^\star \}$. Taking the derivative of
$[f_t^\ast,\mathrm{Op}]=0$ with respect to time $t$ at $t=0$, we obtain
\citep[see, e.g.,][Sec. 8.3, pp. 171]{Fec06}
\[
[\mathsterling_\kappa,\mathrm{Op}]=0, \quad \mathrm{Op}\in\{ \flat_g,\, \sharp_g,\, \star_g, \, d_g^\star \}.
\]

\subsection{Riemannian connection and covariant derivative}
\label{ssec:CICF:RCCD}
The velocity vector field lies in the tangent bundle, and so the
acceleration (the ``velocity of the velocity'') lies in the tangent bundle of the
tangent bundle.  The acceleration of the fluid is the rate of change of the
velocity vector field $v$ in the direction $v$ of a trajectory $t\rightarrow\varphi_t$
(with  $\dot{\varphi}_t(a)=v(t,\varphi_t(a))$) and is thus a special case
of what is called the directional derivative.
For clarity of this exposition and leaving apart physical considerations about acceleration,
we assume now that the vector field $v$ is time-independent. We also consider another time-independent
vector field $u$.
The directional derivative of $u$ in the direction of the vector field $v$,
which generates the flow $\varphi_t$, is noted $\nabla_v u$ and is defined by
\begin{equation}
  \label{DirDer}
\nabla_v u(a) = \lim_{t\rightarrow0} \frac{\mathcal{P}_\myparallel^\ast u(\varphi_t(a)) - u(a)}{t},
\end{equation}
where $\mathcal{P}_\myparallel^\ast u(\varphi_t(a))$ denotes a backward parallel transport
of the vector $u(\varphi_t(a))$. Since in an Euclidean space $\R^d$ all tangent spaces
are the same and identified with $\R^d$, the backward parallel transport $\mathcal{P}_\myparallel^\ast$
is just an infinitesimal rigid translation or shift, which alters neither length nor direction
of shifted vectors. However in the case of a manifold $M$, the vector $ u(\varphi_t(a))$ belongs
to the tangent space $TM_{\varphi_t(a)}$, while the vector $u(a)$ belongs to $TM_a$. Such vectors
lie in different vector spaces and thus their difference by using rigid translation has no meaning.
Therefore, on a manifold we need to introduce a rule of parallel transport (satisfying suitable requirements)
as a linear mapping connecting two different tangent spaces, namely
\[
\begin{tabular}{rcll}
$\mathcal{P}_{\myparallel, \varphi, a, b}: $ & $TM_{a}$ & $\rightarrow$ & $TM_{b=\varphi_t(a)}$\\
&$ w $ & $\mapsto$ & $\mathcal{P}_{\myparallel, \varphi, a, b} w$.
\end{tabular}
\]
Note that the rule of parallel transport takes as input not only
the edge point $a$ and $b$, but also a path $\varphi$ connecting
them. So, if a vector field $u$ is given at the point $a$, in addition
to a path from $a$ to $b$, the parallel transport of $u$ is uniquely
defined to the point $b$. Given another path the parallel transport is
unique as well, but the resulting transported vectors may well be different. The path-dependence of parallel transport is an
important and typical feature, which enables one to speak about the
curvature of the manifold. In fact the only situation in which all
parallel transport is independent of path is when there is no
curvature.  In spite of this, the infinitesimal limit in
\eqref{DirDer} is independent of the choice of the curve, so that it
may be used to define the so-called covariant derivative $\nabla_{w}u
$ of $u$ in the direction $w$, for any given vector $w\in TM_a$, since
the limit does not depend on how $w$ is extended to a vector field on
the whole manifold.  In addition, as we shall see, covariant
derivative and parallel transport can be extended to tensors.
Finally, observe that the vanishing of the covariant derivative on
some curve $t\rightarrow \varphi_t$ amounts to stating that the vector field $u$
behaves as if its values along the curve $ \varphi_t$ were arising  by
parallel transport to the whole curve of the value taken at a
particular point on the curve. Such a field along a curve is called an
autoparallel field. The covariant derivative thus measures the
deviation from being autoparallel.  We have seen that the
infinitesimal version of the parallel transport rule allows one to define a
differentiation of one vector field with respect to another one; this
differentiation process is called a linear connection and is noted
$\nabla$. In fact the role of the parallel transport and the covariant
derivative can be reversed.  Indeed, when it is technically feasible  to perform
the operation of covariant derivative, one can construct a
parallel transport rule, which is simply obtained by performing the
transport in such a way  that the covariant derivative vanishes. This is the usual
way of introducing the concept of linear connection on a manifold, which 
we now state formally.

To each vector field $w\in \mathcal{T}_0^1(M)$, one  associates an operator
$\nabla_w$, the covariant derivative along the field $w$, satisfying the following properties:
\begin{itemize}
\item[1.]It is a linear operator on the tensor algebra, which preserves the degree:
\begin{eqnarray*}
& & \nabla_w:  \mathcal{T}_p^q(M) \rightarrow \mathcal{T}_p^q(M),\\
  & & \nabla_w( \Uptheta_1
  + \lambda  \Uptheta_2)  = \nabla_w  \Uptheta_1 + \lambda\nabla_w \Uptheta_2, \quad \Uptheta_1,
  \   \Uptheta_2\in \mathcal{T}_p^q(M), \ \ \lambda\in\R.
\end{eqnarray*}
\item[2.] It is a derivative, i.e. it satisfies the Leibniz rule:
\[
\nabla_w( \Uptheta_1 \otimes  \Uptheta_2) =  \nabla_w \Uptheta_1 \otimes  \Uptheta_2
+   \Uptheta_1 \otimes \nabla_w  \Uptheta_2, \quad  \Uptheta_1\in \mathcal{T}_{p_1}^{q_1}(M), \
  \Uptheta_2\in \mathcal{T}_{p_2}^{q_2}(M).
 \]
\item[3.]
It is $\mathcal{F}$-linear with respect to $w$, i.e.
$
\nabla_{v+\lambda w} =\nabla_{v}+\lambda\nabla_{ w}. 
$
\item[4.]
$
\nabla_w f=w(f)=\mathsterling_wf, \quad \forall f\in \mathcal{F}(M)=\mathcal{T}_0^0(M).
$
\item[5.]
   $\nabla_w$ commutes with the operation of contracted multiplication.
\end{itemize}

Given $(U,\phi)$ an admissible local chart on a $d$-dimensional manifold $M$,
the natural basis for $\mathcal{T}_0^1(U)$ is $\{\partial/\partial a_i \}_{1\leq i\leq d}$,
while the natural basis for  $\mathcal{T}_1^0(U)$ is the dual basis $\{ da^i \}_{1\leq i\leq d}$.
The covariant derivative is uniquely specified by the coefficients of linear
connection $\Gamma_{jk}^i(a)$ with respect to the natural basis and  are functions defined by
\[
\nabla_i \partial_j =:  \Gamma_{ij}^k \partial_k , \quad
\nabla_i\, da^j =  -\Gamma_{ki}^j da^k, \quad \nabla_i:=\nabla_{\partial_i},
\]
with the notation $\partial_i \equiv \partial/\partial a^i$.
Let us note that if $ \Uptheta\in \mathcal{T}_p^q(M)$, then the covariant derivative (also called
the absolute differential) $\nabla  \Uptheta= \nabla_{\!\ell}  \Uptheta da^\ell= \nabla_{\!\ell}  \Uptheta \otimes da^\ell$ is
a tensor of type $(q,p+1)$. Therefore if $ \Uptheta\in \mathcal{T}_p^q(M)$, then $\nabla  \Uptheta \in \mathcal{T}_{p+1}^q(M)$
with
\[
\nabla  \Uptheta = \nabla_{\!\ell}  \Uptheta_{j_1\ldots j_p}^{i_1\ldots i_q} 
\frac{\partial}{\partial {a^{i_1}}} \otimes\ldots\otimes \frac{\partial}{\partial{a^{i_q}}}\otimes da^\ell
\otimes da^{j_1}\otimes \ldots \otimes da^{j_p},
\]
with
\[
\nabla_{\!\ell}  \Uptheta_{j_1\ldots j_p}^{i_1\ldots i_q}= \partial_\ell \Uptheta_{j_1\ldots j_p}^{i_1\ldots i_q}
\quad + \ \Gamma_{ik}^{i_1} \Uptheta_{j_1\ldots j_p}^{ki_2\ldots i_q}\ +\ \mbox{all upper indices} \quad -  
\ \Gamma_{ij_1}^{l} \Uptheta_{l j_2\ldots j_p}^{i_1\ldots i_q}\ -\ \mbox{all lower indices}. 
\]

Under a change of natural basis, resulting from a change of coordinates
$(a^i)_{1\leq i\leq d}\mapsto (\tilde{a}^i:=\tilde{a}^i(a))_{1\leq i\leq d}$,
the following transformation holds:
\[
\widetilde{\Gamma}_{jk}^i= \frac{\partial \tilde{a}^i}{\partial a^l}\frac{\partial a^m}{\partial \tilde{a}^j} 
\frac{\partial a^n}{\partial \tilde{a}^k}\Gamma_{mn}^l  +  \frac{\partial \tilde{a}^i}{\partial a^l}
\frac{\partial^2 a^l}{\partial \tilde{a}^j\partial \tilde{a}^k}.
\]
From this expression, we observe that the  coefficients of the linear
connection $\Gamma_{jk}^i$, called the Christoffel symbols of the second kind, are not tensors
since they do not satisfy the tensoriality criterion given by the change of coordinate formula
for the components of a tensor \eqref{DefOfPushForward}.
On a $\mathscr{C}^k$ $(k\geq 2)$ manifold a connection is said to be of class $\mathscr{C}^r$
if, in all charts of an atlas, the $\Gamma_{jk}^i$ are of class $\mathscr{C}^r$. If $r\leq k-2$
the definition is coherent and does not depend on the atlas. If $ \Uptheta$ is  of class $\mathscr{C}^{k-1}$
and the connection of class $\mathscr{C}^{k-2}$, then $\nabla  \Uptheta$ is of class $\mathscr{C}^{k-2}$.

Let $(M,\nabla)$ be a manifold endowed with a linear connection, $\nabla$ the corresponding
covariant derivative operator, $t\rightarrow \gamma_t$ a curve on $M$, and $V\in \mathcal{T}_0^1(M)$.
The absolute derivative of the field $V$ along $\gamma$ is defined as
\begin{equation}
  \label{ABDV}
\frac{DV(t)}{Dt}:= \nabla_{\dot{\gamma}}V.
\end{equation}
The vector field $V$ on $\gamma$ is called autoparallel if its absolute derivative along $\gamma$ vanishes, i.e.
if the right-hand side of \eqref{ABDV} vanishes. The straight lines that result from iteration
of the infinitesimal parallel transport of the velocity vector, i.e. the trajectories $a\mapsto\gamma_t(a)$, $a\in M$,
with zero acceleration ($\nabla_{\dot{\gamma}}\dot{\gamma}=0$), are called the affinely parametrised geodesics on $(M,\nabla)$.

A fundamental object associated to a manifold $(M,\nabla)$ with a linear connection is the
torsion operation $\mathfrak{t}$, defined by
\[
\begin{tabular}{rcll}
$\mathfrak{t}: $ & $\mathcal{T}_0^1(M)\times \mathcal{T}_0^1(M)$ & $\rightarrow$ & $\mathcal{T}_0^1(M)$\\
&$ (u,\, v)  $ & $\mapsto$ & $\mathfrak{t}(u,v)=\nabla_u\nabla_v -\nabla_v\nabla_u -[u,v]$.
\end{tabular}
\]
We observe that $\mathfrak{t}$ is antisymmetric since $\mathfrak{t}(u,v)=-\mathfrak{t}(v,u)$.
The torsion tensor field $\uptau\in\mathcal{T}_2^1(M)$ is defined by $\uptau(\alpha,u,v)=\alpha(\mathfrak{t}(u,v))$,
for all $u, \, v\in  \mathcal{T}_0^1(M)$ and $\alpha\in \Lambda^1(M)=\mathcal{T}_1^0(M)$. Using the natural basis one has
$[\partial/\partial a^j,\partial/\partial a^k]=0$, so that the components of $\uptau$ are given by
\[
\uptau_{jk}^i=\Gamma_{jk}^i-\Gamma_{kj}^i.
\]
On a Riemannian manifold there exists a unique linear connection such that $\uptau=0$
and $\nabla g =0$ (i.e. $\nabla_ig_{jk}= \nabla_ig^{jk}=0$). Such a connection is called
a Riemann-Levi-Civita (RLC) connection. The condition $\uptau=0$ means that the connection $\nabla$ is
torsion-free and thus that the  Christoffel symbols are symmetric. The condition
$\nabla g =0$, which is equivalent to stating that $\nabla$ is a metric connection, ensures
the preservation of length of vectors, which are generated by parallel transport.
For an RLC connection the Christoffel symbols can be expressed in terms
of the partial derivatives of the metric tensor $g$:
\[
\Gamma_{jk}^i=\frac12 g^{il}(\partial_j g_{kl} + \partial_k g_{jl}- \partial_l g_{jk}), \quad \Gamma_{jk}^i=\Gamma_{kj}^i.
\]

Let $(M,g,\nabla,\mu)$ a Riemannian manifold endowed with a RLC connection $\nabla$ and a volume form $\mu$.
If $v=v^i\partial_i$, then
\[
{\rm div}_\mu v=\nabla_i v^i=\frac{1}{\sqrt{g}}\partial_i(\sqrt{g}v^i).
\]
Commonly used differential operators such as the exterior derivative or the codifferential
can be expressed in terms of covariant derivatives \citep[][Sec. V.B.4, pp. 316; see also \citet{DeR84} Chapter V, $\S26$, pp. 106]{CDD77}.

A detailed  description of linear connections and parallel transport can be found
in \citet[][Sec. 15.2, pp. 372]{Fec06} and  \citet[][Sec. V.B.1, pp. 300]{CDD77}.
We refer the reader to \citet[][Sec. 15.3, pp. 382 and Chapter 15]{Fec06} and \citet[][Sec. V.B.2, pp. 308 and Chapter V]{CDD77}
for more details about RLC connections (e.g. curvature tensor).

\subsection{Incompressible or divergence-free vector fields}
\label{ssec:CICF:IVF}
Let $(M,g,\nabla,\mu)$ be a Riemannian manifold endowed with an RLC connection $\nabla$ and a volume form $\mu$.
We say that a vector field $v\in\mathcal{T}_0^1(M)$ is incompressible or divergence-free (with respect to $\mu$) 
if  $\mathrm{div}_\mu v=0$. A divergence-free time-dependent smooth vector field $v\in\mathcal{T}_0^1(M)$ 
is the infinitesimal generator of a one-parameter family of volume-preserving smooth maps 
$\varphi_t:M\rightarrow M$, which satisfy 
\[
\dot{\varphi_t}:=\frac{d\varphi_t}{dt}  =v(t,\varphi_t), \quad \varphi_0=e:={\rm Identity}.
\] 
Then $v$ is incompressible (i.e.  $\mathrm{div}_\mu v=0$)
if and only if the flow $\varphi_t:M\rightarrow M$ is volume preserving; that is
the local diffeomorphism $\varphi_t:U\rightarrow V$ is volume preserving with respect to
$\mu_{|_{U}}$ and $\mu_{|_{V}}$ for all $U\subset M$.
Let us introduce $J_\mu(\varphi_t)$, the Jacobian of the flow $\varphi_t$ with respect to the volume
form $\mu$, defined by
\begin{equation*}
\label{def:Jac}
J_\mu(\varphi_t) = {\varphi^\ast \mu}/{\mu}.
\end{equation*}
Then the time-evolution of  the Jacobian $J_\mu(\varphi_t)$ is given by the classical differential identity
\begin{equation}
\label{eqn:Jac}
\frac{d}{dt} J_\mu(\varphi_t)= J_\mu(\varphi_t) \nabla_i v^i \circ \varphi_t.
\end{equation}
From \eqref{eqn:Jac} we directly see that the volume-preserving
property
 of the flow,
in other words incompressibility,  $\varphi_t$, i.e $J_\mu(\varphi_t)=1$, is equivalent --- as 
in a flat space --- to the divergence-free condition for the vector field $v$, i.e $\nabla_i v^i=0$.
The differential identity \eqref{eqn:Jac} can be easily proved from  the Lie derivative theorem
(see  Sec.~\ref{sec:CLT}), which states that
\begin{equation}
\label{LieDTmu}
\frac{d}{dt} \varphi_t^\ast \mu = \varphi_t^\ast \left(\partial_t\mu + \mathsterling_{v} \mu\right),
\end{equation}
where $\mathsterling_{v} \mu $ is the Lie derivative of the volume form $\mu$ with
respect to the vector field $v$. From a geometric point of view, the Lie derivative of the form $\mu$ 
measures  the rate of change of volume of a parallelepiped spanned by $d$ vectors that are pushed forward 
by the flow $\varphi_t$ of $v$ (see Sec.~\ref{sec:CLT}).
Indeed, dividing \eqref{LieDTmu} by $\mu$, and using the properties
$
\partial_t\mu=0, 
$
and
$
\mathsterling_{v} \mu=:({\rm div}_\mu v)\mu= \nabla_i v^i \mu, 
$
we obtain
\begin{equation*}
\frac{d}{dt} J_\mu(\varphi_t)= \left[\varphi_t^\ast\left(\nabla_i v^i \mu\right)\right]/\mu=
\left[\left(\nabla_i v^i  \circ \varphi_t\right)\varphi_t^\ast\mu\right]/\mu=
 J_\mu(\varphi_t) \nabla_i v^i \circ \varphi_t.
\end{equation*}

\subsection{Integration of differential forms and the Stokes theorem}
\label{ssec:CICF:IDF}
The standard $p$-simplex in an oriented Euclidean space $\R^p$,
is the oriented convex closed set $\mathbb{S}_p=\{x\in\R^p\ | \ 0\leq x^i\leq 1, \ \sum_{i=1}^p x^i\leq 1 \}$.
The vertices, which generate $\mathbb{S}_p\subset \R^p$, are the $p+1$
points $V_0=(0,\ldots,0)$, $V_1=(1,0\ldots,0)$, $\ldots$, $V_p(0,\ldots,0,1)$.
We shall write  $\mathbb{S}_p=(V_0, \ldots, V_p)$.
Opposite to each vertex $V_k$ there is the $k$th face of $\mathbb{S}_p$, which
is not a standard Euclidean simplex, sitting as it does in $\R^p$ instead of $\R^{p-1}$.
We shall rather consider it as a singular simplex in $\R^p$. In order to do this we must
exhibit a specific map $f_{p-1}^{k}:\mathbb{S}_{p-1}\rightarrow \mathbb{S}_p$ given by
\[
f_{p-1}^0(y^1,\ldots,y^{p-1})=\left(
1-\sum_{i=1}^{p-1}y^i,y^1,\ldots,y^{p-1}
\right),  \ \mbox{and} \
f_{p-1}^k(y^1,\ldots,y^{p-1})=\left(
y^1,\ldots,y^{k-1},0,y^k,\ldots, y^{p-1} 
\right),
\]
if $k\neq 0$. A $\mathscr{C}^m$-singular $p$-simplex on a $\mathscr{C}^r$-manifolds $M$,
 $1\leq m\leq r$, is a  $\mathscr{C}^m$-map $\mathcal{S}_p:\mathbb{S}_{p}\rightarrow M$. The points
$\mathcal{S}_p(V_0), \ldots, \mathcal{S}_p(V_p)$ are the vertices of the singular $p$-simplex
$\mathcal{S}_p$ and the maps $\mathcal{S}_p\circ f_{p-1}^k:\mathbb{S}_{p-1}\rightarrow M$ are 
called the $k$th face of the singular  $p$-simplex $\mathcal{S}_p$. We emphasise that 
there is no restriction on the rank (dimension of the image in $M$) of the map $\mathcal{S}_p$;
for example the image of $\mathbb{S}_{p}$, which is also denoted by  $\mathcal{S}_{p}$ may be 
a single point in $M$. A ($\mathscr{C}^m$-singular) $p$-chain $c_p$ on $M$ is a finite linear combination
with real coefficients $\lambda_j\in\R$ of $\mathscr{C}^m$-singular $p$-simplexes $\{\mathcal{S}_{p,j}\}_{1\leq j\leq n}$; 
that is
$
c_p = \sum_{j=1}^n \lambda_j \mathcal{S}_{p,j}.
$
The boundary of a 
singular $p$-simplex  $\mathcal{S}_{p}$ is the $(p-1)$-chain $\partial \mathcal{S}_{p} $ defined by
\[
\partial  \mathcal{S}_{p}= \sum_{k=0}^p(-1)^k\mathcal{S}_{p}\circ f_{p-1}^k,
\]
and that of a singular $p$-chain is obtained by extending the operator
$\partial$ from simplexes  to chains by linearity. For example, in $\R^2$ the $2$-simplex is a triangle 
$\mathbb{S}_2=(V_0, V_1, V_2)$, and its boundary is the $1$-chain
$\partial \mathbb{S}_2= (V_1,V_2) - (V_0,V_2) + (V_0,V_1)$. 
Using the relation $ f_{p-1}^j\circ  f_{p-2}^j =  f_{p-2}^{i-1}\circ f_{p-2}^{i-1}$ for $j<i$, we can verify the
property
\[
\partial^2=\partial\circ \partial=0.
\]
The singular $p$-simplex $\mathcal{S}_p:\mathbb{S}_{p}\rightarrow M$ 
is the natural object over which one integrates $p$-forms of $M$ via
the pullback 
\[
\int_{\mathcal{S}_p} \alpha =\int_{\mathbb{S}_{p}} \mathcal{S}_p^\ast \alpha, \quad \alpha\in\Lambda^p(M).
\]
Integration of a $p$-form over a $p$-chain is easily obtained by linear extension. Finally, we give the Stokes theorem
on chains. If $c$ is any $p$-chain and $\alpha\in \Lambda^{p-1}(M)$, then
\[
\int_c d\alpha = \int_{\partial c} \alpha.
\] 
A detailed description of the Stokes theorem on chains can be found in
\citet[][Sec. 7.2C, pp. 495]{AMR88} and \citet[][Sec. 3.3, pp. 110 and Sec. 13.1, pp. 333]{Fra12}.

\subsection{From local to global geometry: Betti numbers and Hodge's generalisation of the Helmholtz decomposition}
\label{ssec:CICF:HCHD}

Throughout our study of hydrodynamics using a geometrical point of
view, we have encountered questions that depend on the global
topological structure of the space in which the flow takes place. One
frequently occurring example is the need to know under what conditions
a differential form  that is closed (i.e. has a vanishing exterior
derivative) is also exact (i.e. is the exterior derivative of
some other form). Another instance has do with the
generalisation of the well-known Helmholtz decomposition. The latter
states that in the full 3D space, any square integrable vector field can
be orthogonally decomposed into the sum of two vector fields, one being
a gradient and the other one a curl. In terms
of differential forms this amounts to decomposing a differential form
into the sum of an exact form and of a co-exact form. Actually, the correct
decomposition, called the Hodge decomposition, has sometimes a third
term, which is harmonic (of vanishing Laplacian).

The appropriate tool to address such gobal topological issues is known
as cohomology, a central subject in modern mathematics. Here we give
only a glimpse of some key results that matter for the geometrical
approach to fluid mechanics.  The emphasis will be on \textit{Betti
numbers} that give necessary and sufficient conditions for a closed
$p$-form to be exact.

Let $M$ (resp. $N$) be a differentiable manifold of
dimension $d$ (resp. $n$). Singular $p$-chains have been defined in Appendix~\ref{ssec:CICF:IDF}.
The collection of all singular $p$-chains of $M$ with coefficients in $\R$
forms an Abelian (commutative) group, the (singular) $p$-chain group of $M$ with coefficients in
$\R$, written $C_p(M; \R)$. The boundary operator $\partial$ defines the homomorphism
$\partial : C_p(M; \R) \rightarrow C_{p-1}(M; \R)$. Given a map
$\varphi: M \rightarrow N$ we have an induced homomorphism
$\varphi_\ast : C_p(M; \R) \rightarrow C_p(N; \R)$ and the boundary homomorphism $\partial$
is natural with respect to such maps, i.e.  $\partial \circ \varphi_\ast = \varphi_\ast \circ \partial$.
We define a (singular) $p$-cycle to be a $p$-chain $c_p$ whose boundary is 0. The collection
of all $p$-cycles,
\[
Z_p(M; \R):= \{c_p \in C_p\ |\ \partial c_p = 0\} = {\rm ker} \ \partial:C_p \rightarrow C_{p-1},
\]
that is, the kernel of the boundary homomorphism $\partial $, is a subgroup (the $p$-cycle group) of the chain group $C_p$.
We define a $p$-boundary $\beta_p$ to be a $p$-chain that is the boundary
of some $(p + 1)$-chain. The collection of all such chains
\[
B_p(M; \R) := \{\beta_p \in C_p\ | \beta_p = \partial c_{p+1},
\ \mbox{for some}\  c_{p+1} \in C_{p+1}\} = {\rm Im} \ \partial : C_{p+1} \rightarrow C_p,
\] 
the image or range of $\partial$, is a subgroup (the $p$-boundary group) of $C_p$. In addition,
$\partial \beta = \partial\partial c = 0$ implies that $B_p \subset Z_p\subset C_p$.
When considering closed forms, we observe that boundaries contribute nothing to integrals. Thus, when
integrating closed forms, we may identify two cycles if they differ by a boundary. Therefore
we say that two cycles $c_p$ and $c_p'$ in $Z_p(M; \R)$ are equivalent or homologous if they differ by a boundary,
that is, an element of the subgroup $B_p(M; \R)$ of $Z_p(M; \R)$.
The quotient group
\[
H_p(M; \R) := \frac{Z_p(M; \R)}{B_p(M; \R)},
\]
is called the $p$-th homology group.
When $B_p$ and $Z_p$ are infinite-dimensional, in many cases $H_p$ is
nevertheless finite-dimensional.
For example, this is the case when $M$ is a compact finite-dimensional manifold.
The dimension of the vector space $H_p$ is called the $p$-th Betti number, written $b_p = b_p(M)$ and defined by
\[
b_p(M) := {\rm dim }\  H_p(M;\R).
\]
In other words, $b_p$ is the maximum number of $p$-cycles on $M$, such that all real linear combinations 
with non-vanishing coefficients are never a boundary.
Since $\varphi_\ast$ commutes with the boundary homomorphism $\partial$, we know that
$\varphi_\ast$ takes cycles into cycles and boundaries into boundaries. Thus $\varphi_\ast$ sends
homology classes into homology classes, and we have an induced homomorphism
$\varphi_\ast : H_p(M; \R) \rightarrow H_p(N; \R)$. We now give some fundamental examples.
If $M$ is compact (path-)connected (any two points of $M$ can be connected by a piecewise smooth curves)
then $H_0(M,\R)=\R$ and $b_0(M)=1$. If $M$ is compact but not connected, i.e. it consists of $k$ connected
pieces then $H_0(M,\R)=\R^k$ and $b_0(M)=k$. If $M$ is a $d$-dimensional closed manifold
(compact manifold without boundary), then $H_p(M,\R)=0$ and $b_p(M)=0$, for $p>d$. If $M$ is compact and simply-connected
(i.e. path-connected and every path between two points can be continuously transformed, staying on $M$,
into any other such path while preserving the two endpoints in question; in other words $M$ is connected
and every loop in $M$ is contractible to a point) then  $H_1(M,\R)=0$ and $b_1(M)=0$. More examples can be found
in \citet[][Sec. 13.3, Chapter 13, pp. 347]{Fra12}.

We set $Z^p(M; \R)$ the subspace of $\Lambda^p(M)$ constituted of all closed $p$-forms, also called
$p$-cocyles. We set $B^p(M; \R)$ the subspace of $Z^p(M;\R)$ constituted of all exact $p$-forms, also called
$p$-coboundaries. Integration allows us to associate to each closed $p$-form on $M$ a linear fonctional
on $p$-cycles. This linear functional remains the same if we add to a closed $p$-form an exact $p$-form or
if we add to a $p$-cycle a $p$-boundary. Therefore this linear functional defines a linear transformation
from the quotient space $Z^p(M;\R)/B^p(M;\R)$ to $H_p^\ast(M;\R)$ that is the dual space of $H_p(M;\R)$. This
dual space is called the $p$-th cohomology vector space and is noted $H^p(M;\R)$. Moreover it can be shown
that this linear functional is an isomorphism: this is the celebrated
de Rham theorem \citep[see, e.g.,][Sec. 13.4, Chapter 13, pp. 355]{Fra12}.
Therefore we have
\[
H^p(M;\R):= H_p^\ast(M;\R) = \frac{Z^p(M;\R)}{B^p(M;\R)}.
\]
Two closed forms are equivalent or cohomologous if they differ by an exact form.
As a consequence a closed $p$-form is exact if and only if its integral on any $p$-cycles vanishes
or if it is cohomologous to zero.
Since a finite-dimensional vector space has the same dimension as its dual space, we have
${\rm dim}\ H^p(M;\R)=b_p(M)$ for $M$ compact, where $b_p(M)$ is the $p$-th Betti number.
Thus $b_p(M)$ is also the maximum number of closed $p$-forms on $M$, such that all linear
combinations with non-vanishing coefficients are not exact.
The knowledge of the Betti numbers of a given manifold
$M$ for $p\geq 1$ yields an exact quantitative answer to the question about exactness of a closed $p$-form:
\[
\mbox{a closed $p$-form is exact if and only if } b_p(M)=0. 
\]
From the Poincar\'e lemma \citep[see, e.g.,][Lemma 6.4.18]{AMR88},
if $M$ is a compact $d$-dimensional contractible manifold  (see Appendix~\ref{ssec:CICF:Ma} for
the definition), all the Betti numbers  $p\geq 1$ vanish, i.e.  $b_1(M)=\ldots = b_d(M) =0$,  and $b_0(M) = 1$.
Contractibility is, however, an excessivily strong constraint to
ensure the equivalence of closeness and exactness. For $p$-forms of a
given degree $p$,  the vanishing of just the   Betti
number, $b_p(M)$  is actually sufficient. Let us remark that from the duality between the
finite-dimensional vector spaces $H^p(M;\R)$ and $H_p^\ast(M;\R)$ exactness of $p$-form can be
determined from the topological properties of $M$.

The Laplace-de Rham operator $\Delta_{\rm H}: \Lambda^p(M) \rightarrow
\Lambda^p(M)$ is defined by $\Delta_{\rm H} := d d^\star + d^\star d=(d
+d^\star)(d + d^\star)$. A form $\alpha$ for which $\Delta_{\rm H} a = 0$ is
called harmonic.  Let $\mathcal{H}^p(M) :=\{ \alpha \in
\Lambda^p(M)\ | \ \Delta_{\rm H} \alpha=0\}$ denote the vector space of
harmonic $p$-forms. If $M$ is a closed Riemannian manifold (i.e.  a
compact boundaryless oriented Riemannian manifold) and $ \alpha \in
\Lambda^p(M)$, then $\Delta_{\rm H} \alpha = 0$ if and only if $d\alpha = 0$
and $d^\star \alpha = 0$. If $M$ is a compact Riemannian manifold with
boundary, the condition that $d\alpha =0$ and $d^\star \alpha=0$ is now
stronger than $\Delta_{\rm H} \alpha=0$. Thus the vector space of harmonic
$p$-form is defined by $\mathcal{H}^p(M) =\{ \alpha \in
\Lambda^p(M)\ | \ d \alpha= d^\star\alpha=0\}$.  The Hodge theorem
\citep[see, e.g.,][Theorem 14.28, Chapter 14, pp. 371; see also
  \citet{DeR84}, Theorem 22, Chapter V, $\S 1$, pp. 131]{Fra12} states
that if $M$ is a closed Riemannian manifold, then the vector space of
harmonic $p$-form is finite dimensional and the Poisson equation
$\Delta_{\rm H} \alpha = \rho$ has a solution if and only if $\rho$ is
orthogonal to $\mathcal{H}^p(M)$, that is $\langle h, \rho
\rangle_g=0$, for all $h\in \mathcal{H}^p(M)$ and where (see
Appendix~\ref{ssec:CICF:HSECD})
\[
\langle \alpha, \beta \rangle_g:=\int_M \alpha \wedge \star \beta= \int_M (\!( \alpha, \beta )\!)_g \mu, \quad \alpha,\ \beta \in
 \Lambda^p(M).
\]
If $h_1,h_2,\ldots, h_q$ is an orthonormal basis of $\mathcal{H}^p(M)$
and $\beta \in \Lambda^p(M)$ then $\beta-h:=\beta -\sum_i \langle
\beta, h_j \rangle_g h_j$ is orthogonal to $\mathcal{H}^p(M)$ and so,
by Hodge's theorem we can solve the equation $\Delta_{\rm H} \alpha =\beta -h$
for $\alpha\in \Lambda^p(M)$.  In other words, for any $\beta \in
\Lambda^p(M)$ on a closed Riemannian manifold $M$ we can write
\[
\beta=d(d^\star \alpha) + d^\star(d\alpha) + h.
\]
Thus any $p$-form $\beta$ on a closed Riemannian manifold $M$ can be
written as the sum of an exact $p$-form $d(d^\star \alpha)$, a
co-exact $p$-form $d^\star(d\alpha)$ and   a harmonic $p$-form
$h$. Hence, we obtain the Hodge decomposition
\[
\Lambda^p(M) = d\Lambda^{p-1}(M) \oplus d^\star\Lambda^{p+1}(M) \oplus \mathcal{H}^p(M),
\]
where the three subspaces are mutually orthogonal. As already observed
the Hodge decomposition generalises and extends 
the Helmholtz decomposition, for which the harmonic term is absent
(because
in $\R^d$,  the 1-cohomology $H^1=0$). In particular, from the Hodge decomposition, if
$\beta \in \Lambda^{p-1}(M)$ is closed on a closed manifold $M$, then
$\beta = d \alpha + h$ where $\alpha \in \Lambda^{p-1}(M)$ and  $h\in\mathcal{H}^p(M)$.
Thus in each $p$-cohomology vector space there is a unique harmonic representative, or in other words
the spaces $\mathcal{H}^p(M)$ and ${H}^p(M;\R)$ are isomorphic:
\[
{H}^p(M;\R)\equiv \mathcal{H}^p(M).
\]
The Hodge theorem and decomposition
have been extended to non compact spaces by \citet[][see Chapter V, $\S
  32$, pp. 136]{DeR84} and to
a compact Riemannian  manifold with boundary \citep[see,
  e.g.,][Sec. 7.5, pp. 541; see also \citet{Fra12}, Sec. 14.3, pp. 375
  and references therein]{AMR88}. In the latter case, the space of
closed (resp. exact) $p$-forms must be replaced
by the space of normal 
$p$-forms that are  closed (resp. exact). Furthermore, the space of
co-closed (resp. co-exact)  $p$-forms must be  replaced   by the space
of  co-closed (resp. co-exact) tangent $p$-forms \citep{Sch95}. Here, ``normal'' means  with
vanishing tangential components and ``tangent''  with vanishing normal
components.

Finally we recall the
Bochner theorem \citep[see, e.g.,][Theorem 14.33, Sec. 14.2, pp. 374]{Fra12}, which states
that if a closed Riemannian manifold $M$ has positive Ricci curvature,
then a harmonic $1$-form must vanish identically, and thus $M$ has first
Betti number $b_1 = 0$ and $1$-cohomology $H^1(M,\R)=0$.

\end{document}